\documentclass[]{cas-sc}
\usepackage[backend=bibtex,sorting=none,maxcitenames=1000,maxbibnames=99,style=numeric-comp]{biblatex}
\graphicspath{{./figures/},{./results/},{./results/2D_results/final_plots/}}

\usepackage{graphicx}
\usepackage{amsmath,amsthm, amssymb, float, bm, subcaption}
\usepackage{placeins}
\usepackage{longtable,tabularx,booktabs}
\usepackage{overpic}
\usepackage{cancel}
\usepackage{fancyhdr}

\usepackage{setspace}
\usepackage{soul}
\usepackage[most]{tcolorbox}
\definecolor{C0}{HTML}{1F77B4}
\definecolor{C1}{HTML}{FF7F0E}
\definecolor{C2}{HTML}{2ca02c}
\definecolor{C3}{HTML}{d62728}
\definecolor{C4}{HTML}{9467bd}
\definecolor{C5}{HTML}{8c564b}

\ifdefined\usetodonotes
\usetikzlibrary{calc}
\usepackage[textsize=scriptsize]{todonotes}
\newcommand{\tododone}[1]{\todo[disable]{#1}\addcontentsline{tdo}{todo}{\st{#1}}}

\fancypagestyle{mypage}{%
    \fancyhf{}
    \fancyhead[LO]{%
        \begin{tikzpicture}[overlay,remember picture]
            \fill [color=black!25] (current page.north east) rectangle
                ($ (current page.south east) - (3.25cm,0cm) $);
        \end{tikzpicture}
        }
    \fancyhead[RE]{%
        \begin{tikzpicture}[overlay,remember picture]
            \fill [color=orange](current page.north east) rectangle
                ($ (current page.south east) + (-1cm,0cm) $);
        \end{tikzpicture}
        }
    \fancyfoot[C]{\thepage}

    }
\else

\newcommand{\tododone}[1]{}
\fi

\ifdefined\usechanges
\usepackage{fontawesome}
\usepackage[commentmarkup=todo,authormarkup=none]{changes} 
\definechangesauthor[color=C0]{R1} 
\definechangesauthor[color=C1]{R5} 
\definechangesauthor[color=C2]{RA} 
\makeatletter
\@namedef{Changes@AuthorColor}{C2}
\colorlet{Changes@Color}{C2}
\renewcommand{\Changes@Markup@comment}[3]{%
  \IfStrEq{\Changes@optioncommentmarkup}{todo}%
		{\colorlet{Changes@todocolor}{authorcolor}\todo[color=Changes@todocolor!10, bordercolor=Changes@todocolor, linecolor=Changes@todocolor!70, nolist]{\textbf #1}}{}}
\makeatother
\fi

\makeatletter
\def\ps@pprintTitle{%
  \let\@oddhead\@empty
  \let\@evenhead\@empty
  \def\@oddfoot{\reset@font\hfil\thepage\hfil}
  \let\@evenfoot\@oddfoot
}
\makeatother

\usepackage[noabbrev,capitalise]{cleveref}

\newtheorem{theorem}{Theorem}

\makeatletter                                                                   
\newlength{\bibsep}{\@listi \global\bibsep\itemsep \global\advance\bibsep by\parsep} 
\makeatother

\begin{document}

\shorttitle{}
\shortauthors{Boyd, Sandall, Meier, Quinlan, Runnels}
\title[mode=title]{A diffuse boundary method for phase boundaries in viscous compressible flow}


\author[uccs]{Emma M. Boyd}
\author[isu]{Eric Sandall}
\author[uccs]{Maycon Meier}
\author[uccs]{J. Matt Quinlan}
\author[isu,uccs]{Brandon Runnels}[orcid=0000-0003-3043-5227]
\cormark[1]
\cortext[1]{Corresponding author}
\address[isu]{Department of Aerospace Engineering, Iowa State University, Ames, IA,  USA}
\address[uccs]{Department of Mechanical and Aerospace Engineering, University of Colorado, Colorado Springs, CO, USA}

\begin{abstract}
  Many physical systems of interest involve the close interaction of a flow in a domain with complex, time-varying boundaries.
  Treatment of boundaries of this nature is cumbersome due to the difficulty in explicity tracking boundaries that may exhibit topological transitions and high curvature.
  Such conditions can also lead to numerical instability.
  Diffuse boundary methods such as the phase field method are an attractive way to describe systems with complex boundaries, but coupling such methods to hydrodynamic flow solvers is nontrivial.
  This work presents a systematic approach for coupling flow to arbitrary implicitly-defined diffuse domains.
  It is demonstrated that all boundary conditions of interest can be expressed as suitable fluxes, noting that angular momentum flux is necessary in order to account for cases such as the no-slip condition.
  Moreover, it is shown that the diffuse boundary formulation converges exactly to the sharp interface solution, resulting in a well-defined error bound.
  The method is applied in a viscous compressible flow solver with block-structured adaptive mesh refinement, and the convergence properties are shown.
  Finally, the efficacy of the method is demonstrated by coupling to other classical flow problems (vortex shedding), problems in solidification (coupling to dendritic growth), and flow through eroding media (coupling to the Allen Cahn equation).
\end{abstract}

\maketitle

\ifdefined\usetodonotes\thispagestyle{mypage}\fi

\section{Introduction}


 
Complex solid-fluid multiphase modeling is essential to modern science and technology, from dendrite growth in batteries\cite{chen2021phase} to fluid flow through microporous media for electronics cooling\cite{arshad2025heat} to solid rocket motor propellant combustion\cite{meier2024diffuse}.
Each of these applications involves complex phase boundary geometries. 
In particular, topological changes such as particle agglomeration, complete erosion of solids in a flow, and spontaneous phase transformations due to uneven heating pose significant numerical challenges.
Many important applications involve extreme flow regimes, including turbulent, compressible, and creeping flows.
These complex flow conditions give rise to various solid-fluid interactions, where viscous effects, heat transfer, and mass transfer may all play critical roles.
In this work, a diffuse interface numerical method is developed that is adept at handling topological changes, multiscale behavior, and time-varying solid-fluid boundaries. 
Our aim is to develop a versatile method that can be easily adapted for various applications.

A wide range of methods have been proposed to treat internal boundaries in multiphase applications.
In cases where the interface evolution is slow, a viable approach is to conform the mesh to the boundary itself, for instance, using the method of boundary-fitted coordinates (BFC) \cite{thompson1982boundary,bijl1998unified}.
Several methods have been used to model boundary-driven flows with extreme time-varying morphology.
The marker-and-cell (MAC) method is one way to account for complex evolving boundaries of incompressible fluids in an Eulerian frame \cite{welch1965mac,tome1994gensmac,mckee2004recent,mckee2008mac}.
Alternatively, Lagrangian approaches associate the mesh with the material, which can provide high accuracy boundary evolution~\cite{quan2007moving}; these methods, however, are prone to ``mesh tangling'' or negative cell volumes near complex curvatures.
Similarly, front-tracking methods require explicitly calculating the location of the interface \cite{chertock2008interface, sato2013sharp} and may become intractable when topological changes of the boundary occur~\cite{maltsev2022high}.
One example is the immersed boundary method (IBM), which is a diffuse method that requires explicit interface tracking \cite{peskin1972flow,sotiropoulos2014immersed,griffith2020immersed}.
Another example is the volume-of-fluid method (VOF).
In VOF, the volume fraction of each phase is tracked within each computational cell, allowing a relatively simple approach to reconstruct the phase boundaries. 
By solving advection equations for the volume fractions, VOF can accurately capture large deformations of the interface, such as splashing, coalescence, and breakup\cite{pilliod2004second}. 
VOF is used, for example, in free-surface flows, bubble dynamics\cite{yujie2012three}, and droplet impact\cite{bussmann1999three}.
VOF is, however, unsuitable for multiphase problems involving voids and may become computationally untenable for non-convex geometries \cite{maric2020unstructured,niethammer2019extended,fuster2018momentum}.
Given these challenges in multiphase problems with complex topology, an approach that avoids explicitly tracking or reconstructing the interface is appealing\cite{mirjalili2019comparison}.

Diffuse interface methods have been used extensively in fluid applications~\cite{anderson1998diffuse,saurel2018diffuse}.
These methods rely on a phase indicator field defined over the entire problem domain, allowing implicit interface tracking as the indicator field evolves according to the boundary interactions between the phases.
This framework provides a closure model describing the phase boundary, as in the five-equation model of~\textcite{allaire2002five} or the recent work of \textcite{goulding2025conservative}.
This approach also shares many similarities with the phase field method of solid mechanics, which has been directly applied to the Navier-Stokes equations in the work of \textcite{huang2020consistent}.
The evolution equation used to evolve the interface may cause the interface to diffuse over time, and a correction term may be necessary to retain the interface geometry\cite{jain2020conservative}.
In other similar models~\cite{abgrall1996prevent, sainsaulieu1995finite, johnsen2012preventing,mirjalili2024conservative}, the order parameter is subject to an advection-diffusion equation coupled to the governing equations of the system. 
These approaches perform well in immiscible multiphase flows, where the indicator field may represent a mass or volume fraction \cite{fakhari2010phase,soligo2019mass,wang2019brief,nathan2025accurate,abels2024mathematical}.

For a more complicated system such as deflagration to detonation transition, further governing equations are required to capture mass and energy transfer between the phases, as seen in the seven equation model of~\textcite{baer1986two}. 
This seven-equation model has been used in several combustion applications, including hydrogen\slash oxygen combustion \cite{murrone2019eulerian}, premixed turbulent combustion \cite{bohbot2010multi}, and burning droplets \cite{ruggirello2011reaction}.
The model extends the classical Euler compressible flow system by introducing additional equations to account for the conservation of mass, momentum, and energy for each phase, along with equations governing the mass and energy transfer between the phases due to the phase change. 
The seven-equation model simultaneously solves conservation statements in the interface with the governing equations outside of the interface in order to evolve the phase boundary.
For complex problems, the number of equations required to evolve the boundary becomes prohibitively large.
Some other diffuse interface approaches in combustion include artificially thickened flame models for premixed combustion, which determine a diffuse length scale based on the criteria that the deflagration speed be consistent with the free boundary problem~\cite{legier2000dynamically, poinsot2005theoretical}.
These applications include multiple mechanisms for phase change, and 
using a single surrogate field to describe the phase boundary does not easily lead to a single physical interpretation.
Again, the number of governing equations required to evolve such a field from first principles would be prohibitively large.
The current work takes a numerical approach, which, while consistent with the models described above, allows for an arbitrarily evolving order parameter.

This phase boundary evolution approach has been applied to solid-fluid interactions in the absence of mass transfer~\cite{wallis2021diffuse,maltsev2022high, kemm2020simple}.
Notably, the works of Jain \cite{jain2021assessment,jain2020diffuseinterface}, Ghaisas\cite{ghaisas2018unified}, Subramaniam\cite{subramaniam2018high}, and Adler\cite{adler2020diffuseinterface} are examples of diffuse interface models for elastic and plastic deformation of solids in compressible flow with immiscible phases.
This bears some similarity to the wall potential construction used in some cases to restrict fluid motion to the outside of the solid~\cite{favrie2009solid,jacqmin2000contact}.
Localized artificial diffusivity (LAD) methods, which share many characteristics with diffuse interface methods, have been applied to combustion and compressible flows~\cite{lee2017localized,jainstable}. 
Crucially, these methods fail to guarantee that the added diffuse length scale will leave the flow far from the phase boundary unchanged.
Additionally, it is necessary to develop some mechanism that accounts for complicated solid-fluid boundary conditions at the interface without prohibitively increasing the size of the governing system of equations.

This work provides a systematic approach for determining diffuse boundary source terms that replicate the effect of an implicitly defined diffuse boundary on the flow field, with guaranteed convergence properties.
This computational method is thus capable of modeling solid-fluid phase boundaries on a single computational mesh, eliminating the need for external coupling of different simulation techniques.
The simultaneous evolution of all fields tightly couples both phases and avoids explicit tracking or reconstruction of the solid-fluid phase boundary to mitigate the computational burden.
These features are crucial for efficient simulations of solid composite propellant combustion, as well as many other simulation domains that feature nontrivial interactions between the solid and fluid phases.

This paper is structured in the following way.
First, a generalized form of the diffuse boundary source terms and governing equations are developed, which are then specialized to a set of example cases (\cref{sec:theory}).
Next, some of the practical considerations necessary in the implementation of the diffuse boundary theory are discussed (\cref{sec:numerical_considerations}), followed by a brief discussion of the computational framework and methods used in this work (\cref{sec:implementation}).
Finally, a selection of representative examples is presented that demonstrate the accuracy and convergence properties of the method, along with a variety of salient problems that demonstrate the method's versatility (\cref{sec:examples}).

\section{Theory}\label{sec:theory}

This section builds on and substantially generalizes the diffuse boundary method proposed by the authors in \textcite{schmidt2022self}.
The previous work established the ability of diffuse source terms to capture salient boundary conditions, written in the form of fluxes, for 1D and 2D inviscid flow; however, it did not treat viscous flow nor did it account for the important and highly prevalent no-slip boundary condition.
A key point in this work is to extend the ``flux'' interpretation of diffuse boundary source terms to the no-slip condition for viscous flow.

\subsection{Domain definition with implicit, diffuse boundaries}

In a discrete formulation, a solution is defined in a region $\Omega\in \mathbb{R}^3$, which is at least partially bounded by $\partial\Omega$ along which boundary conditions are prescribed.
The diffuse boundary formulation admits solutions on the entirety of $\mathbb{R}^3$ by replacing the domain $\Omega$ with a representative function $\eta$, known as the order parameter.
The order parameter $\eta$ is a time-varying, Lipschitz continuous function over $\mathbb{R}^3$ with range $[0,1]$.
The key feature of the order parameter is that in the sharp interface limit, (that is the limit as $\epsilon\rightarrow0$), the support of $\eta$ is exactly the discrete-boundary domain $\Omega$. 
In the case that $\epsilon>0$, the support of $\eta = 1$ is referred to as $\Omega_\epsilon$, and the support of $\nabla\eta$ is referred to as $\partial\Omega_\epsilon$---this defines the diffuse boundary region.
The gradient, $\nabla\eta$, is required to be 0 outside of $\Omega_\epsilon\cup\partial\Omega_\epsilon$, and there is assumed some parameterization of $\eta$, $\hat\eta(s)$, such that
\begin{align}
  \eta(\bm{y}+s\bm{n}) = \hat\eta(s) \quad \forall\bm{y}\in\partial\Omega, \quad s\in[-\epsilon/2,\epsilon/2],
\end{align}
where $\bm{n}$ is the inward-facing normal at $\bm{y}$, $\hat\eta(-\epsilon/2)=0$, $\hat\eta(\epsilon/2)=1$, and $d\hat\eta/ds$ is a mollifier over $[-\epsilon/2,\epsilon/2]$.
Thus, $\nabla\eta(\bm{y}+s\bm{n}) = \bm{n}\,d\hat\eta/ds$ and $\nabla\eta=\bm{0}$ outside of $\partial_\epsilon\Omega$.
These definitions lead to the following theorem, which forms the basis for the work presented here:

\begin{theorem}\label{thm:boundary_integral}
  Let $\eta$ be an idealized order parameter with length scale $\epsilon$, and let $f$ and $g$ be either scalar or
vector-valued bounded functions, with $\bm{n} \cdot \nabla g$ bounded in $\partial\Omega_\epsilon$.
  \begin{align}
    \lim_{\epsilon\to0}\int_A\int_{-\epsilon/2}^{\epsilon/2}\Big[f\eta + g|\nabla\eta|\Big]ds\,dA = \int_A\,g\,dA \ \ \ \ \  \forall A \subset \partial \Omega.
  \end{align}
\end{theorem}
\begin{proof}
  The proof of this theorem may be found in \textcite{schmidt2022self}. 
\end{proof}
In practice, it is convenient to select $\eta$ from a family of functions that, in the sharp interface limit as $\epsilon\rightarrow 0$, approach a Heaviside function. Consequently,  $\nabla\eta$ limits to the Dirac delta function with respect to the interface normal parameter, $\nabla^2\eta$ limits to a unit doublet, and so on.
Prior to this limit, however, it is often the case that the condition of compact support ($\nabla\eta=0$) in the diffuse boundary region is left unenforced.
This is particularly true in the cases in which the diffuse boundary is governed by a separate set of equations, as is the case when coupling to a phse field method.
(This may be dealt with, as discussed in subsequent sections, by using a numerical cutoff to prevent ``leakage'' of the flow solution into the solid domain.)
In the development of our theory, however, the condition of compact support is always assumed in order to avoid unnecessary complexity.

The order parameter $\eta$ describes the smooth distribution of a phase in $\mathbb{R}^3$. In the case of two phases, $1 - \eta$ is the order parameter of the second phase and a field $\varphi$ is defined for both phases over all of $\mathbb{R}^3$ by the mixture rule
\begin{align}
  \varphi = \varphi^{\mathrm{Phase 1}} \eta + \varphi^{\mathrm{Phase 2}} (1 - \eta).  \label{eq:mix}
\end{align}
For more than two phases, additional order parameters are introduced to the mixture rule such that the supports of the order parameters for an open cover of $\mathbb{R}^3$.
In general, the fields $\varphi^{\mathrm{Phase 1}}$ and $\varphi^{\mathrm{Phase 2}}$ are solutions to different sets of governing equations.
Of interest is a set of governing equations that hold for all of $\mathbb{R}^3$ and whose solution is $\varphi^{\mathrm{Phase 1}}\eta$. 
These governing equations must, in the sharp interface limit, reduce to the governing equations of $\varphi^{\mathrm{Phase 1}}$ over $\Omega \subset \mathbb{R}^3$, and boundary conditions applied over $\partial\Omega \subset \mathbb{R}^3$.  
A set of diffuse interface equations is valid only if it obeys this limit.

\subsection{Diffuse formulation of the Navier-Stokes equations}
\label{sec:equations}

The setting for this work is the flow of a compressible, viscous fluid.
This flow is given by the solution to the Navier-stokes equations for the conservation of mass, momentum, and energy:
\begin{subequations}\label{ns}
  \begin{align}
    \label{eq:mass}
    \mbox{Mass conservation:}&&\frac{\partial \rho}{\partial t} + \nabla\cdot (\rho \bm{u})  &= 0  & \bm{x} &\in \Omega, \\
    \label{eq:mom} \mbox{Linear momentum conservation:} &&
    \frac{\partial}{\partial t}(\rho\bm{u}) + \nabla\cdot(\rho\bm{u}\otimes\bm{u} - \mathbf{T} + p\,\mathbf{I}) &= \bm{0} & \bm{x} &\in \Omega, \\
    \label{eq:angmom}\mbox{Angular momentum conservation:} &&
    \big(\mathbf{T} - \bm{\sigma}(\nabla\bm{u})\big) &= \bm{0} & \bm{x} &\in \Omega \\
    \mbox{Energy conservation:}&&\frac{\partial \mathrm{E}}{\partial t} + \nabla\cdot(\mathrm{E} \bm{u} - \mathbf{T}\bm{u}) + \dot{q}&= 0 & \bm{x} &\in \Omega \label{eq:eng} \\
    \intertext{
    These equations govern  the flow within the interior of the domain $\Omega$.
    For an explicitly defend boundary $\partial\Omega$, the behavior of the solution fields across the boundary are determined by the flux equations:
    }
    \mbox{Mass transport:}&& \rho\bm{u} \cdot\bm{n}&= \dot{m}  &  \bm{x} &\in \partial\Omega \label{eq:mass_flux}\\
    \mbox{Linear Momentum transport:}&& \big(\rho\bm{u}\otimes\bm{u} - \mathbf{T} + p\,\mathbf{I}\big)\bm{n} &= \dot{\mathbf{P}} &  \bm{x} &\in \partial\Omega \label{eq:mom_flux}\\
    \mbox{Angular momentum transport:} && \mathbf{T}\bm{n} &= \dot{\mathbf{L}}  &  \bm{x} &\in \partial\Omega \label{eq:ang_flux}\\
    \mbox{Energy transport:}&& \big(E\bm{u}-\bm{T u}\big)\cdot\bm{n} &= \dot{e}  &  \bm{x} &\in \partial\Omega \label{eq:eng_flux},
  \end{align}
\end{subequations}
where $\rho$ is the density, $\bm{u}$ is the fluid velocity, $p$ is pressure (adopting the convention that positive pressure corresponds to compression), and $\operatorname{dim}$ is the spatial dimension.
The quantities $\dot{m}$, $\dot{e}$, $\dot{\mathbf{P}}$, and $\dot{\mathbf{L}}$ are fluxes for mass, energy, linear momentum, and angular momentum.
$E$ is internal energy per volume, which is defined to be comprised of kinetic and internal energy components: 
\begin{align}
  E = \mathrm{K} + \mathrm{U} = \frac{1}{2}\rho\,|\bm{u}|^2 + \mathrm{U}(p,T).
\end{align}
The tensors $\mathbf{T},\bm{\sigma}$ both correspond to the Cauchy stress tensor and are identical everywhere in the discrete boundary case.
It is useful in the diffuse boundary case, however, to distinguish between the actual flow stress, $\mathbf{T}$, and the Cauchy stress as determined by the flow's constitutive response that is governed by $\bm{\sigma}$.
The system is thus closed by the selection of the equation of state ($\mathrm{U}(p,T)$) and the viscous response $\sigma(p,\nabla\bm{u})$.
This work leaves the equation of state general; for the viscous response, a Newtonian fluid is assumed:
\begin{align}
  \bm{\sigma}(\nabla\bm{u}) = \mathbb{M}\nabla\bm{u}
\end{align}
where the tensor $\mathbb{M}$ is a fourth order tensor possessing both major and minor symmetry.
For an isotropic flow, $\mathbb{M}$ is defined in terms of Lam\'e parameters $\mu,\lambda$ as
\begin{align}
  \mathbb{M}\nabla\bm{u} = 2\mu\operatorname{sym}(\nabla\bm{u}) + \lambda\,(\nabla\cdot{\bm{u}})\,\mathbf{I}.
\end{align}

It is taken implicitly in the above formulation that any boundary condition can be achieved through the proper prescription of the flux variables $\dot{m}$, $\dot{\mathbf{P}}$, $\dot{\mathbf{L}}$, and $\dot{e}$.
If one prescribes boundary conditions in terms of these fluxes, then the formulation is complete.
At the same time, it is usually inconvenient (or impossible) to define realistic boundaries in this way, as this would involve the mixing of both essential and natural boundary conditions together.
Therefore, each flux variable is considered composed of a mixture of prescribed quantities and derived quantities, which then interact with the internal governing equations to produce the desired values.
In this way, the above generalized flux formulation can be adapted to any realistic set of boundary conditions by prescribing the appropriate quantities and leaving the others as solution fields.
For natural boundaries, the appropriate flux terms themselves are prescribed in terms of applied forces, such as pressure, traction at the boundary, or heat flux. 
On the other hand, for essential boundary conditions, the primitive constituents of the flux terms (such as density, velocity, or temperature) are prescribed, with all other terms left as unknowns.
Note that setting an essential boundary specifies the values of certain degrees of freedom at the boundary, fixing them to a particular value. 

Also of note in this formulation is the explicit treatment of the conservation of angular momentum.
It is common practice to reduce the angular momentum to a simple symmetry restriction $\mathbf{T}=\mathbf{T}^T$; one can verify that this substitution over an explicitly defined domain reduces the Navier-Stokes equations to their more traditional form.
In the context of a diffuse boundary, however, it will be shown that this reduction of angular momentum is no longer possible, and that certain boundary conditions manifest as angular momentum fluxes (and, hence, breakers of the symmetry of the stress tensor).
This will be discussed in more detail subsequently.
It is now possible to construct the following diffuse interface replacement of \cref{eq:mass,eq:mom,eq:angmom,eq:eng,eq:mass_flux,eq:mom_flux,eq:ang_flux,eq:eng_flux}, which are considered to be valid (but inactive) over $\mathbb{R}^{\operatorname{dim}}$:
  \begin{subequations}\label{diffuse}\begin{align}\label{eq:diffuse_mass}
      \mbox{Mass conservation:} &&
      \eta\frac{\partial \rho }{\partial t} + \nabla\cdot(\eta\rho\bm{u}) &= \dot{m}|\nabla\eta| && \forall \bm{x}\in\mathbb{R}^{\operatorname{dim}},\\ 
      \label{eq:diffuse_mom}
      \mbox{Linear momentum conservation:} &&
      \eta\frac{\partial (\rho u)}{\partial t} + \nabla\cdot(\eta(\rho\bm{u}\otimes\bm{u} - \mathbf{T} + p\,\mathbf{I})) &= \dot{\mathbf{P}}|\nabla\eta| && \forall \bm{x}\in\mathbb{R}^{\operatorname{dim}},\\
      \label{eq:diffuse_ang}
      \mbox{Angular momentum conservation:} &&
      \eta \mathbf{T} - \mathbb{M}(\nabla\eta\bm{u}) &= \dot{\bm{\mathbf{L}}} |\nabla \eta| && \forall \bm{x}\in\mathbb{R}^{\operatorname{dim}},\\
      \mbox{Energy conservation:} &&
      \eta\frac{\partial (E)}{\partial t} + \nabla\cdot(\eta (E\bm{u} - \mathbf{T}\bm{u})) + \eta\dot{q}&= \dot{e}|\nabla\eta| && \forall \bm{x}\in\mathbb{R}^{\operatorname{dim}} \label{eq:diffuse_eng},
    \end{align}\end{subequations}
where $\eta_{\epsilon}$ is an idealized order parameter (as defined above) corresponding to the explicit domain $\Omega$,  characterized by diffuse boundary thickness $\epsilon$.
In the above, the right-hand side is nearly identical to the corresponding conservation equations except for the presence of $\eta$ in each term.
The right hand side depends on the replacement of each boundary-valued source term ($\dot{m}:\partial\Omega\to\mathbb{R}^{\operatorname{dim}}$, etc.,) by equivalent domain-valued source term fields ($\dot{m}:\Omega\to\mathbb{R}^{\operatorname{dim}}$, etc.).
Because the value of the domain-value source term fields only matters within the diffuse boundary itself, their values are considered to be determined by normal extrusion from the sharp boundary.
As long as $\epsilon$ is sufficiently small as to maintain validity (i.e., smaller than the smallest curvature radius of interest), such an extrusion should produce negligible variation in the direction normal to the sharp boundary. 

The validity of \cref{eq:diffuse_mass,eq:diffuse_mom,eq:diffuse_ang,eq:diffuse_eng} is established by demonstrating that the original governing equations (\cref{eq:mass,eq:mom,eq:angmom,eq:eng,eq:mass_flux,eq:mom_flux,eq:ang_flux,eq:eng_flux}) are recovered in the limit as $\epsilon\to0$.
To demonstrate this efficently, the diffuse equations are written compactly as
\begin{align}
  \eta\frac{\partial \bm{f}}{\partial t} + \nabla\cdot(\eta \mathbf{G}) + \eta\nabla\cdot(\mathbf{H}) = \mathbf{K}\cdot\nabla\eta,
\end{align}
where any of the four governing equations are recoverd through the appropriate selection of quantities $\bm{f}$, $\mathbf{G}$, $\mathbf{H}$, and $\mathbf{K}$.
For example, $\bm{f}$ is a scalar for the mass and energy conservation equations, a vector for linear momentum, and zero for angular momentum.
The remaining variables $\mathbf{G}$, $\mathbf{H}$, and $\mathbf{K}$ are tensors of order $n+1$ where $n$ is the rank of the equation.
Now, equivalence is proven by writing the diffuse equations in weak form:
\begin{align}
  \int_B\Big[  \eta\frac{\partial \bm{f}}{\partial t} + \nabla\cdot(\eta \mathbf{G}) + \eta\nabla\cdot(\mathbf{H}) - \mathbf{K}\cdot\nabla\eta\Big]dV = 0 \ \ \ \ \forall B\subset\mathbb{R}^{\operatorname{dim}}
\end{align}
To derive the original conservation equations, $\mathbb{B}$ is restricted to the interior of the diffuse boundary $\Omega_\epsilon$, recalling that $\Omega_\epsilon=\operatorname{supp}(\eta=1)$.
\begin{align}
  \int_B\Big[\frac{\partial \bm{f}}{\partial t} + \nabla\cdot(\mathbf{G}) + \nabla\cdot(\mathbf{H})\Big]dV = 0 \ \ \ \ \forall B\subset\Omega_\epsilon
\end{align}
Passing back into the strong form, where the restriction to $\Omega_{\epsilon}$ (naturally) limits the application to the interior of the domain, yielding the interior governing equation:
\begin{align}
  \frac{\partial \bm{f}}{\partial t} + \nabla\cdot(\mathbf{G}) + \nabla\cdot(\mathbf{H})= 0 \ \ \ \ \  \forall \bm{x} \in\Omega,
\end{align}
as expected, since $\Omega_\epsilon\to\Omega$ as $\epsilon\to0$.
Next, the boundary terms are considered.
Observing that the weak form holds for all $B\in\mathbb{R}^{\operatorname{dim}}$, it must also hold for all $A\times[-\epsilon/2,\epsilon/2] \ \ \ \forall A\subset\partial\Omega$.
Consequently, it follows that
\begin{align}
  \int_A\int_{-\epsilon/2}^{\epsilon/2}\Big[  \eta\frac{\partial \bm{f}}{\partial t} + \nabla\cdot(\eta \mathbf{G}) + \eta\nabla\cdot\mathbf{H} - \mathbf{K}\cdot\nabla\eta\Big]dV = 0 \ \ \ \ \forall A\subset\partial\Omega
\end{align}
Application of the product rule and light rearranging of terms yields
\begin{align}
  \int_A\int_{-\epsilon/2}^{\epsilon/2}\Big[
  \Big(\frac{\partial \bm{f}}{\partial t}
  + \nabla\cdot \mathbf{G}
  + \nabla\cdot\mathbf{H}\Big)\,\eta
  + \Big(\mathbf{G} - \mathbf{K}\Big)\cdot\nabla\eta
  \Big]dV = 0 \ \ \ \ \forall A\subset\partial\Omega
\end{align}
Noting that the order parameter is defined such that $\nabla\eta=\bm{n}\,|\nabla\eta|$, Theorem 1 may be applied to recover the weak form of the boundary flux terms,
\begin{align}
  \int_\Omega (\mathbf{G}-\mathbf{K})\,dA = \bm{0} \ \ \ \forall A\subset\partial\Omega,
\end{align}
which lead to the strong form that holds over the boundary of the domain:
\begin{align}
  \mathbf{G}-\mathbf{K} = \bm{0} \ \ \ \ \forall \bm{x}\in\partial\Omega.
\end{align}
This shows that the diffuse boundary source term formulation is equivalent to the discrete boundary condition in the sharp interface limit, as long as the boundary conditions are prescribed in the form of fluxes.
The next step is to show that all boundary conditions of interest can be written in such a form.

\subsection{Boundary conditions}

This section specializes the boundary fluxes $\dot{m}$, $\dot{\mathbf{P}}$, $\dot{\mathbf{L}}$, and $\dot{e}$ for various boundary conditions of interest.
As a general comment, it is always assumed that only one type of boundary condition is applied through the diffuse source term.
This is not a limitation of the theory, which is certainly capable of treating multiple boundary condition types.
At the same time, such a generalization introduces needless complexity and will therefore not be considered here.

\subsubsection{The no-boundary condition}

In many cases, it is desirable to choose source terms that entirely eliminate the effect of the boundary on the flow.
This is accomplished by selecting flow variables that correspond to terms in the governing equations that are then canceled.
The no-boundary terms, which will sometimes be called ``passive'' source terms, are:
\begin{align}
  \dot{m} &\mapsto \rho\,\bm{u}\cdot\bm{n}
  &
    \dot{\textbf{P}} &= \rho\bm{u}\otimes\bm{u} - \mathbb{M}\nabla\bm{u} + p
  &
    \dot{\textbf{L}} &= \bm{0}
  &
    \dot{e} &= \Big[\frac{1}{2}\rho\,(\bm{u}\cdot\bm{u})\,\bm{u} + \big[U(p,T)+p\big]\,\bm{u} - (\mathbb{M}\nabla\bm{u})\,\bm{u}\Big]\cdot\bm{n} .
              \label{eq:noboundary}
\end{align}
Including these terms as source terms cancel out the effect of $\nabla\eta$ on the governing equations, successfully removing the influence of the boundary on the flow.
Note that it is essential to use consistent stencils when calculating derivatives for these terms, and sometimes it makes more computational sense to rearrange the governing equations rather than use these terms explicitly.
In this work, the passive source term formulation will be used for purposes of clarity, with the understanding that each boundary condition may be optimized differently in a final implementation.

\subsubsection{Normal mass flux and the non-penetration condition}

Here we begin with the basic boundary conditions of prescribed flow normal to the boundary.
It is often enforced implicitly with the no-slip condition, but we will consider it to be a distinct condition in this work for reasons that will be discussed subsequently.
First, consider a boundary through which a normal velocity, $u_0$, is prescribed.
Density and energy are left as unknowns.
Then, the prescription of normal velocity at the boundary is achieved by defining the mass flux condition
\begin{align}
  \dot{m} &\mapsto \rho\,u_0\bm{n}
\end{align}
The other fluxes are defined as in \cref{eq:noboundary}.
In this case, the unknown density is actually the density of the flow at the time; that is to say, the strength of the source term is determined by the density of the flow within the diffuse boundary.
This can then be substituted into the continuity equation and rearranged to show
\begin{align}
  \Big[\frac{\partial\rho}{\partial t} + \nabla\cdot(\rho\bm{u})\Big]\,\eta + \Big[\rho(u_0-\bm{u}\cdot\bm{n})\Big]|\nabla\eta| = 0.
\end{align}
This recovers, by application of theorem 1, assuming that $\rho>0$, the boundary condition in the limit as $\epsilon\to0$:
\begin{align}
  \bm{u}\cdot\bm{n}=u_0 \ \ \ \forall\bm{x}\in\partial\Omega.
\end{align}
Alternatively, one may wish to prescribe a total normal mass flux through the boundary $\dot{m}_0$.
In this case, the mass flux condition is simply
\begin{align}
  \dot{m} \mapsto \dot{m}_0,
\end{align}
that is, the mass flux is simply set to the prescribed value.
When substituted into the continuity equation and taken to the limit, it yields
\begin{align}
  \rho(\bm{u}\cdot\bm{n}) = \dot{m}_0  \ \ \ \ \forall\bm{x}\in\partial\Omega,
\end{align}
as desired.
Obviously, in the case of the non-penetration condition, $\dot{\bm{m}}_0$ is simply zero. 

It is important to observe that both of these cases allow for either normal velocity or normal mass flux to be prescribed; it is impossible to specify both simultaneously using this form of the boundary condition.
This is because the boundary condition is mediated through the conservation of mass, which is only a single scalar equation.
In order to prescribe two variables simultaneously (e.g., velocity and density), boundary conditions enforced through another equation (e.g., energy conservation) is required.

Though the non-penetration condition is inherently tied to the conservation of mass, it has an obvious effect on the fluid momentum as well.
The translation from the mass-conserving effect of the non-penetration condition is communicated to the momentum conservation equations through the equation of state: violation of the non-penetration condition induces a localized compensating mass flux at the boundary.
This compensating mass flux increases the flow density in the vicinity of the diffuse boundary, which translates to an increase in pressure as a result of the EOS.
The resulting pressure gradient then acts as an effective, conjugate, momentum flux resulting from the presence of an impenetrable boundary.
Thus the non-penetration condition, which is inherently tied to the conservation of mass, is mediated by the EOS. 

This mediating effect of the EOS is crucial, because it means that the EOS is essentially responsible for properly enforcing the boundary condition.
That is, an incomplete or non-closed EOS may result in improperly enforced boundary conditions, for example by allowing infinite accumulation of mass within the diffuse boundary and no compensating pressure gradient.

In some cases, however, it is desirable to use a simplified EOS.
For example, assuming that the fluid is both thermally and calorically perfect, the gamma-law EOS may be used, which only relates energy to pressure;
the temperature changes enter only implicitly and may be recovered through the thermal EOS if desired.
In the case of diffuse boundaries, such a simplified EOS is not capable of properly enforcing the non-penetration condition, instead allowing density to increase indefinitely without incurring a pressure gradient.
One way to rectify this issue is, naturally, to complete the EOS, though this may incur an undesirable degree of complexity for the remaining flow.

Another option is to augment the EOS with an additional term that specifically models the fluid interaction with the wall.
Numerous such EOS models exist and could readily be tailored to match a particular desired wall behavior (for instance, in the case of a semi-permeable membrane).
In the present context, a simple ``wall pressure'' augmentation is defined that is nothing other than a linear response to the normal velocity difference within the diffuse boundary
\begin{align}
  p_w = \lambda \,(\bm{u} - u_0\bm{n})\cdot\nabla\eta,
\end{align}
where $\lambda$ is a constitutive coefficient.
When the pressure is substituted into the linear momentum equation, expansion of the derivative leads to
\begin{align}
  \nabla \cdot ( \eta p_w )
  =
  \lambda \nabla\eta \cdot(\bm{u}-u_0\bm{n})\nabla\eta
  + \lambda \eta \nabla\big[(\bm{u}-u_0\bm{n})\cdot\nabla\eta\big] 
\end{align}
In the present work, only the first term is considered, and so the resultant effect of the additional ``wall pressure'' is the prescription of a linear momentum source term:
\begin{align}
  \dot{\mathbf{P}}_0\nabla\eta = \lambda \nabla\eta\,(\bm{u} - u_0\bm{n})\cdot\nabla\eta.
  \label{eq:lagrange_multiplier}
\end{align}
The wall strength, $\lambda$, controls the degree to which the wall is able to resist the flow, and the limit $\lambda\to\infty$ corresponds to a perfectly impenetrable wall.
In practice, and in the present work, $\lambda$ may simply be chosen to meet the desired tolerance.

\subsubsection{Tangential velocity and the no-slip condition}

Though often relatively simple to implement in the context of discrete boundaries, boundary conditions that specify a particular tangential velocity (this is subsequently referred to as no-slip) can present a unique challenge when implemented as a diffuse source term.
This difficulty may stem from the fact that it is not immediately obvious the ``flux'' to which a no-slip-type boundary condition corresponds.
It is proposed here that any prescribed tangential velocity most naturally corresponds to a flux of angular momentum into the flow, and is therefore most naturally represented as an angular momentum source term.
Consequently, the following form for the corresponding angular momentum source term is proposed: 
\begin{align}
     \dot{\mathbf{L}}_0 \mapsto -\mathbb{M}(\bm{u}^0\otimes\bm{n}),
\end{align}
in which $\bm{u}^0$ is the vector-valued prescribed velocity at the boundary.
This may be substituted into the corresponding diffuse conservation equation,
\begin{align}
  \eta\mathbf{T} - \mathbb{M}\nabla(\eta\bm{u}) = -\mathbb{M}(\bm{u}^0\otimes\nabla\eta).
\end{align}
Applying the product rule, and noting the requirement that $\mathbb{M}$ has major and minor symmetry,
\begin{align}
  \eta\mathbf{T} - \eta\mathbb{M}\nabla\bm{u} - \mathbb{M}(\bm{u}\otimes\nabla\eta) = -\mathbb{M}(\bm{u}^0\otimes\nabla\eta)
\end{align}
Rearranging to a form consistent with theorem 1,
\begin{align}
  \eta\big[\mathbf{T} - \mathbb{M}\nabla\bm{u}\big] = \big[\mathbb{M}(\bm{u}\otimes\bm{n}) - \mathbb{M}(\bm{u}^0\otimes\bm{n})\big]|\nabla\eta|
\end{align}
Application of theorem 1 returns $\bm{T}=\mathbb{M}\nabla\bm{u}$ in the interior, and
\begin{align}
  \mathbb{M}\big((\bm{u}-\bm{u}^0)\otimes\bm{n}\big) = \bm{0},
\end{align}
on the boundary.
Since $\mathbb{M}$ is singular in general, it is not possible to invert it to recover the matched tangential velocities.
Therefore, {\it the matched velocity is only enforced when $(\bm{u}-\bm{u}^0)\otimes\bm{n}$ is in the domain of $\mathbb{M}$}.
For example, in the case of inviscid flow when $\mathbb{M}=0$, then no part of $\bm{u}$ can be matched to $\bm{u}^0$ and the no-slip condition is effectively unenforceable---as expected.
Similarly, if $\mathbb{M}$ depends only on the deviatoric portion of $(\bm{u}-\bm{u}^0)\otimes\bm{n}$, then only the projection $(\bm{u}-\bm{u}^0)\cdot\bm{t} \ \ \ \forall \bm{t}\perp\bm{n}$ can be satisfied.

Similarly to the above case of non-penetration, it can now be seen that the no-slip condition (like no-penetration, a primal or essential condition), must be mediated through the constitutive behavior of the fluid.
Just as a proper EOS is necessary to enforce the non-penetration constraint, it is also the case that a suitable viscosity model is necessary to enforce matched prescribed velocities.

In practice, angular and linear momentum equations are combined and solved simultaneously.
Care must be taken at this stage due to the delicate positioning of $\eta$.
We begin by identifying the following form of angular momentum, which is:
\begin{align}
  \eta\mathbf{T}  = \eta\mathbb{M}\nabla\bm{u} + \mathbb{M}(\bm{u}-\bm{u}^0)\otimes\nabla\eta .
\end{align}
The momentum equation depends on the divergence of this expression, which must be calculated and expressed using index notation:
\begin{align}
  \frac{\partial}{\partial x_j}(\eta\mathrm{T}_{ij})
  &= \frac{\partial}{\partial x_j}\Big[\eta\mathbb{M}_{ijkl}u_{k,l} + \mathbb{M}_{ijkl}(u_k-u_k^0)\eta_{,l}\Big] \\
  &= 
    \eta\mathbb{M}_{ijkl}u_{k,lj} +
    \mathbb{M}_{ijkl}u_{k,l}\eta_{,j} +
    \mathbb{M}_{ijkl}u_{k,j}\eta_{,l} +
    \mathbb{M}_{ijkl}(u_k-u_k^0)\eta_{,jl}
\end{align}
This indicates that the passive linear momentum source terms must be updated to compensate for this specific form of the $\nabla\eta$ boundary term:
\begin{align}
  (\dot{\mathbf{P}}_0\nabla\eta)_i \mapsto \rho u_iu_j\eta_{,j} - \mathbb{M}_{ijkl}u_{k,l}\eta_{,j} - \mathbb{M}_{ijkl}u_{k,j}\eta_{,l} + p\,\delta_{ij}\eta_{,j}
\end{align}
When these have been substituted into the linear momentum equation, the result is
\begin{align}
  \eta\Big[\frac{\partial}{\partial t}(\rho u_i)  + \frac{\partial}{\partial x_j}\Big(\rho\,u_i\,u_j - \mathbb{M}_{ijkl} u_{k,l} + p\delta_{ij}\Big)\Big] = \mathbb{M}_{ijkl}(u_k-u_k^0)\,\frac{\partial^2\eta}{\partial x_j\partial x_l}
\end{align}
It is then easy to see that the normal momentum equations are recovered in the interior, indicating that the momentum source term has been correctly adjusted.
Theorem 1 cannot be directly applied to consider the boundary term, however, because it depends on the second derivative of $\eta$ rather than the gradient.
Further discussion about the effect of the angular momentum source term is presented in the results.

\subsubsection{Discussion on the role of angular momentum in diffuse boundary driven flow}

The above derivation of the no-slip condition shows that the tangential velocity boundary condition is naturally interpreted as a flux of angular momentum into the flow.
In sharp-interface continuum theory, angular momentum is typically reduced to the symmetry of the Cauchy stress tensor in the flow.
This raises interesting questions about the form taken by an angular momentum flux, which is considered in detail in this section.
(It is noted that this subsection is incidental to the flow of the rest of the paper and can be skipped without loss of continuity.
Additional details on the derivation are also presented in \cref{sec:angular_momentum_detailed_derivation})

The integral form of the conservation of angular momentum for a volume $\Omega(t)$ that is defined such that it is advected with the flow.
In the absence of body forces, this is written as
\begin{align}
  \frac{d}{dt}\int_{B(t)} (\bm{x} \times \rho\bm{u})\, dV
  &= \int_{\partial B(t)} (\bm{x} \times \bm{\sigma}\bm{n})\, dA
    + \int_{\partial B(t)} \dot{\mathbf{L}}\bm{n}\,dA \ \ \ \ \forall B(t) \in \phi(\Omega,t)
\end{align}
The terms correspond to a total change in the angular momentum in $B$ (on the left hand side); the total moment generated by surface tractions $\bm{\sigma}\bm{n}$ (the Cauchy stress $\bm{\sigma}=\mathbf{T}-p\,\bm{I}$ acting on the surface normal $\bm{n}$); and the total angular momentum flux (on the right side, respectively). 
The above must hold for every possible subregion in the domain $\Omega$, and this will be tacitly assumed subsequently.
Application of the Reynolds's transport theorem along with judicious use of the divergence theorem and product rule yields, after a straightforward yet tedious calculation (see Appendix A for more details)
\begin{align}
  \int_B\bm{x}\times\Big[\frac{\partial\rho\bm{u}}{\partial t} + \operatorname{div}\big(\rho\bm{u}\otimes\bm{u}-\bm{\sigma}\big)\Big]\,dV
  = \int_B\,\Big[\hat{\bm{g}}_i\epsilon_{ijk}\sigma_{jk} + \nabla\cdot\dot{\mathbf{L}}\Big]\,dV,
\end{align}
where $\hat{\bm{g}}_i$ is the i$^{\text{th}}$ basis vector in the coordinate system of choice, and the divergence $\nabla\cdot$ on tensors is assumed to contract the second index. 
Now, the governing equations are recast from the discrete boundary to the diffuse boundary frame, with the replacements $\Omega\to\mathbb{R}^n$ and $dV\mapsto\eta\,dV$, resulting in:
\begin{align}
  \int_{B}\bm{x}\times\Big[\eta\frac{\partial\rho\bm{u}}{\partial t} + \eta\operatorname{div}\big(\rho\bm{u}\otimes\bm{u}-\bm{\sigma}\big)\Big]\,dV
  = \int_{B}\Big[\epsilon_{ijk}\eta\sigma_{jk}\bm{g}_i + \nabla\cdot\dot{\mathbf{L}}\eta\Big]\,dV.
\end{align}
Light manipulation with the product rule allows the equation to be written as
\begin{align}
  \int_{B}\bm{x}\times\Big[\underbrace{\eta\frac{\partial\rho\bm{u}}{\partial t} + 
    \operatorname{div}\big(\eta(\rho\bm{u}\otimes\bm{u}-\bm{\sigma})\big)}_{\text{linear momentum}}
    -
    \operatorname{grad}\eta\cdot\big(\rho\bm{u}\otimes\bm{u}-\bm{\sigma}\big)
    \Big]\,dV = \int_{B}\Big[\epsilon_{ijk}\sigma_{jk}\bm{g}_i  + \nabla\cdot\dot{\mathbf{L}}\Big]\eta\,dV,
\end{align}
where \cref{eq:diffuse_mom} is identified and substituted, yielding the intermediate result:
\begin{align}
  \int_{B}\bm{x}\times\Big[\dot{\mathbf{P}} - \rho\bm{u}\otimes\bm{u}+ \bm{\sigma}\Big]\nabla\eta\,dV
  = \int_{B}\Big[\hat{\bm{g}}_i\epsilon_{ijk}\sigma_{jk} + \nabla\cdot\dot{\mathbf{L}}\Big]\,\eta\,dV.
\end{align}

The first observation is that the left hand side has an explicit dependence upon the coordinate variable $\bm{x}$.
Because of translation invariance, however, the equation must hold under the shift $\bm{x}\mapsto\bm{x}+\bm{b}$ for any $\bm{b}$.
This is only possible if the multiplier itself is zero, which leads to the observation
\begin{align}
  \Big[\dot{\mathbf{P}} + \bm{\sigma} - \rho\bm{u}\otimes\bm{u}\Big] \bm{n} = \bm{0}
\end{align}
at all points within the diffuse boundary.
In the typical case in which $\dot{\mathbf{P}}$ contains a ``passive'' advection component, this condition merely leads to the balance of the momentum source term with the Cauchy stress tensor.

Now consider the right hand side, which has been shown to be equal to zero.
Recalling that the above holds for all $\bm{B}$ in real space, one obtains the strong form relationship
\begin{align}
  \nabla\cdot\dot{\mathbf{L}} = -\bm{g}_i\epsilon_{ijk}\sigma_{jk}.
\end{align}
Usually, the Cauchy stress tensor is symmetric, causing the right hand side to vanish.
This implies that, when an angular momentum source term is present, the Cauchy stress tensor contains an {\it anti-symmetric component}.
That is, the effect of an angular momentum source term, such as a no-slip boundary condition, is to break the symmetry of the stress tensor to generate a localized diffuse couple moment within the boundary.
Moreover, this explains why the $\dot{\mathbf{L}}$ has tensorial character: it is mapped to a skew-symmetric tensor, and therefore has only 3 degrees of freedom (in 3D) rather than 9.

\subsubsection{Other boundary conditions}

Here, a basic treatment of other types of less common boundary conditions is presented, including ``natural'' boundary conditions, for purposes of demonstrating the versatility of the diffuse boundary framework.
One such boundary condition is the imposition of a ``traction'' at a surface, perhaps imposed by a force-controlled sliding wall.
If the known, prescribed traction is $\bm{\tau}_0$, then the corresponding diffuse source term is
\begin{subequations}\begin{align}
  \bm{\dot{P}}_0 &\mapsto \rho\bm{u}\otimes\bm{u}  + p - \bm{\tau}_0,
\end{align}\end{subequations}
which mixes passive and active terms.
(Other fluxes are defined as in \cref{eq:noboundary})
Substituting into the diffuse momentum equation, and rearranging, gives
\begin{align}
      \eta\Big[\frac{\partial (\rho u)}{\partial t} + \nabla\cdot(\rho\bm{u}\otimes\bm{u}) - \nabla\cdot(\mathbf{T})\Big] &= \Big(\mathbf{T} - \bm{\tau}_0 \Big)\nabla\eta,
\end{align}
recovering the discrete boundary condition in the sharp interface limit by application of Theorem 1.
Observe that the condition requires $\mathbf{T}=\bm{\tau}_0$, which requires that the fluid is capable of sustaining such a stress tensor.
For an inviscid flow, this condition is ill-defined and will lead to numerical instability.
Related to the applied traction boundary condition is the prescribed pressure, or ``wall'' boundary condition by modifying the linear momentum source term as
\begin{subequations}\begin{align}
  \bm{\dot{P}}_0 &\mapsto \rho\bm{u}\otimes\bm{u}  + p_0 - \mathbf{T}
\end{align}\end{subequations}
Then, once again substituting into the diffuse interface equations and applying the product rule,
\begin{subequations}\begin{align}
  \eta\left(\frac{\partial (\rho u)}{\partial t} + \nabla\cdot(\rho\bm{u}\otimes\bm{u} - \mathbb{M}\nabla\bm{u} p)\right) \eta
  &= (p - p_0) \nabla \eta 
\end{align}\end{subequations} 
As before, the unknown quantities are allowed to evolve according to the governing equations, leaving the prescribed pressure to drive the flow.
Applying the Navier-Stokes equations and applying the sharp interface limit as in previous sections, one can see that $p = p_0$ and $U(p,T) = U(p_0,T)$.
(Note that when $U(p,T)$ is linear, such as for a calorically perfect gas, this second equation is redundant.
For more complicated equations of state, nevertheless, both conditions are needed to ensure recovery of boundary conditions from the diffuse source terms.)

Finally, one may consider the imposition of an energy flux through the following definition of the energy flux term
\begin{align}
  \bm{\dot{e}}_0 &\mapsto \frac{1}{2}\rho (\bm{u}\cdot\bm{u})\bm{u} + (U(p,T) + p)\bm{u} -(\mathbb{M}\nabla\bm{u})\bm{u} + \bm{q}_0
\end{align}
(with mass and momentum sources unchanged).
One can follow the exact same procedure as above to recover $\bm{q} = \bm{q}_0$ at the boundary in the sharp interface limit.

More complex boundary condition types can be constructed by following the same procedure.
Moreover, one can apply this same process to other formulations of the flow equations.
The essential step is to identify the appropriate flux term to which the desired boundary condition corresponds, and then to design the source term to produce the desired condition in the sharp interface limit.

\section{Numerical considerations}\label{sec:numerical_considerations}

This section addresses some of the numerical techniques that are necessary in the implementation of the diffuse boundary theory.
The above theory is complete and has been shown exact in the sharp limit, but there are many complicating factors in the actual implementation that can lead to complexity or undesired behavior.

\subsection{Treatment of conserved variables as mixtures}

The replacement $\Omega\to\mathbb{R}^n$ prompts the question of what to do with the field variables in the region outside of the original domain.
Since the diffuse boundary theory amounts to $\bm{0}=\bm{0}$ in $\mathbb{R}^n\setminus\Omega_\epsilon$, there are no equations to govern the behavior of the field variables, and their values are therefore undefined.
While this does not ultimately affect the theory, the lack of definition can result in numerical complications, especially at the exterior edge of the diffuse domain, where spurious gradients are left untempered.
In this formulation, it is desirable to work exclusively with the conserved fields of density, momentum, and energy, as they are the most conducive to the flux source term formulation.
(Primitive variables, when needed, are simply derived from the conserved quantities.)

To ensure that all solution fields are well-defined everywhere, define ``mixed'' fields, $\bar{\rho}$, $\bar{\bm{p}}$, and $\bar{E}$, corresponding to the following mixture rule
\begin{align}
  \bar{\rho} &= \eta\,\rho_f + (1-\eta)\rho_s
  &
    \bar{\bm{p}} &= \eta\,\bm{p}_f + (1-\eta)\bm{p}_s
  &
  \bar{E} &= \eta\,E_f + (1-\eta)E_s 
\end{align}
where $\rho_f$, $\bm{p}_f$, and $E_f$ are the conserved field solutions to the diffuse viscous flow equations (the ``fluid'' values), and $\rho_s$, $\bm{p}_s$, and $E_s$ are externally defined states of the solution in the exterior of the fluid domain (the ``solid'' values).
The choice of the solid field variables is arbitrary, and can even be determined by a completely separate set of governing equations.
In this work, they are simply set to constant values.
(Nevertheless, it is important to recognize that these values must be compatible with the numerical methods in the fluids solver; setting $\rho_s\to0$, for instance, can trigger numerical error in a Riemann solver unless special care is taken to avoid this issue.)

The advantage of the mixed field formulation is that they are well-defined everywhere in $\mathbb{R}^n$, even though the fluid fields are only defined in the fluid domain (and undefined elsewhere).
It is essential to work with well-defined fields, because lack of definition can lead to numerical complications, especially outside of the fluid region.
Then, the challenge is how to formulate the theory to rely only on the mixture fields and the solid fields, eschewing any explicit dependence on the fluid fields.

Starting with the density field, we take the time derivative of the mixed density field $\bar{\rho}$, and perform a straightforward application of the product rule.
Then, by substituting again the definition of $\bar{\rho}$, the following expression is obtained:
\begin{align}
  \frac{\partial\bar{\rho}}{\partial t} =
  \eta\frac{\partial\rho_f}{\partial t} +
  (1-\eta) \frac{\partial\rho_s}{\partial t} + 
  \frac{1}{\eta}\frac{\partial\eta}{\partial t}\Big(\bar{\rho} - \rho_s\Big)
\end{align}
This equation represents the relationship between the time evolution of the mixed field and the time evolution of the fluid field, which is the first term on the right hand side.
The second two terms on the right hand side enforce, respectively, consistency between the mixed and solid field if the solid field is changing in time; and any time-dependent changes in the domain. 

The fluid fields $\rho_f$, $\bm{p}_f$, and $E_f$ are not tracked explicitly, and so must be calculated on-the-fly in the computation of their $\eta$-multiplied partial derivatives as determined from the governing equations.
They can be calculated simply by inverting the mixture rule to obtain
\begin{align}
  \rho_f = \frac{1}{\eta}\big(\bar{\rho} - (1-\eta)\rho_s\big),
\end{align}
and so on, for $\bm{p}_f$ and $E_f$.
Clearly, fluid values are well-defined within the support of $\eta$, but undefined in the solid region.
Theoretically, this is problem free, since the derivative itself is multiplied by $\eta$.
Numerical issues may arise, however, especially if using solvers that are agnostic to the diffuse interface method.
Riemann solvers that are not very robust may cause errors when seemingly problematic flow values are encountered.
These issues are effectively remediated by (1) adding a small value ($\zeta<<1$) to $\eta$ to prevent numerical error when $\eta \rightarrow 0$ in the denominator; and (2) choosing a cutoff value of $\eta$ under which all fluxes are explicitly set to zero.
In the work presented here, both (1) and (2) proved effective at stabilizing the solver and preventing unwanted information propagating through the solid region.
Naturally, many other stabilizing measures may exist; these are left to future work.

As only conserved variables are tracked in this method, primitive variables (velocity, pressure) are only calculated as needed or during post-processing.
Unlike the conserved variables, for which a mixture rule provides global well-posedness, the primitive variables are necessarily undefined except in the fluid region.
Therefore, when primitive values are calculated, they are usually multiplied by $\eta$ to eliminate misleading values within the solid.

\section{Numerical implementation}\label{sec:implementation}

This work uses an advection-diffusion, Roe approximate Godunov, finite volume scheme to solve the diffuse interface equations in two dimensions \cite{godunov1959finite,roe1981approximate,laney1998computational}.
This solver leverages equivalent wave speeds, and eigenvectors shared between the Navier-Stokes equations and the diffuse interface system, as demonstrated in the work of \textcite{schmidt2022self}.
The first order explicit Euler method is employed for time integration, and second order accurate finite difference stencils are used for second order spatial derivatives.

All algorithms are implemented in the authors' in-house open-source code Alamo \cite{runnels2021massively}, which is built upon the AMReX framework \cite{zhang2019amrex}.
The block-structured AMR capability provided through AMReX enables sufficient resolution of the diffuse boundaries without prohibitive computational cost.
(All of the simulations presented here are run either on a desktop computer or a single compute node, with none of the runs exceeding 24 hours.)
Temporal substepping, with dynamically adjusted timesteps, further reduce the numerical expense. 
Alamo uses a multiple-inheritance polymorphic integrator system to manage multiphysics applications.
This enables the efficient coupling of the present flow solver to a number of other multiphysics applications that are already implemented in Alamo.

\section{Examples}\label{sec:examples}

The remainder of this work is dedicated to a variety of examples that demonstrate diffuse source terms for a range of boundary behavior. 
Initial examples are canonical flow cases that are presented to verify the approach and demonstrate convergence.
The latter examples demonstrate the method's versatility and potential for coupling to complex boundary-driven problems.

\subsection{Parallel shear flows}

The focus of this section is the verification of the diffuse angular momentum source formulation of the tangential velocity (no-slip) boundary condition.
For this, two classic examples of shear flow are investigated: Poiseuille flow and Couette flow between two infinite and parallel plates.
(It is noted that the previous work has already demonstrated the ability of the method to capture other types of boundary conditions, so they are not repeated here.)

\begin{figure}
    \begin{minipage}[t]{0.46\linewidth}
        \vspace{0pt}
        \begin{subfigure}[t]{\linewidth}
            \centering
            \includegraphics[width=\linewidth,clip,trim=0cm 1.3cm 0cm 1.2cm]{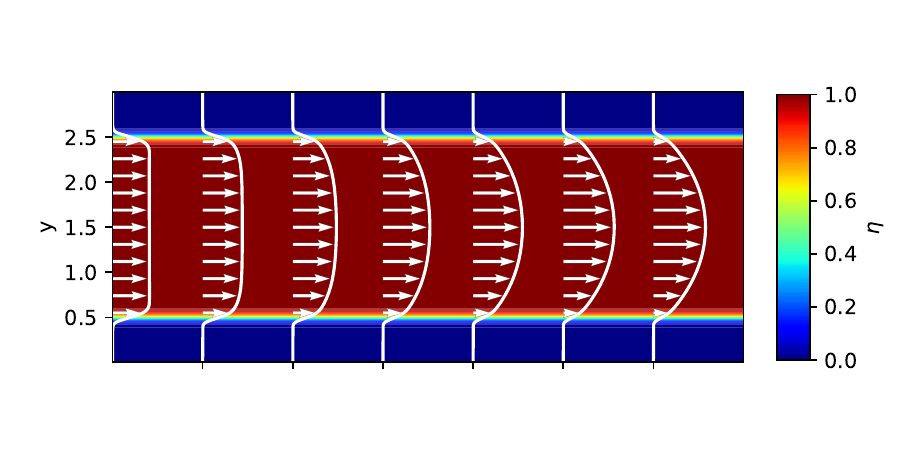}
            \caption{Order parameter $\eta(\bm{x})$}
            \label{fig:channelfieldplot_order}
        \end{subfigure}
        \begin{subfigure}[t]{\linewidth}
            \centering
            \includegraphics[width=\linewidth,clip,trim=0cm 0.4cm 0cm 1.2cm]{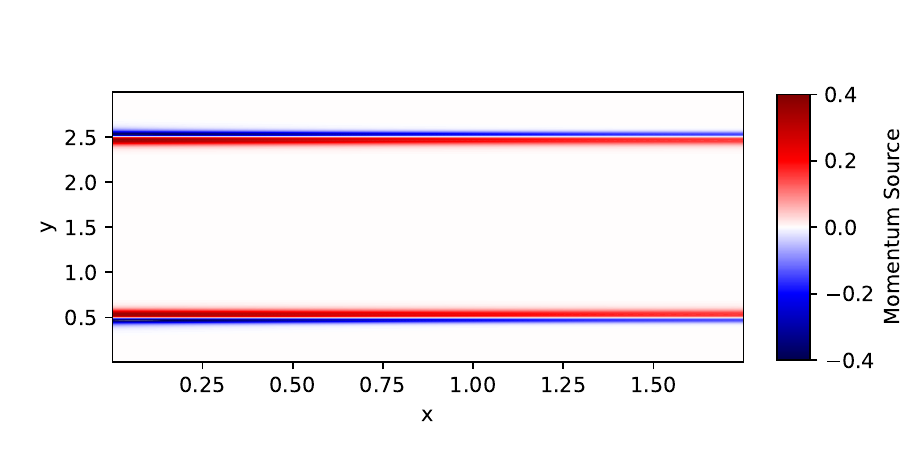}
            \caption{$x$ component of the viscous source term}
            \label{fig:channelfieldplot_source}
        \end{subfigure}
    \end{minipage}
    \hfill
    \begin{minipage}[t]{0.53\linewidth}
        \vspace{0pt}
        \begin{subfigure}[t]{\linewidth}
            \centering
            \includegraphics[width=\linewidth,clip,trim=0cm 0.4cm 0cm 0.2cm]{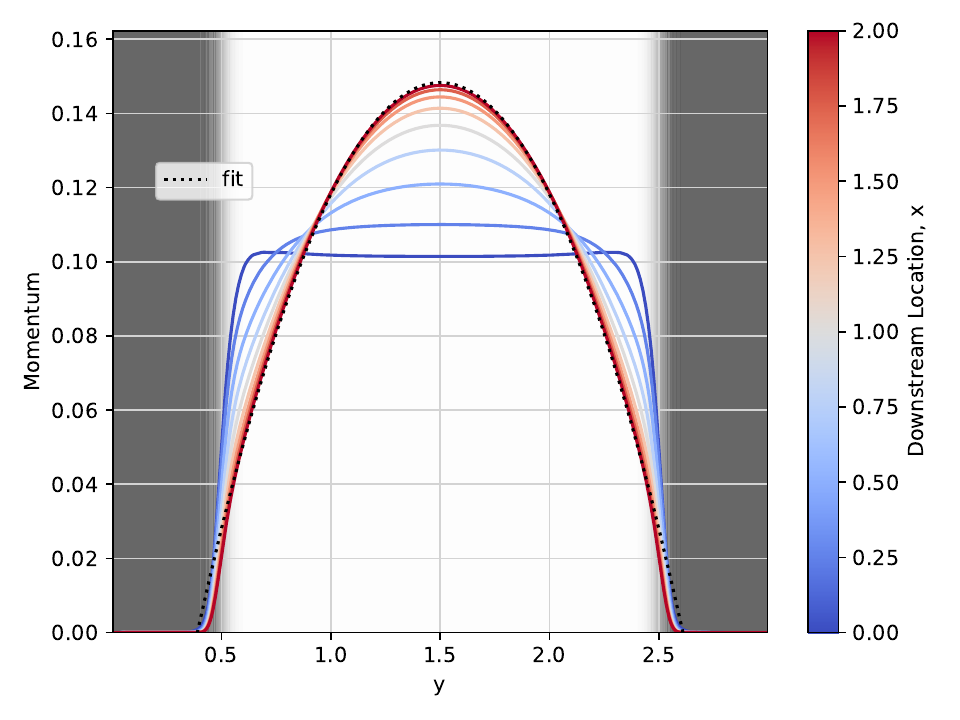}
            \caption{Entrance region}
            \label{fig:channelfieldprofile_entrance}
        \end{subfigure}
    \end{minipage}
    \caption{Flow through the entrance region of a channel with a diffuse thickness of $\epsilon=0.05$. Contours of the order parameter $\eta$, with arrows indicating the development of the flow profile (a).  Doublet structure of the $x$ component of the viscous source term (b). Superimposed momentum profiles from entrance to fully developed flow (c).}
    \label{fig:channelfieldplot}
\end{figure}

\subsubsection{Poiseuille flow}

In Poiseuille flow, often referred to as channel flow or 2d pipe flow, two boundary layers form at the channel walls and thicken until they eventually meet at the center of the channel.
At this point, the transverse velocity profile is parabolic, the pressure gradient in the flow direction is constant, and the flow is fully developed \cite{white2005viscous}.
The implicit method presented in this work accurately captures the development of the flow, which takes place in the channel as implicitly defined by the smoothly varying order parameter $\eta$, resulting in the development of parabolic flow through the cross section (\cref{fig:channelfieldplot_order}).
The no-slip condition at the solid boundary is enforced solely by the viscous source terms within the diffuse region, which can be seen as a doublet along the channel walls (\cref{fig:channelfieldplot_source}).
The momentum source can be imagined as an injection of mass in the diffuse region to counter the flow at the wall and enforce the no-slip condition.
An unsurprising repercussion of this implementation is that momentum is non-physically introduced into the system; thus, a Lagrange multiplier is used to maintain conservation of momentum (\cref{eq:lagrange_multiplier}).
This Lagrange term can be set arbitrarily high depending on the tolerance required for the simulation.
The flow is shown to converge at the end of the entrance region to a steady parabolic profile as derived in the analytic solution (\cref{fig:channelfieldprofile_entrance}).

It has been discussed earlier that the diffuse method proposed in this work converges to the sharp interface solution as the diffuse thickness decreases.
This convergence is exemplified qualitatively by comparing the fully-developed flow profiles with varying diffuse thickness and are plotted against the analytic solution provided by \textcite{anderson1991fundamentals} and shown for reference in \cref{eq:channelflow} (\cref{fig:channelfieldprofile_comp}).
Note that this equation is multiplied by the local density to obtain momentum,
\begin{equation}
    u = \frac{1}{2\mu}\left(\frac{dp}{dx}\right)(y^2-Dy),
    \label{eq:channelflow}
\end{equation}
where $u$ is the axial velocity, $\mu$ is the dynamic viscosity of the fluid, $dp/dx$ is the axial pressure gradient, $D$ is the channel height, and $y$ is the transverse distance along the channel.
This equation is technically only valid for incompressible flow, so the simulations are performed at low Mach numbers, which become incompressible in the limit as $M\to 0$ \cite{white2005viscous}.
Note that the gray regions indicate the solid walls with a diffuse thickness of $\epsilon=0.05$, although multiple values of $\epsilon$ are plotted.
This is simply to give an indication of the location of the diffuse region, but the impact of the diffuse thickness is apparent as $\epsilon$ increases and the flow widens.

To quantify convergence, a simple error is computed by taking the relative difference in the centerline momentum of the analytic solution and the simulation results, given by
\begin{equation}
    \varepsilon_\mathrm{rel} = \frac{|\max({P_\mathrm{exact})-\max(P_\mathrm{approx})|}}{|\max(P_\mathrm{exact})|},
    \label{eq:relerror}
\end{equation}
where $P$ is the $x$ component of momentum and $\varepsilon$ is the true relative error, not to be confused with the diffuse thickness, $\epsilon$.  The error is shown to follow a power law relationship with the diffuse thickness (\cref{fig:channel_relerror}) according to the equation:
\begin{equation}
    \varepsilon_{rel}=0.864\epsilon^{0.863},
\end{equation}
with a coefficient of determination of $0.991$.
Although the measurement of fit is reasonably high, since a limited number of data points were used, it would be unwise to assume this relationship is correct based on these results alone.
(It will be shown in the next section that with Couette flow, however, that this assumption is indeed justified.)

\begin{figure}[b]
  \centering
  \begin{subfigure}{0.48\linewidth}
    \centering
    \includegraphics[height=5.5cm]{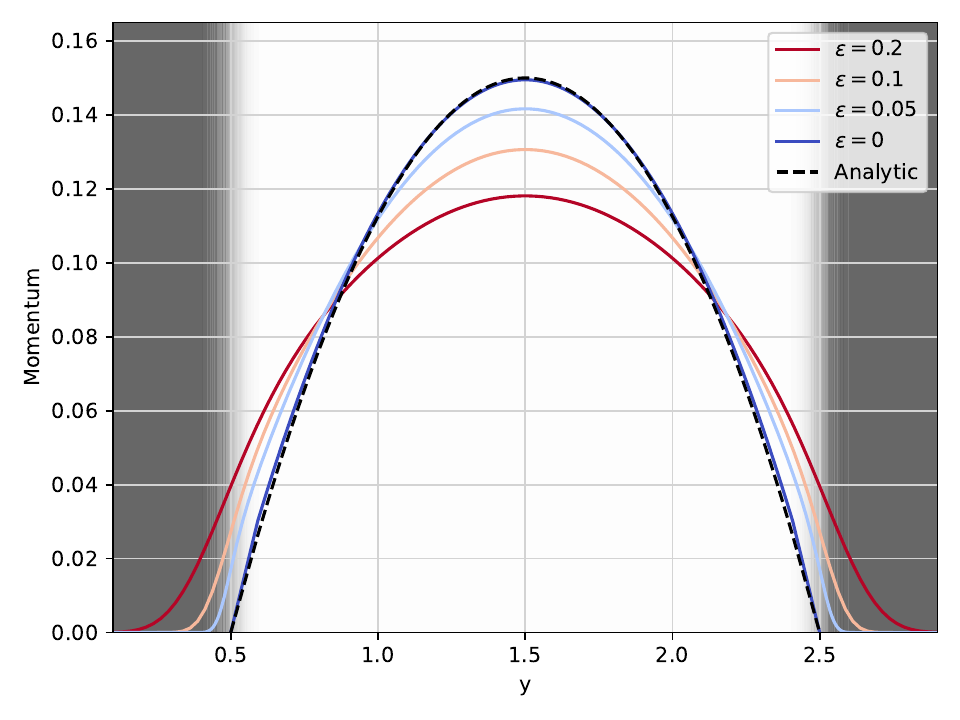}
    \caption{Fully developed}
    \label{fig:channelfieldprofile_comp}
  \end{subfigure}
  \begin{subfigure}{0.48\linewidth}
    \centering
    \includegraphics[height=5.5cm]{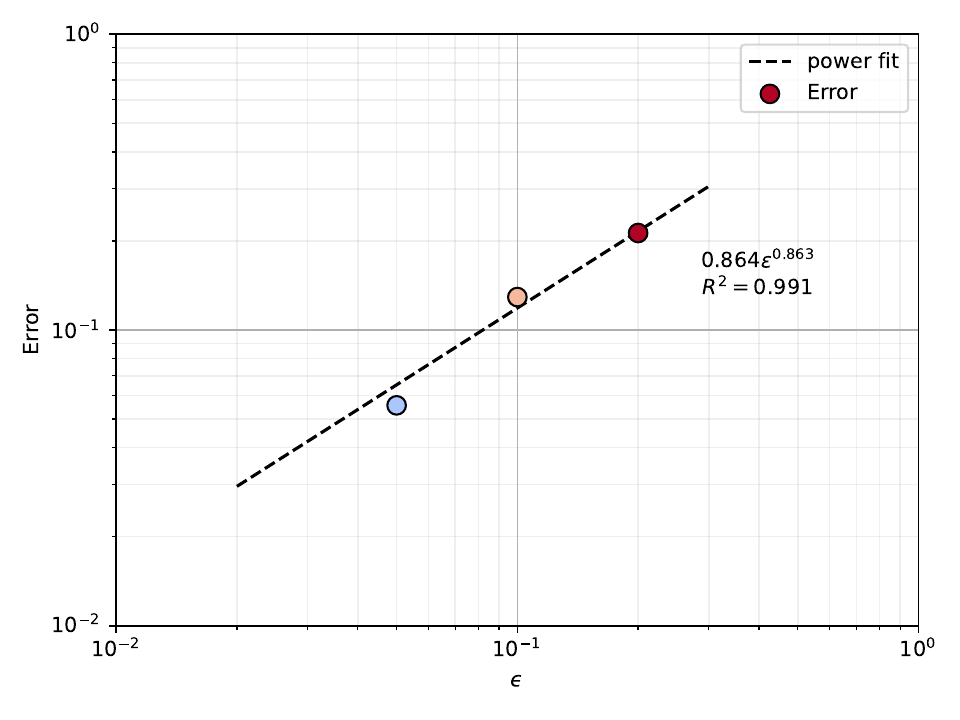}
    \caption{Relative error vs diffuse thickness}
    \label{fig:channel_relerror}
  \end{subfigure}
  \label{fig:channelfieldprofile}
  \caption{
    Convergence towards the analytic solution of fully developed channel flow between two parallel plates is shown. Momentum profiles at varying diffuse interface thicknesses, $\epsilon$ (a) and the power law relationship between relative error and diffuse interface thickness (b).
  }
\end{figure}

\begin{figure}
  \centering
  \begin{subfigure}{0.48\linewidth}
    \centering
    \includegraphics[width=\linewidth,clip,trim=0.5cm 0.5cm 0.5cm 0.5cm]{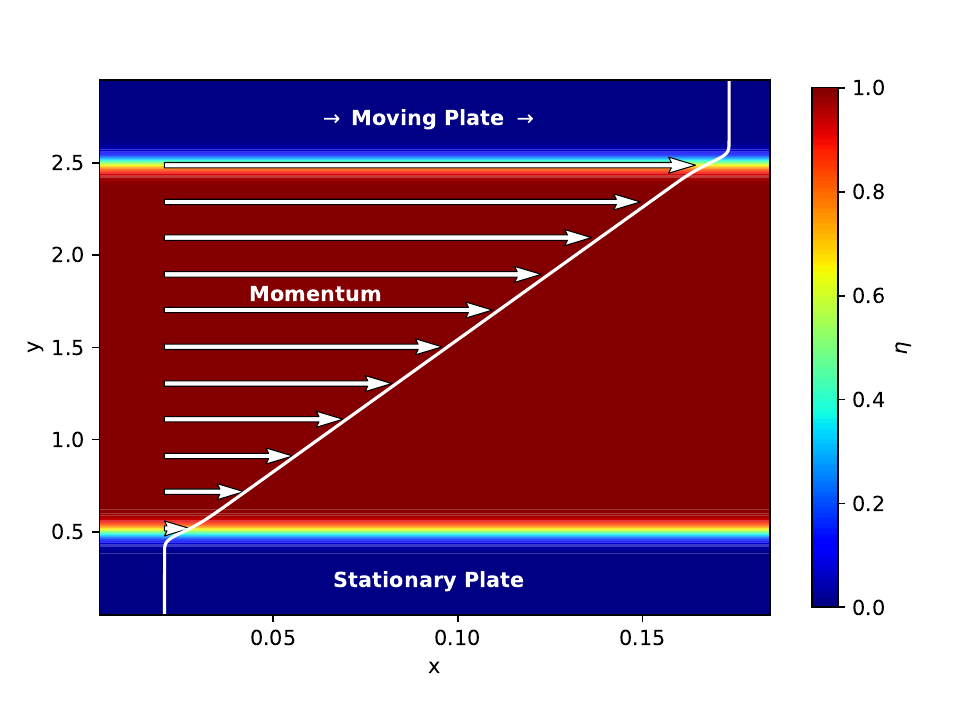}
    \caption{Order parameter $\eta(\bm{x})$}
    \label{fig:couettefieldplot_momentum}
  \end{subfigure}\hfill
  \begin{subfigure}{0.48\linewidth}
    \centering
    \includegraphics[width=\linewidth,clip,trim=0.5cm 0.5cm 0.5cm 0.5cm]{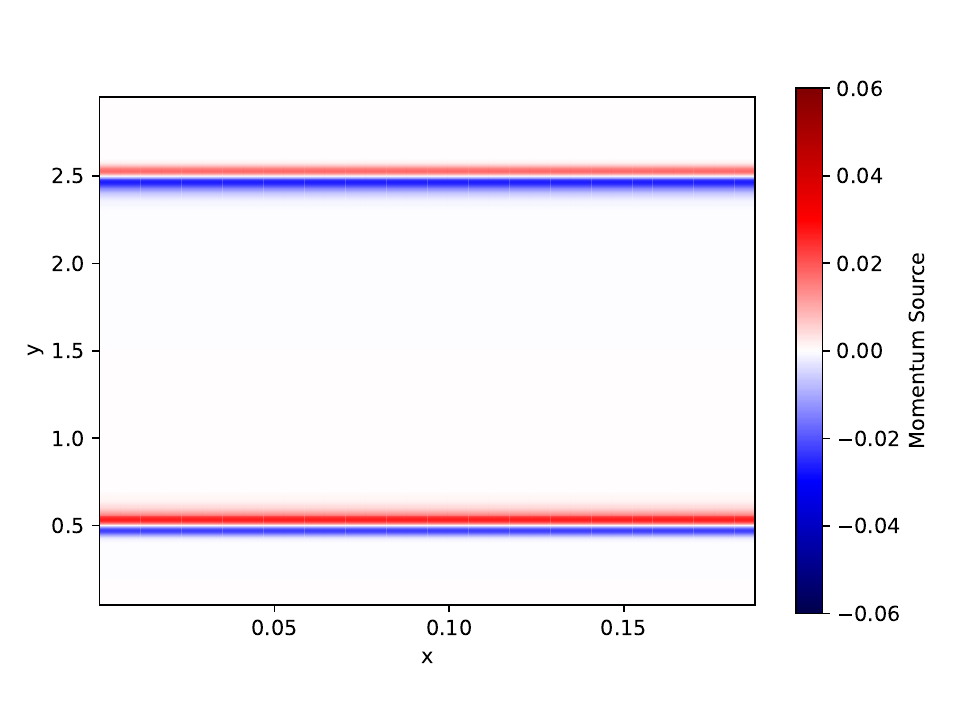}
    \caption{$x$ component of the viscous source term}
    \label{fig:couettefieldplot_source}
  \end{subfigure}
  \caption{Couette flow with a diffuse thickness of $\epsilon=0.05$. Contours of the order parameter $\eta$, with arrows indicating the flow profile (a).  Doublet structure of the $x$ component of the viscous source term (b).}
  \label{fig:couettefieldplot}
\end{figure}
\begin{figure}
  \centering
  \begin{subfigure}{0.48\linewidth}
    \includegraphics[height=5.5cm]{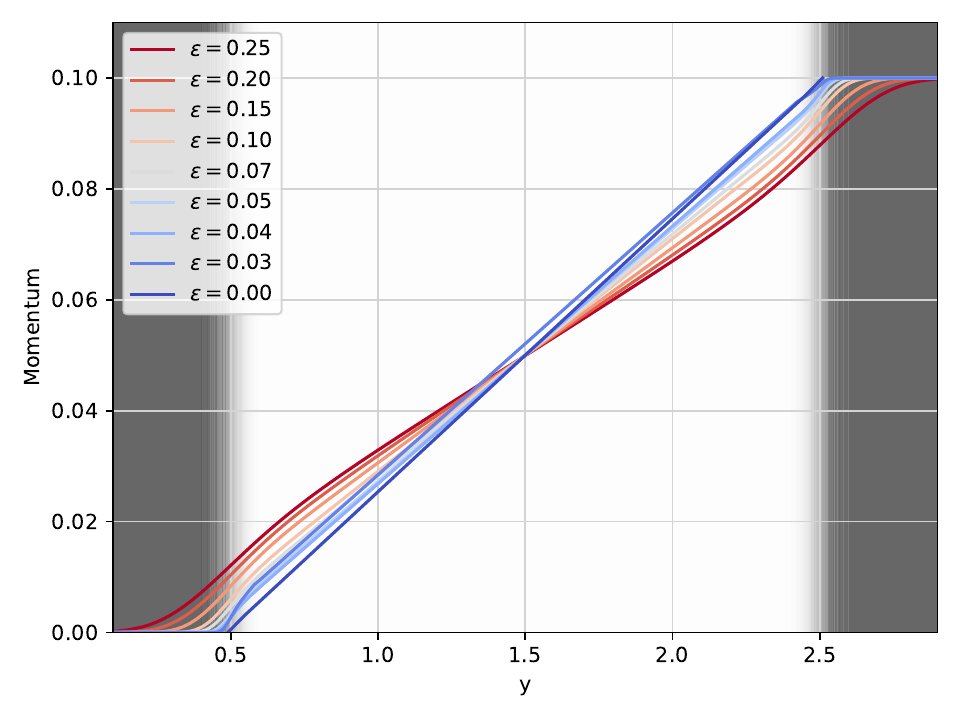}
    \caption{Fully developed}
    \label{fig:couette_convergence}
  \end{subfigure}
  \begin{subfigure}{0.48\linewidth}
    \includegraphics[height=5.5cm]{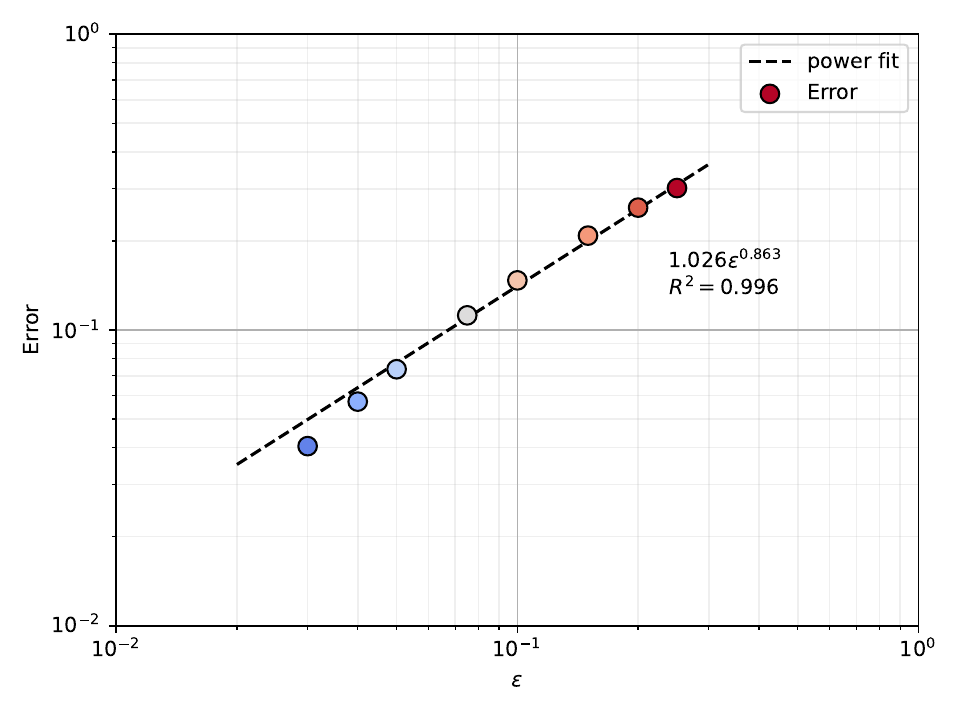}
    \caption{Relative error vs diffuse thickness}
    \label{fig:couette_error}
  \end{subfigure}
  \caption{
    Convergence towards the analytic solution of Couette flow between two parallel plates is shown. Steady-state momentum profiles at varying diffuse interface thicknesses, $\epsilon$ (a) and the power law relationship between relative error and diffuse interface thickness (b).
  }
  \label{fig:couettefieldprofile}
\end{figure}

\subsubsection{Couette flow}

This section considers the viscous flow between two parallel surfaces moving at different velocities.
Once steady-state flow is achieved, the momentum flow profiles traversing the channel are linear \cite{white2005viscous}, which is correctly realized with the current approach (\cref{fig:couettefieldplot_momentum}).
Notice that even within the diffuse region, the linear profile is largely maintained.
As in the channel flow case, the no-slip condition is enforced by the diffuse viscous source terms (\cref{fig:couettefieldplot_source}).
It should be noted that the sign of the source doublet along the upper wall is flipped compared to the channel flow case.
This correctly illustrates that the source term is causing the flow velocity to increase to match the velocity of the moving surface, whereas for channel flow the momentum source is causing the flow to have zero tangential velocity at the surface.

Qualitative convergence is shown by considering the steady-state velocity profiles with decreasing values for $\varepsilon$ (\cref{fig:couette_convergence}), illustrating how the steady-state solution converges toward the sharp-interface solution as the diffuse thickness decreases.
Similarly to the channel flow convergence plot, the gray region represents only a single value of diffuse thickness, $\epsilon=0.05$, although momentum profiles are shown for multiple values of $\epsilon$.
Nevertheless, the impacts of the diffuse thickness are evidenced by the slope of the steady-state profile.
Error is calculated the in the same way as in \cref{eq:relerror}, except with respect to the slope of the steady-state momentum profiles rather than centerline momentum.
It can be seen (\cref{fig:couette_error}) that the error again follows a power law relationship with the diffuse thickness, given by
\begin{equation}
    \varepsilon_\mathrm{rel} = 1.026\epsilon^{0.863}
\end{equation}
Notably, the exponent in this calculation was  the same as determined in the channel flow case, further supporting the reliability of the power law relationship.
It is clear from inspection that this exact match is coincidental (one can determine by inspection a fairly large standard deviation, especially in the Couette flow case); however, the consistency between these two cases builds confidence that the diffuse boundary method converges appropriately.

\subsection{Vortex shedding}

\begin{figure}
  \centering
  \begin{subfigure}{0.48\linewidth}
    \centering
    \includegraphics[width=0.75\linewidth]{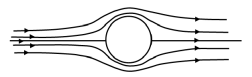}  
    \caption{Re$<<1$}
  \end{subfigure}
  \begin{subfigure}{0.48\linewidth}
    \centering
    \includegraphics[width=0.75\linewidth]{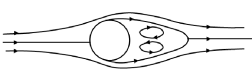}  
    \caption{Re$\approx10$}
  \end{subfigure}
  \begin{subfigure}{0.48\linewidth}
    \centering
    \includegraphics[width=0.8\linewidth]{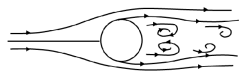}
    \caption{Re$\approx60$}
  \end{subfigure}
  \begin{subfigure}{0.48\linewidth}
    \centering
    \includegraphics[width=0.75\linewidth]{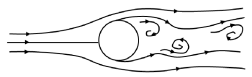}%
    \caption{Re$\approx1000$}
  \end{subfigure}
  \caption{Flow structures around a cylinder in uniform flow  \cite{hughes1999schaum}}
  \label{fig:cylinder}
\end{figure}

Vortex shedding and the resulting vortex street is a natural phenomenon that occurs in viscous flow around bluff bodies, named after Theodore von K\'arm\'an, whose original works were published in German in 1911--1912 \cite{karman1911nach,karman1912nach} and published in English in 2013 \cite{karman2013on1,karman2013on2}.
For the particular case of uniform flow around a circular cylinder, the flow field can vary dramatically for differing Reynolds numbers as illustrated by \textcite{hughes1999schaum} (\cref{fig:cylinder}).
At very low Reynolds numbers ($\text{Re}<<1$), the flow is symmetric top to bottom; a slight deviation exists front to back due to the development of a boundary layer.
As the Reynolds number increases, the inertial forces begin to overcome the viscous forces and the boundary layer separates, forming two symmetric, counter-rotating vortices behind the cylinder, but the flow is still steady.
Further increasing the Reynolds number above $\text{Re}\approx60$ results in unstable, but highly repeatable, asymmetric vortex oscillations that then convect away from the cylinder in what is referred to as the von K\'arm\'an vortex street (see Schlichting-1979 \cite{schlichting1961boundary}).
This instability is magnified as the Reynolds number is further increased, and the onset of oscillations become more sensitive to small perturbations.
Physically, asymmetry is always present due to slight variations in shape or surface roughness, etc.; numerically, even small rounding or truncation errors can be sufficient to induce this asymmetric vortex shedding.

The diffuse interface method is used to simulate the cylinder example just described with a Reynolds number of $\mathrm{Re}=240$ (\cref{fig:cylinder_alamo}).
The order parameter is set to unity everywhere except for the cylinder region and the diffuse boundary surrounding it; adaptive mesh refinement ensures that the diffuse boundary is suitably refined without incurring excess resolution everywhere in the domain (\cref{fig:cylinder_mesh}).
At steady state (\cref{fig:cylinder_vortex}), the vorticity develops qualitatively in the same pattern as previously shown.
The term "steady state" may be misleading in the present context as it is not referring to a flow that is constant in time since the vortex shedding is inherently time dependent; rather, it is referring to the periodic behavior that is consistent and repeatable over time.


The momentum source terms in the $x$ and $y$ directions (\cref{fig:cylinder_source}) indicate flow separation occurs at a point ahead of the top and bottom of the cylinder---as expected.
For instance, the change in sign of the $x$-momentum source term is representative of flow reversal along the wall due to the flow separation.

\begin{figure}
  \begin{subfigure}{\linewidth}
    \includegraphics[height=3.7cm]{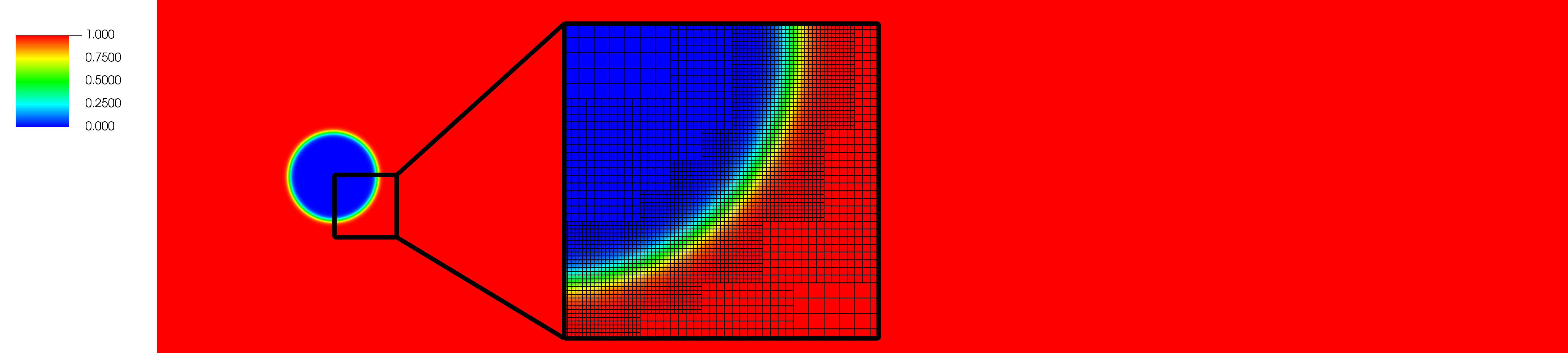}  
    \caption{Order parameter $\eta$; popout shows mesh refinement}
    \label{fig:cylinder_mesh}
  \end{subfigure}
  \begin{subfigure}{\linewidth}
    \includegraphics[height=3.7cm]{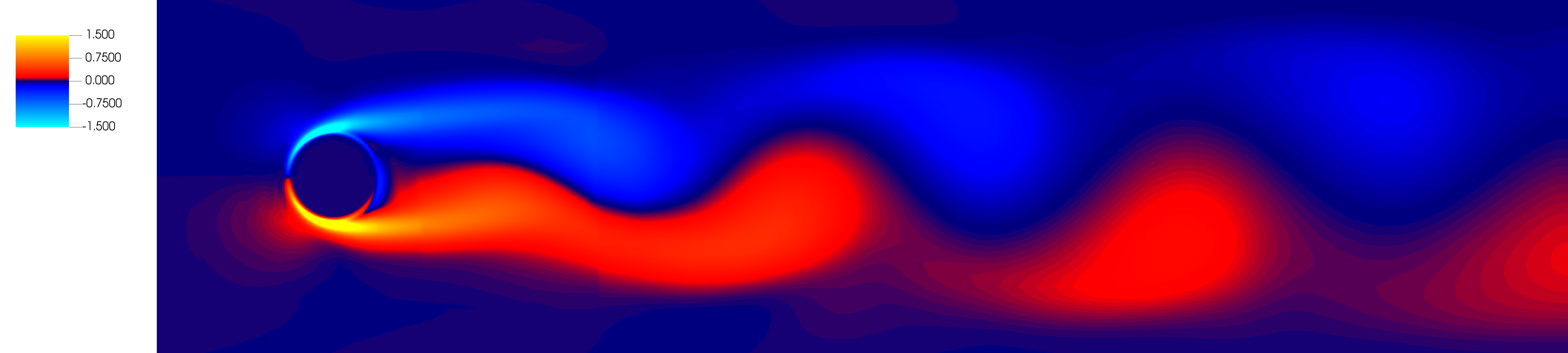}  
    \caption{Vorticity $\omega$}
    \label{fig:cylinder_vortex}
  \end{subfigure}
  \begin{subfigure}{\linewidth}
    \includegraphics[height=3.7cm, clip, trim=0cm 0cm 96cm 0cm]{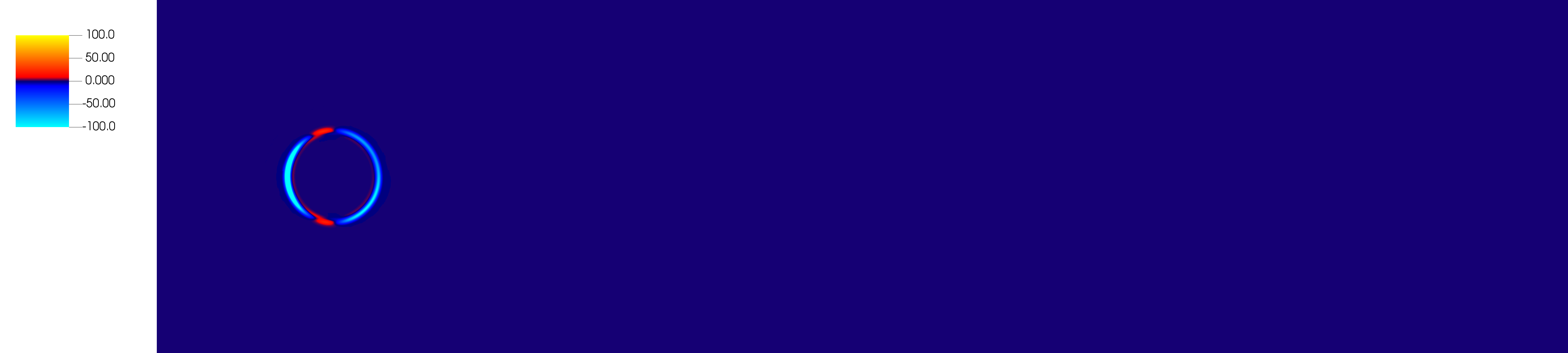}%
    \includegraphics[height=3.7cm, clip, trim=20cm 0cm 96cm 0cm]{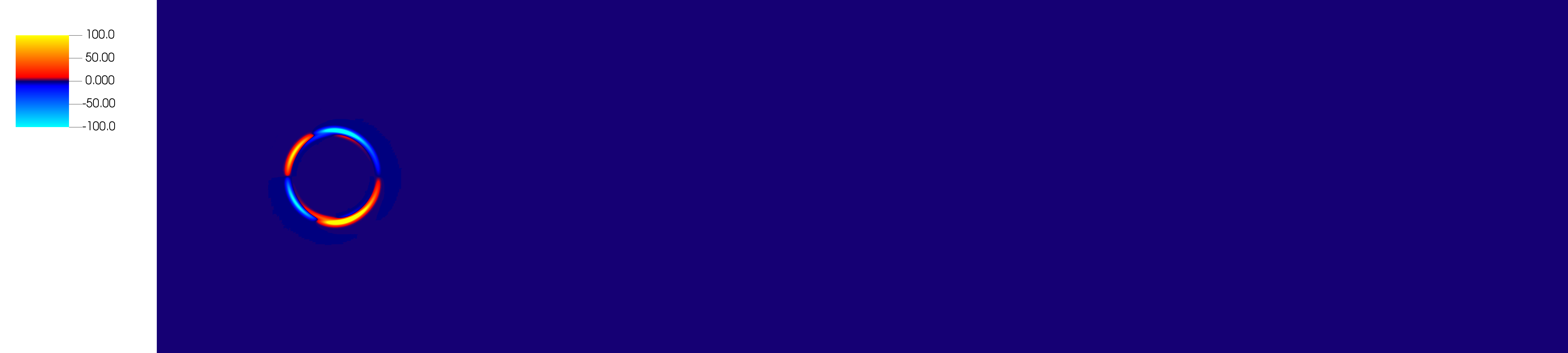}
    \caption{Momentum source terms in x direction (left) and y direction (right)}
    \label{fig:cylinder_source}
  \end{subfigure}
  \caption{Vortex shedding for flow around a cylinder.}
  \label{fig:cylinder_alamo}
\end{figure}

\subsection{Dendrite growth in forced flow}

\begin{figure} 
  \newlength{\dendritefigheight}%
  \setlength{\dendritefigheight}{2.6cm}%
  \begin{subfigure}{0.33\linewidth}
    \includegraphics[height=\dendritefigheight,clip,trim=4.1cm 0 0 0]{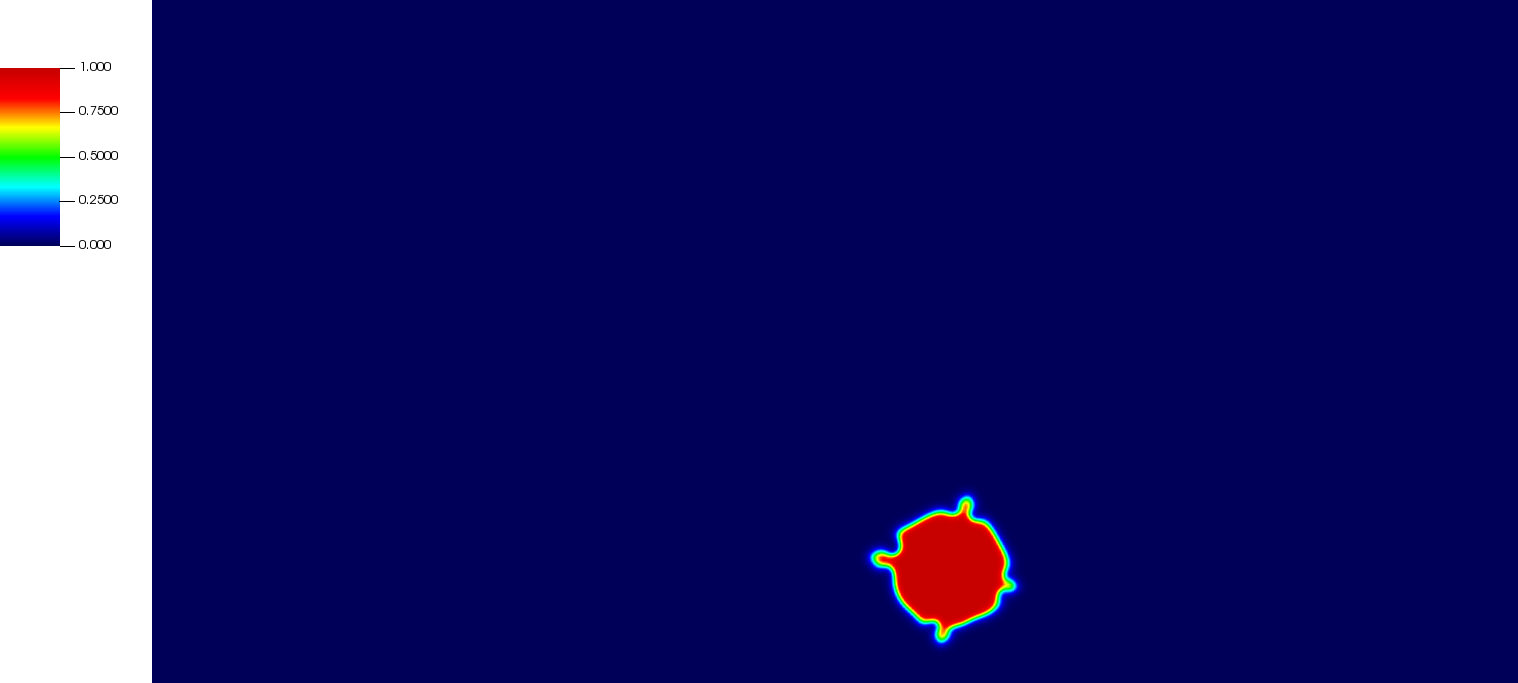}%
    \hspace{-5cm}%
    \includegraphics[height=2.5cm]{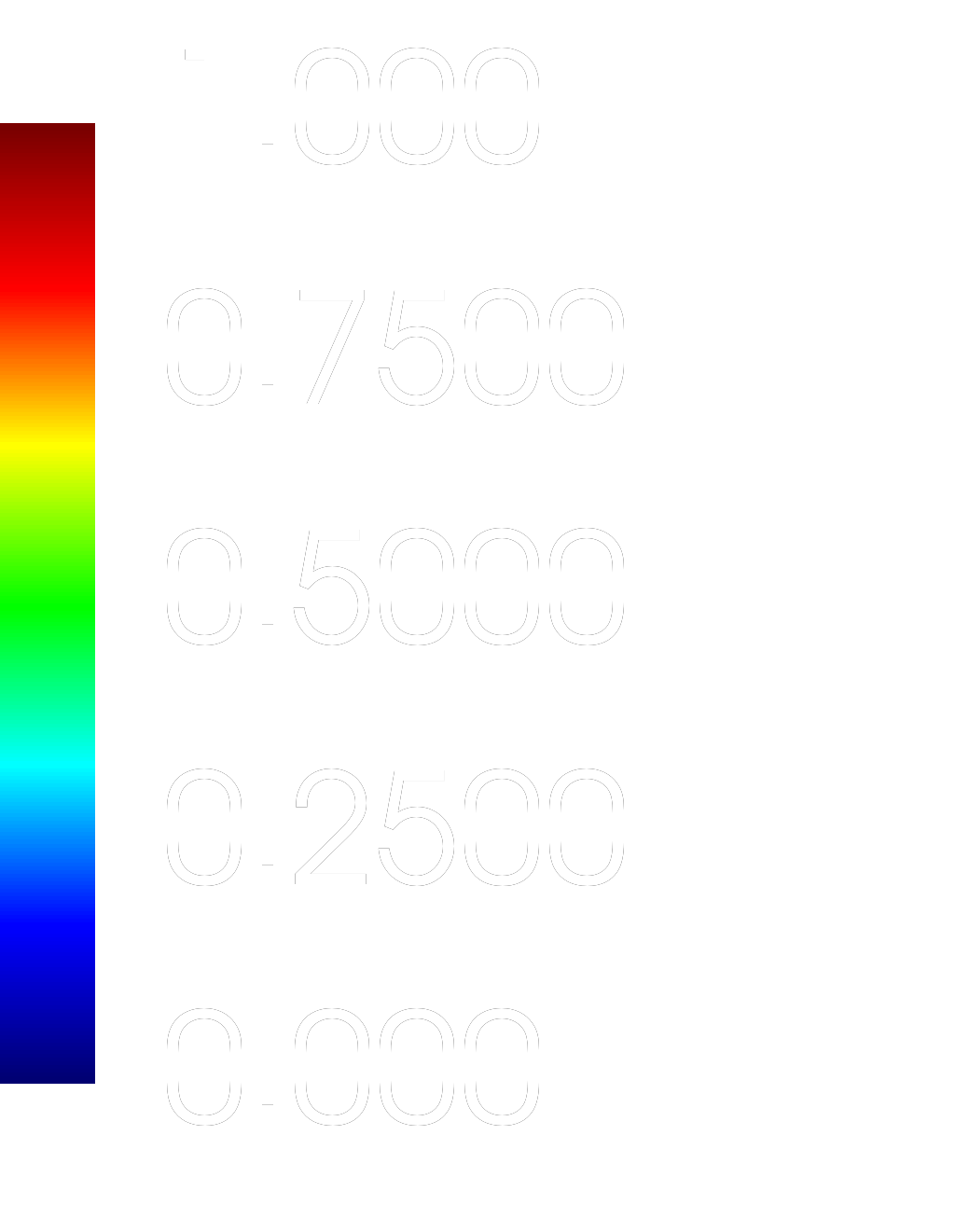}
    \includegraphics[height=\dendritefigheight,clip,trim=4.1cm 0 0 0]{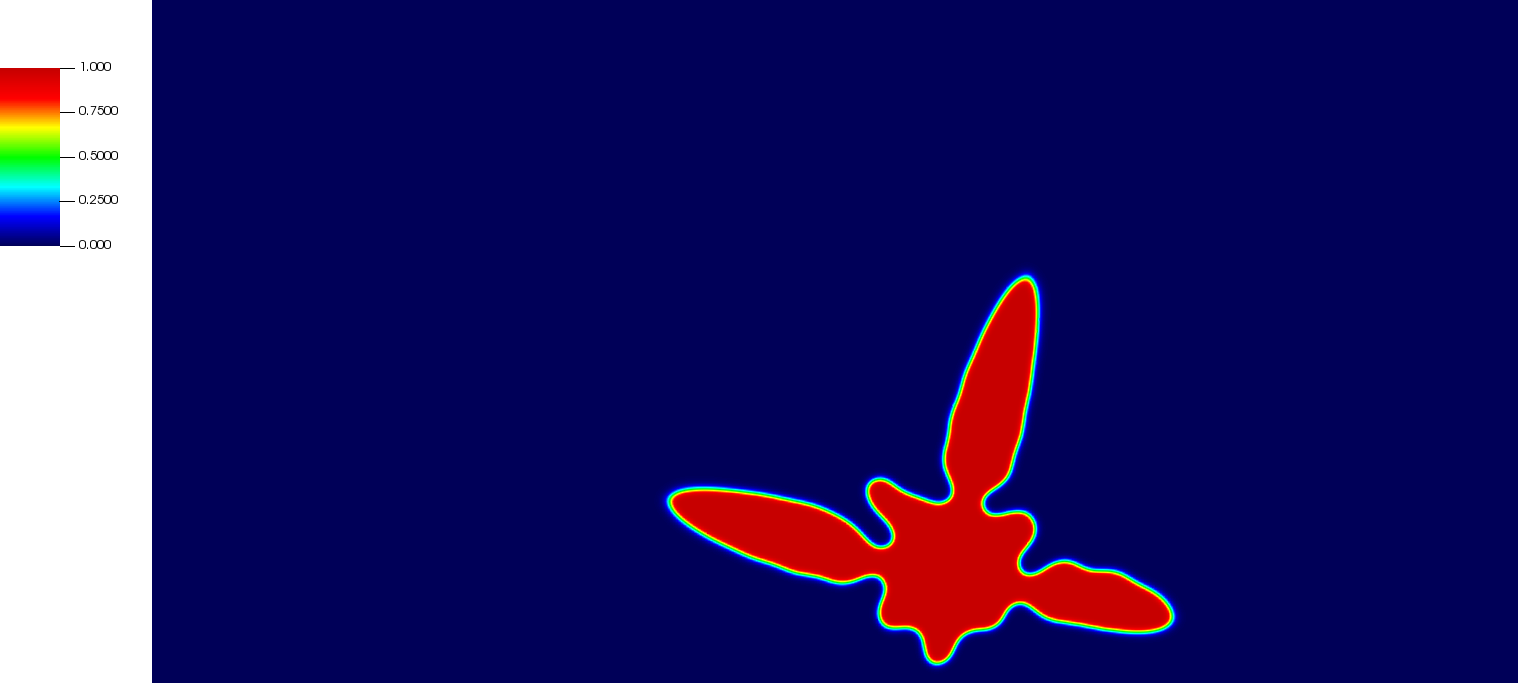}
    \includegraphics[height=\dendritefigheight,clip,trim=4.1cm 0 0 0]{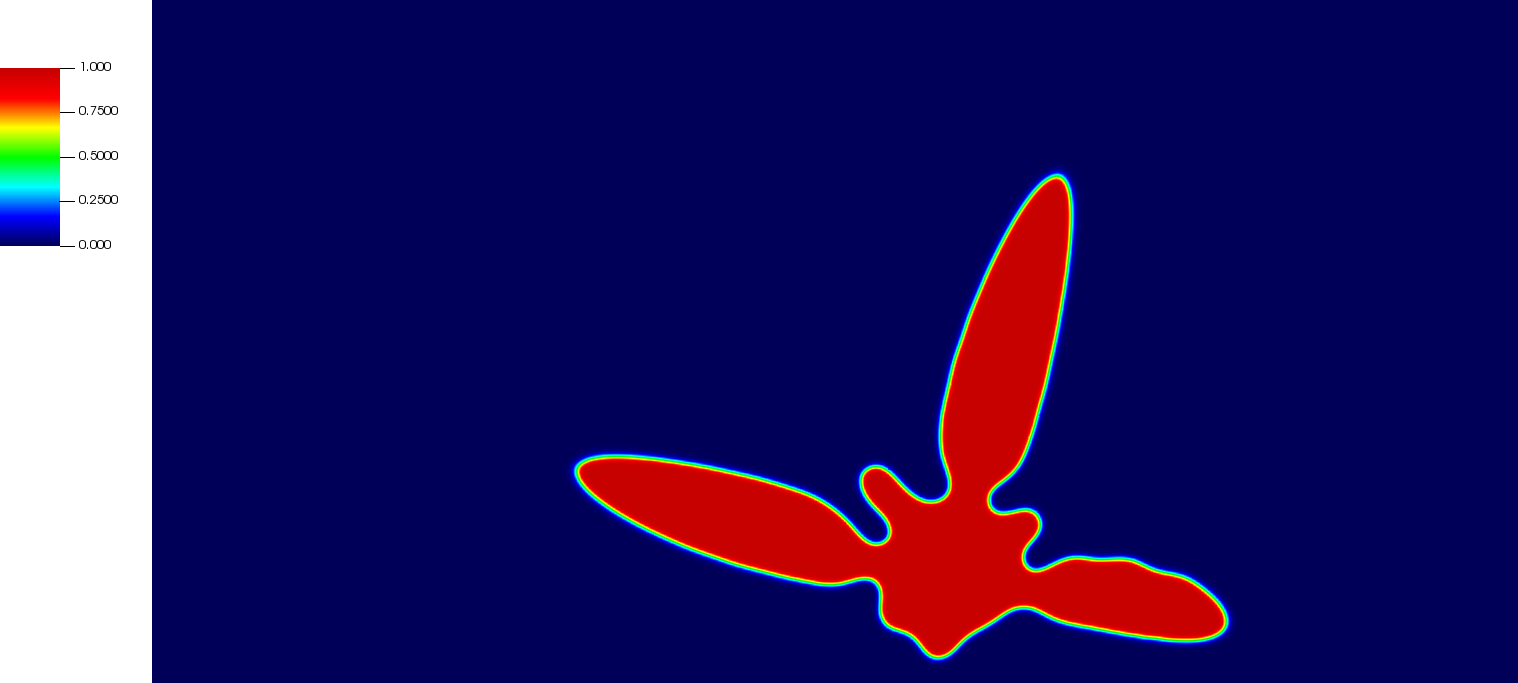}
    \label{fig:dendriteEta}
    \caption{Order parameter \\ }
  \end{subfigure}%
  \begin{subfigure}{0.33\linewidth}
    \includegraphics[height=\dendritefigheight,clip,trim=4cm 0 0 0]{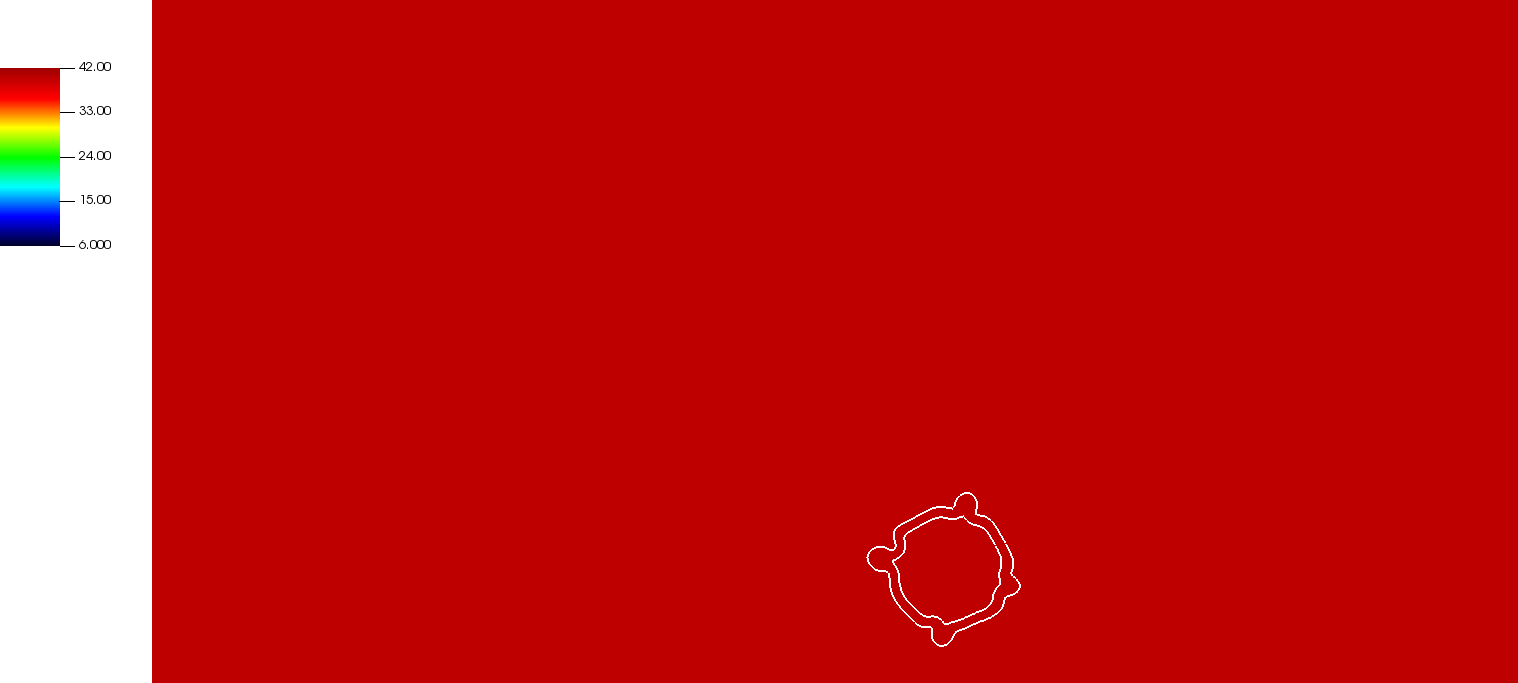}%
    \hspace{-5cm}%
    \includegraphics[height=2.5cm]{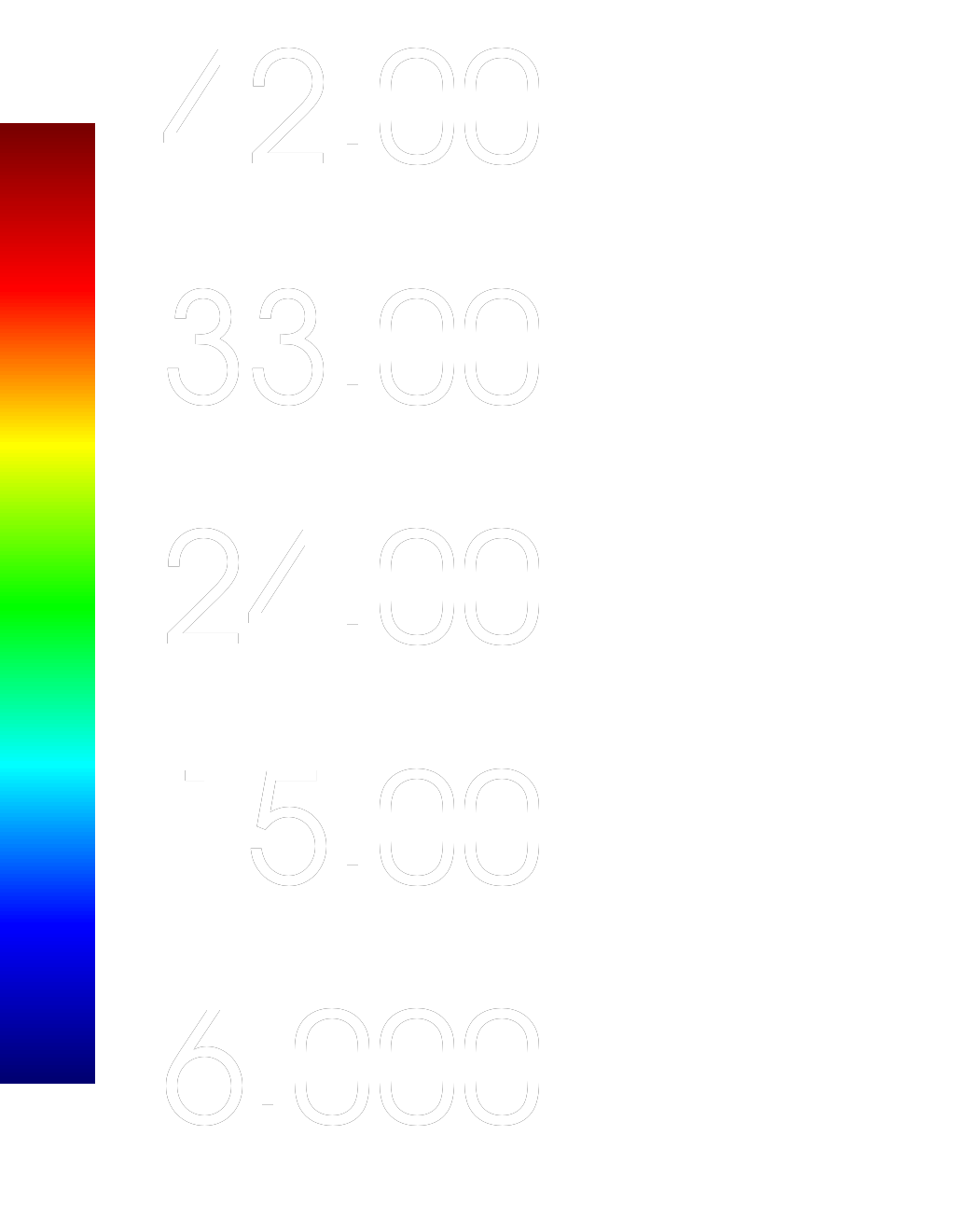}
    \includegraphics[height=\dendritefigheight,clip,trim=4cm 0 0 0]{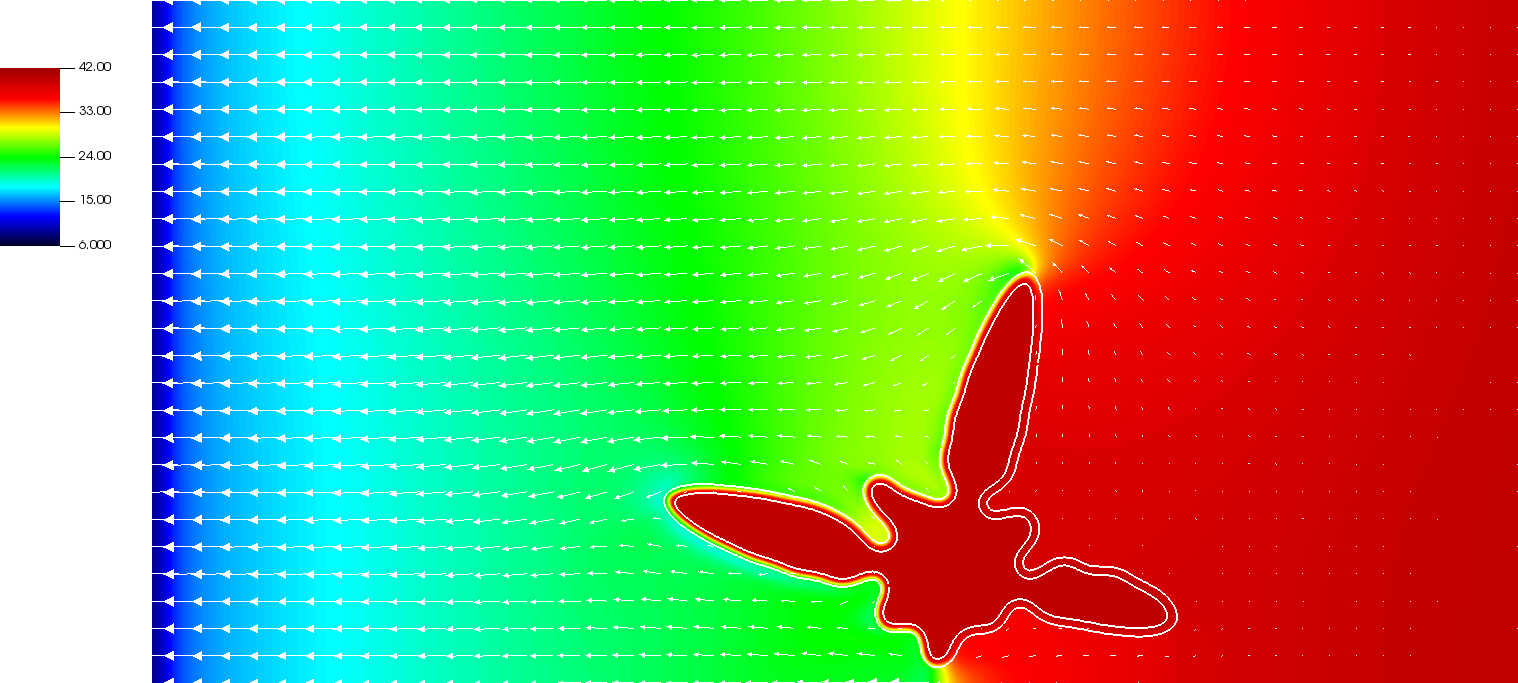}
    \includegraphics[height=\dendritefigheight,clip,trim=4cm 0 0 0]{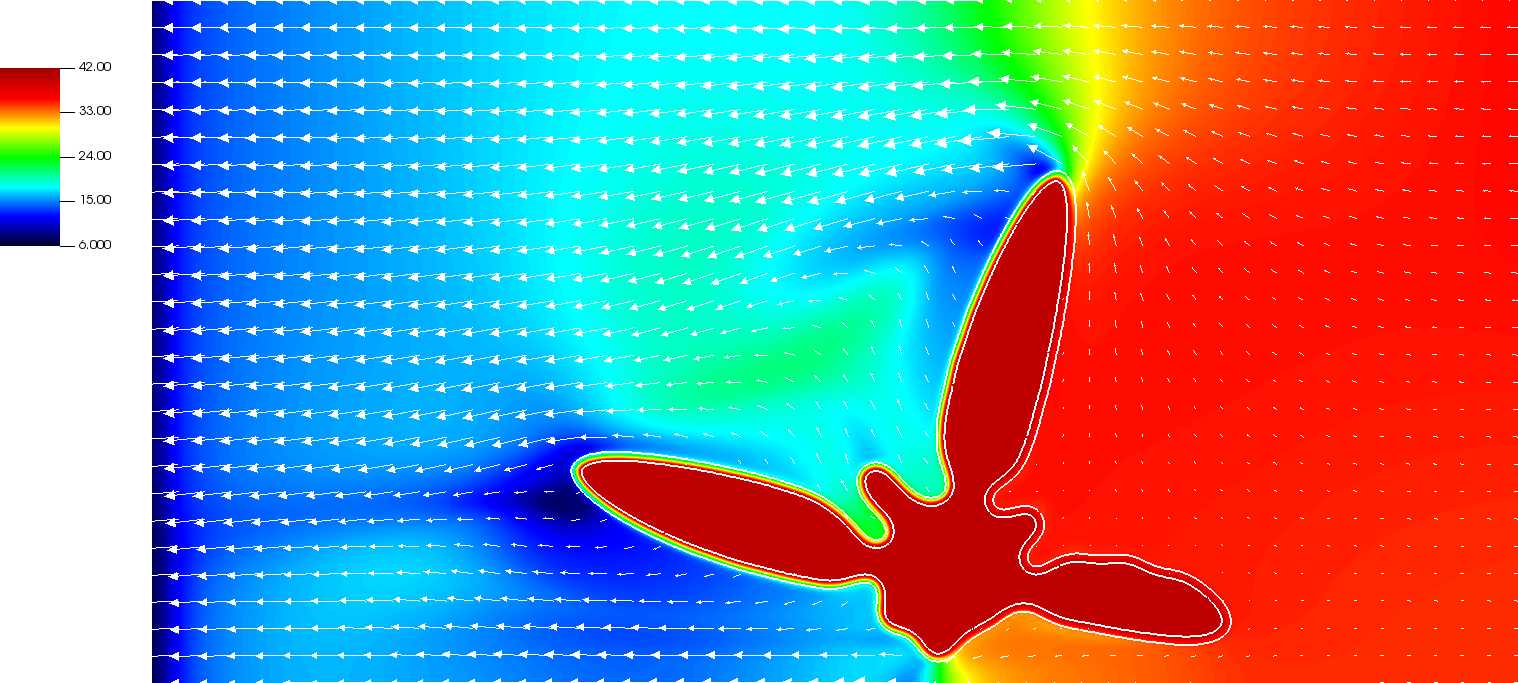}
    \caption{Fluid density and velocity}
    \label{fig:dendriteVelocity}
  \end{subfigure}
  \begin{subfigure}{0.33\linewidth}
    \includegraphics[height=\dendritefigheight,clip,trim=4cm 0 0 0]{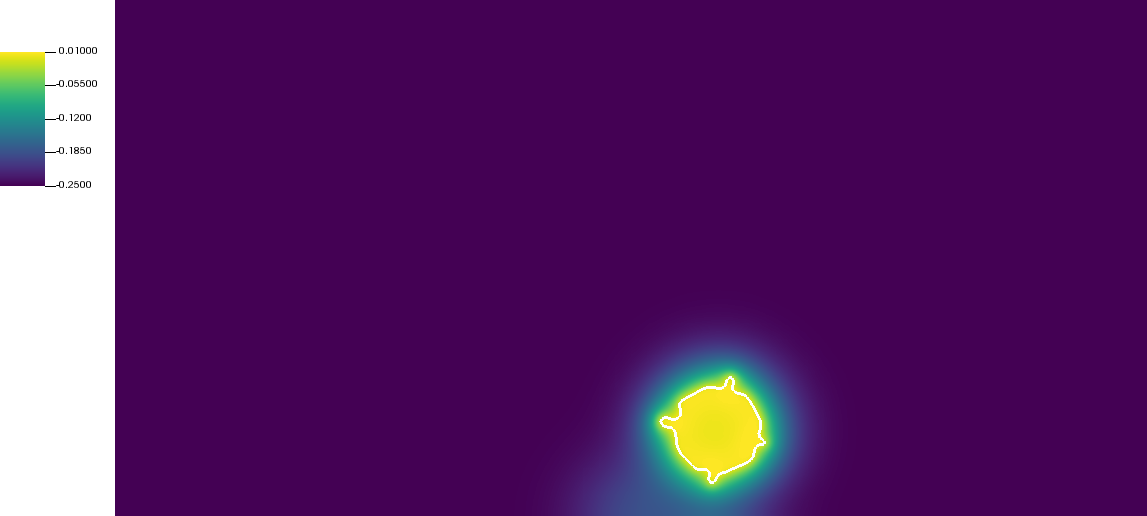}%
    \hspace{-5cm}%
    \includegraphics[height=2.5cm]{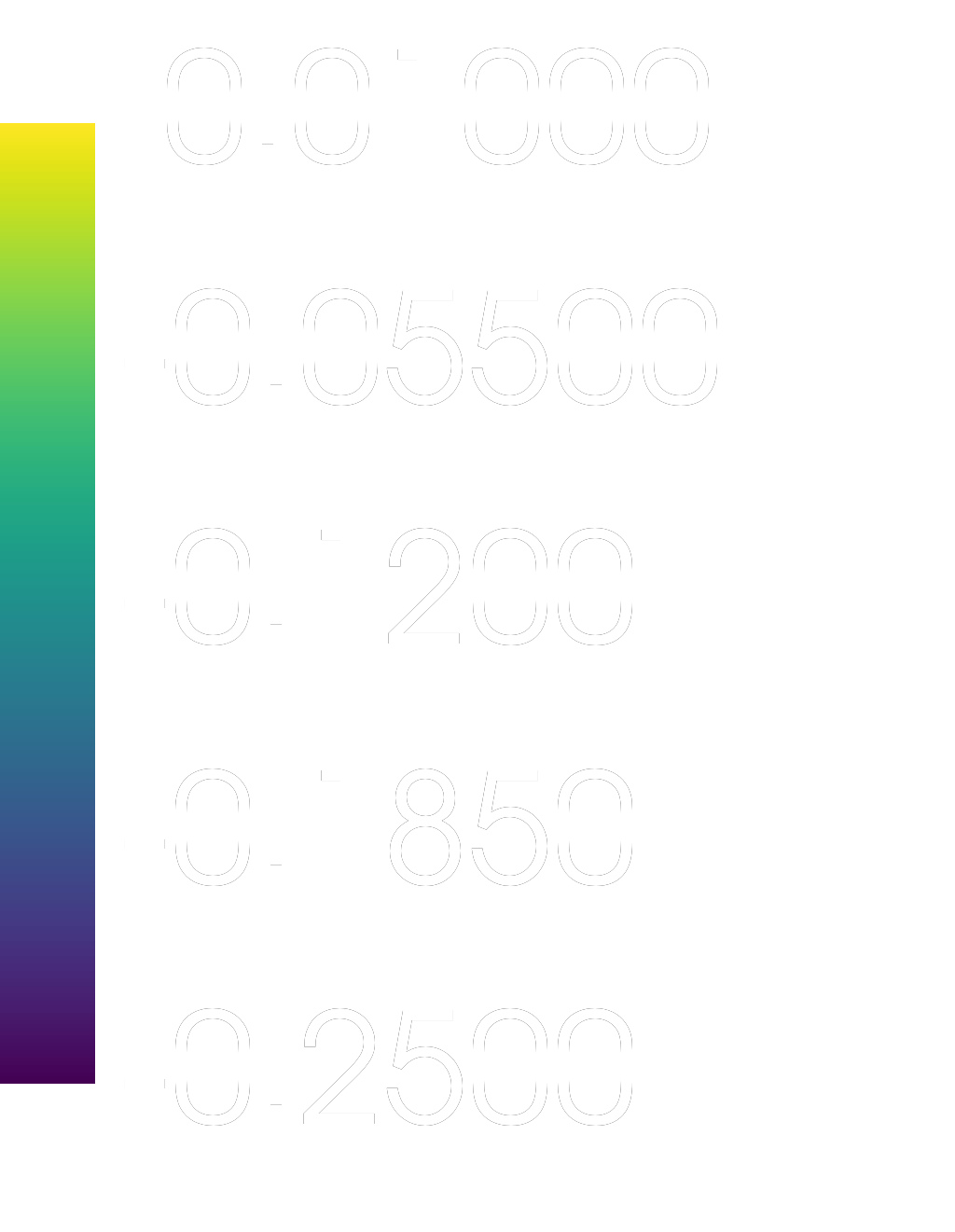}
    
    \includegraphics[height=\dendritefigheight,clip,trim=4cm 0 0 0]{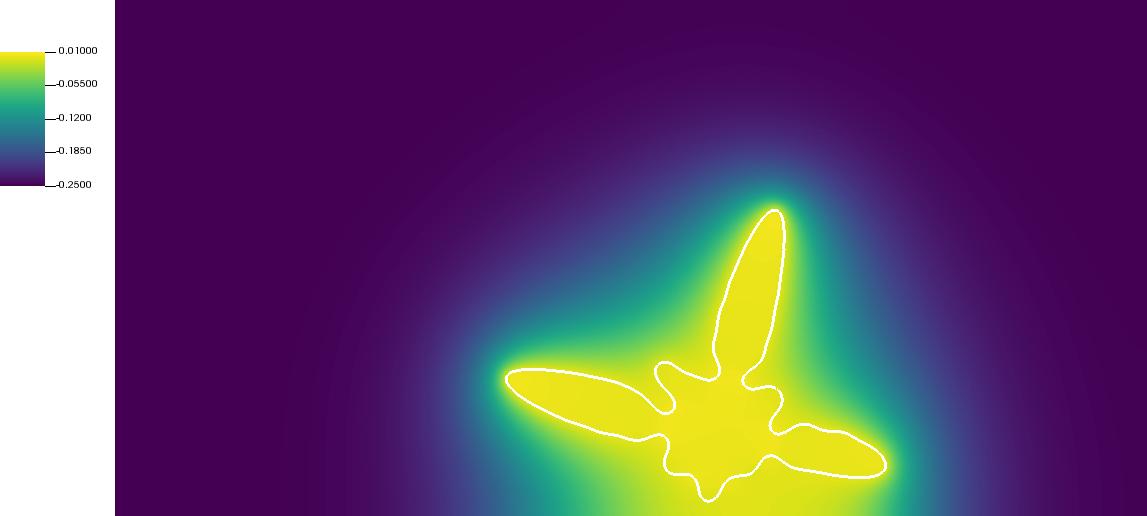}
    \includegraphics[height=\dendritefigheight,clip,trim=4cm 0 0 0]{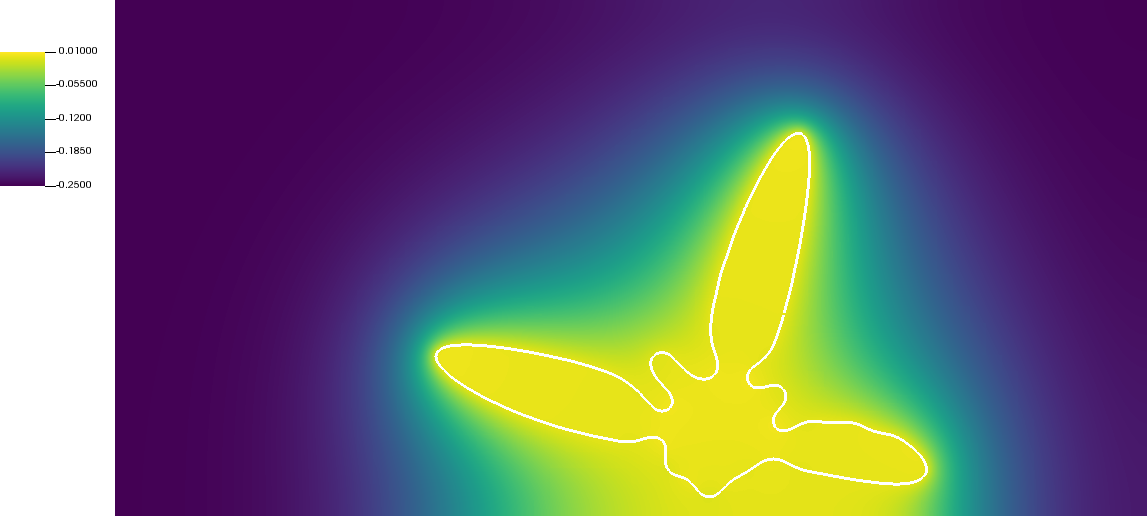}
    \caption{Phase field temperature}
    \label{fig:dendriteTemperature}
  \end{subfigure}
  \caption{Flow around a growing dendrite}
  \label{fig:dendrite}
\end{figure}

The diffuse boundary method's power lies in its ability to treat highly complex boundaries implicitly, with no regard needed for complex boundary descriptions or numerically tricky events such as topological transitions.
Numerous physical domains exist across evolving solid\slash fluid domains, in which neither the evolution of the interface, the flow field, nor the coupling between the two can (or should) be reduced to low order.
An important example is the growth of dendrite structures, which is of paramount importance in solidification, phase transformation, and materials science.
Dendrite growth is well-known to be coupled to the ambient flow conditions and is (by definition) a complex boundary-driven phenomenon.

The conditions for dendritic growth are directly associated with the variation of the Gibbs free energy in the system. 
This gradient energy at the interface can be generated by different mechanics, such as undercooling during solidification \cite{tiller1953redistribution}, and composition gradient in both multi-component and single-component systems, where diffusion can occur between the different species or due to density differences \cite{strickland2020directional,sun2009effect}. 
Thus, in the presence of a viscous flow field, the temperature gradients due to convection and the composition change due to density variation have major roles in the interface energy and, therefore, control the direction and velocity of dendritic growth.

The applications of current models and robust flow solvers for the practical study of dendritic growth are extensive.
The evolution of dendrites during the solidification process was the original motivation for the development of phase-field modeling.
Simulations using forced flow field that account for momentum kinetics have been previously performed, but they are typically very general and fail to assert the convergence of the solution as the diffuse interface length approached the sharp interface limit \cite{tong2001phase, lu2005three}. 
Other sharp interface methods, such as the cellular automata method, have been used to account for mass and momentum transfer between the phases, but they face the challenge of explicitly tracking the interface evolution that relies on strong assumptions over the interface geometry and they have also been limited to incompressible flow cases \cite{yuan2009numerical, yuan2010dendritic, tonhardt2000dendritic}. 
Simulations of both 2D and 3D samples for a viscous flow case using phase-field modeling have been presented by Jeong et al.~\cite{jeong2001phase} using a simplified Cahn-Hillard model for constant mobility, but the fluid solved relied on an averaging solution presented in \cite{beckermann1999modeling}.
Jaafar et al.~present an in-depth review of methods used in the literature, including front-tracking methods, enthalpy methods, fluid volume methods, and level-set methods \cite{jaafar2017review}.
Zhang et al.~have presented a detailed review of the use of phase-field models for dendritic growth, including the use of sophisticated techniques, such as adaptive mesh refinement \cite{zhang2023numerical}.

To demonstrate the capacity and flexibility of the model presented in this paper, a simple dendritic growth model is implemented, following the work presented in \cite{kobayashi1994numerical}.
In order to couple the dendritic evolution with the flow, we start by introducing a new order parameter, $\alpha = 1 - \eta$, so that information can be transferred between the different evolution equations in every timestep, with $\alpha = 1$ representing the solid phase and $\alpha = 0$ representing the fluid phase.
Then the free energy functional for the dendritic model can be written as follows: 
\begin{align}
\label{eq:dendriteFunctional}
  W = \int_\Omega \Big[ \frac{1}{2}\varepsilon^2 |\nabla\alpha|^2
  + \psi
  \Big]\,d\bm{x},
\end{align}
where $\epsilon$ is the interface energy, and the $\psi$ function is a double-well potential, as described in \cite{kobayashi1992simulations}, with local minima at $\alpha = 0$ and $\alpha = 1$.
The kinetic evolution of $\alpha$ can be obtained from the gradient of equation \ref{eq:dendriteFunctional}, through a traditional Allen-Cahn approach and take the following form:
\begin{equation}
\label{eq:dendriteEvolution}
   \tau \frac{\partial \alpha }{\partial t} = -\frac{\partial W}{\partial \alpha} = \varepsilon^2 \Delta\alpha + \alpha(1 - \alpha) (\alpha - \frac{1}{2} +m),
\end{equation}
where $\tau$ is a scaling parameter, and $ m $ is the anisotropy barrier, which is a function of temperature, as described in \cite{kobayashi1994numerical}.
Thus, a diffusion equation is also required to evaluate the temperature evolution, and is implemented as
\begin{equation}
    \label{eq:dendriteDiffusion}
    \frac{\partial T}{\partial t} = \Delta T + \frac{\partial \alpha}{\partial t}.
\end{equation}
This implementation of the heat diffusion equation does not account for any heat generation terms. 
Therefore, by initializing the system with zero flux Neumann boundary conditions for the temperature field, this field is effectively dimensionless and can be interpreted as the level of undercooling in the system, similar to what is presented in \cite{zhang2021solution} and \cite{sun2009effect}.

The simulation is performed by initializing a $12 \times 6$ dimensionless domain with a single circular inclusion with a 0.1\,radius to represent a dendritic nucleation site.
Density is initialized as a constant, with a dimensionless value of 42.0, across all phases.
The domain is set with zero flux Neumann boundary conditions at all edges for $\alpha$ and $\eta$.
The energy and density fields are initialized by setting zero flux Neumann boundary condition is the top, bottom, and right edges. 
The left edge is initialized with Dirichlet boundary condition with value 1.0. 
For the momentum field, the x-momentum is set to zero flux Neumann boundary conditions on all edges, while the y-momentum is set to zero flux Neumann boundaries on the left and right edges, and 0.0 Dirichlet conditions are enforced on the top and bottom edges.

The simulation starts as a stationary system for the fluid flow while allowing initial development of the dendrite, purely based on the four-fold anisotropy imposed through the $m$ in \cref{eq:dendriteEvolution}, as shown in the top images of \cref{fig:dendrite}.
Contour lines are plotted in white to highlight the dendrite geometry.
Then, a pressure drop is introduced to the left-hand edge to initiate the flow field evolution (\cref{fig:dendriteVelocity}).
The velocity fields can be qualitatively compared to those in \cite{sun2009effect}, as the rate of growth is favorably affected by a larger density variation. 
It also demonstrated the model's ability to simultaneously capture the complex flow directions due to the dendritic geometry change.
The undercooling directions remain relatively symmetric as no thermal forces have been added to the simulation (\Cref{fig:dendriteTemperature}).

Further investigation is required to perform a quantitative analysis of the results, which rests beyond the scope of this work. 
The emphasis here is the model's capability to properly capture two complex boundary evolution systems, flow dynamics and dendritic growth, simultaneously and with good qualitative results.
The flow is computed without incompressibility assumptions and is able to track momentum and mass conservation across coupled solvers while implicitly tracking dendritic evolution.

\subsection{Flow interaction with a diffuse Allen-Cahn evolving material} 

\begin{figure}
  \centering
    \includegraphics[height=3.6cm,clip,trim=15cm 0cm 0cm 0cm]{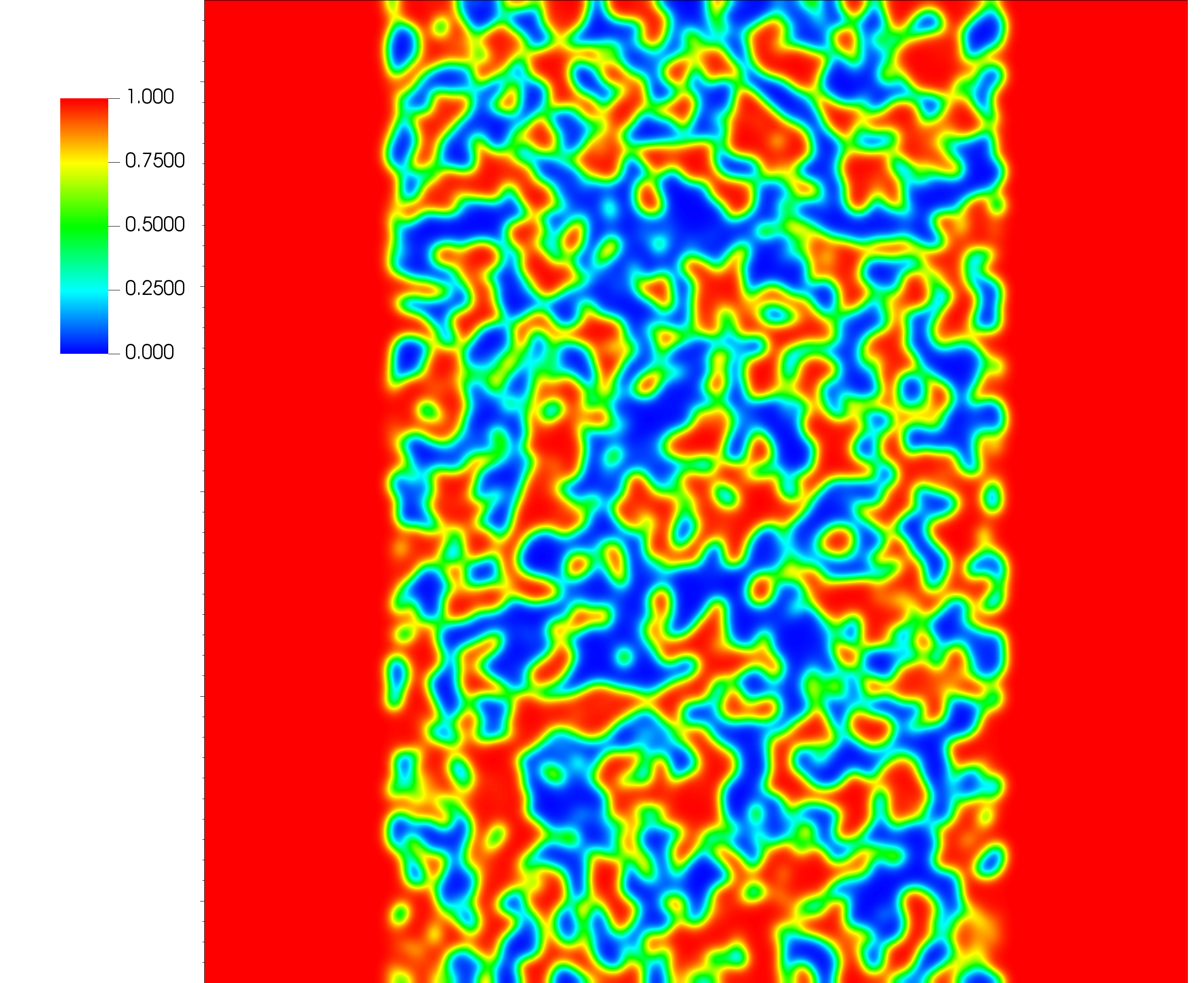}
    \includegraphics[height=3.6cm,clip,trim=15cm 0cm 0cm 0cm]{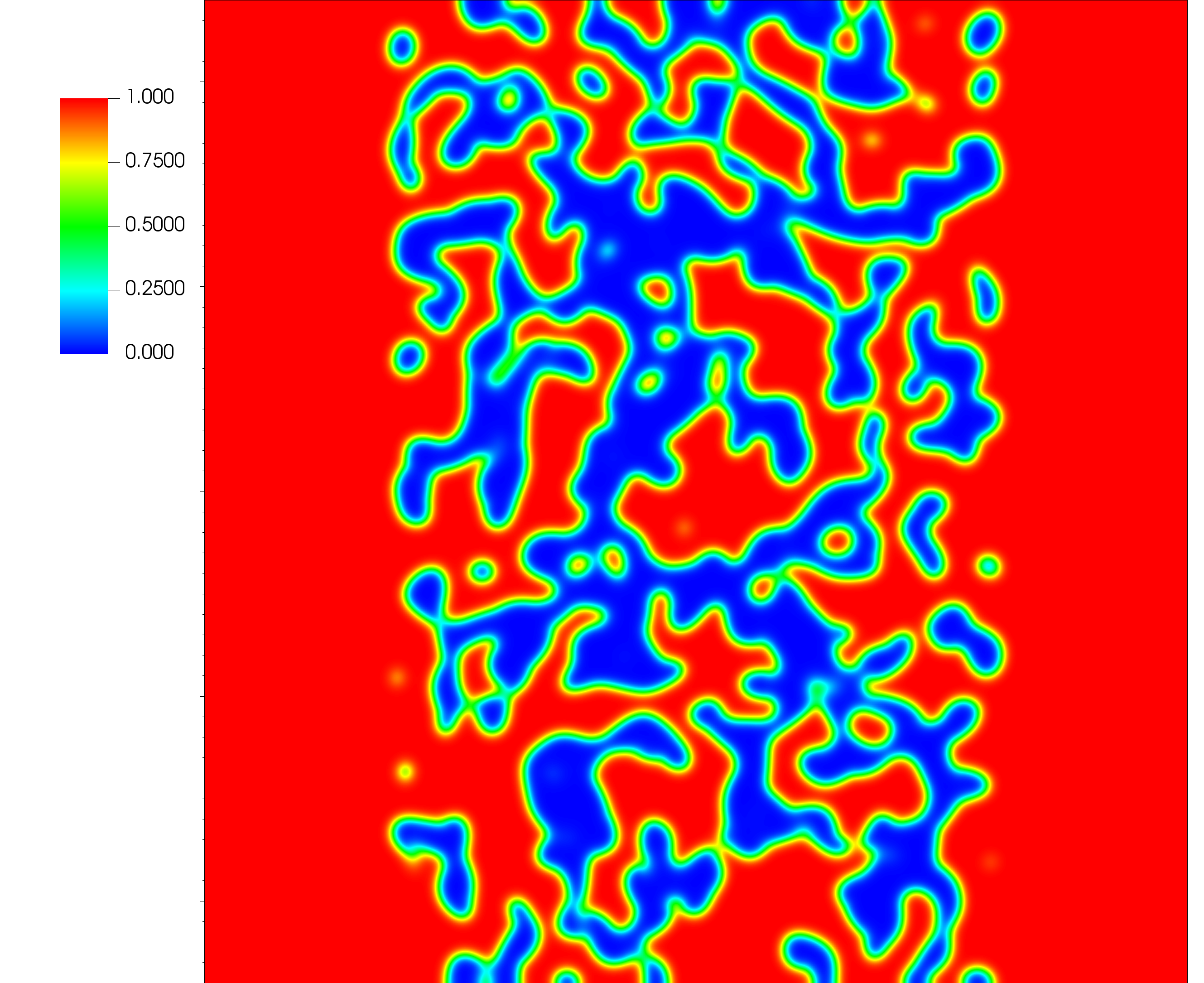}
    \includegraphics[height=3.6cm,clip,trim=15cm 0cm 0cm 0cm]{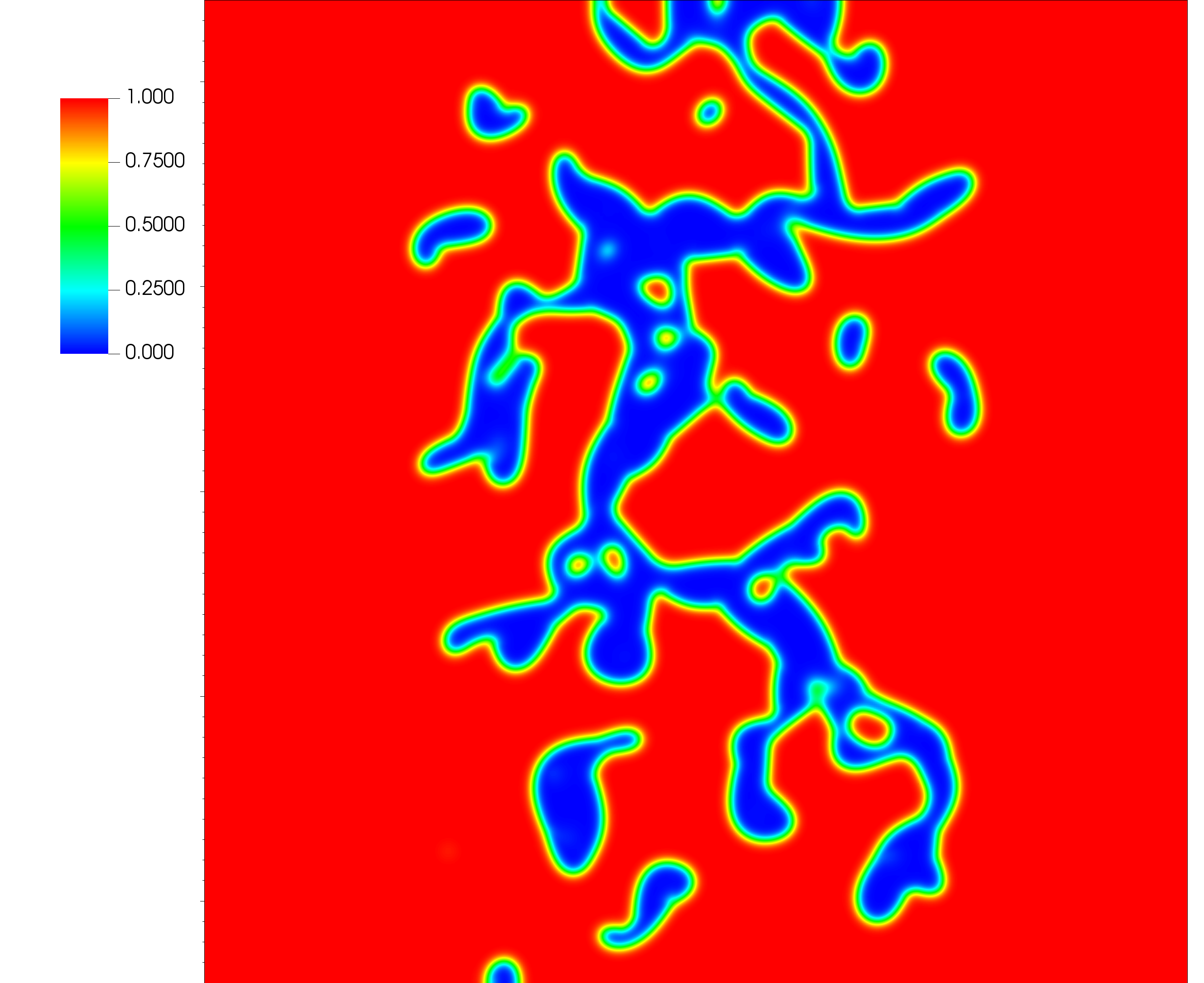}
    \includegraphics[height=3.6cm,clip,trim=15cm 0cm 0cm 0cm]{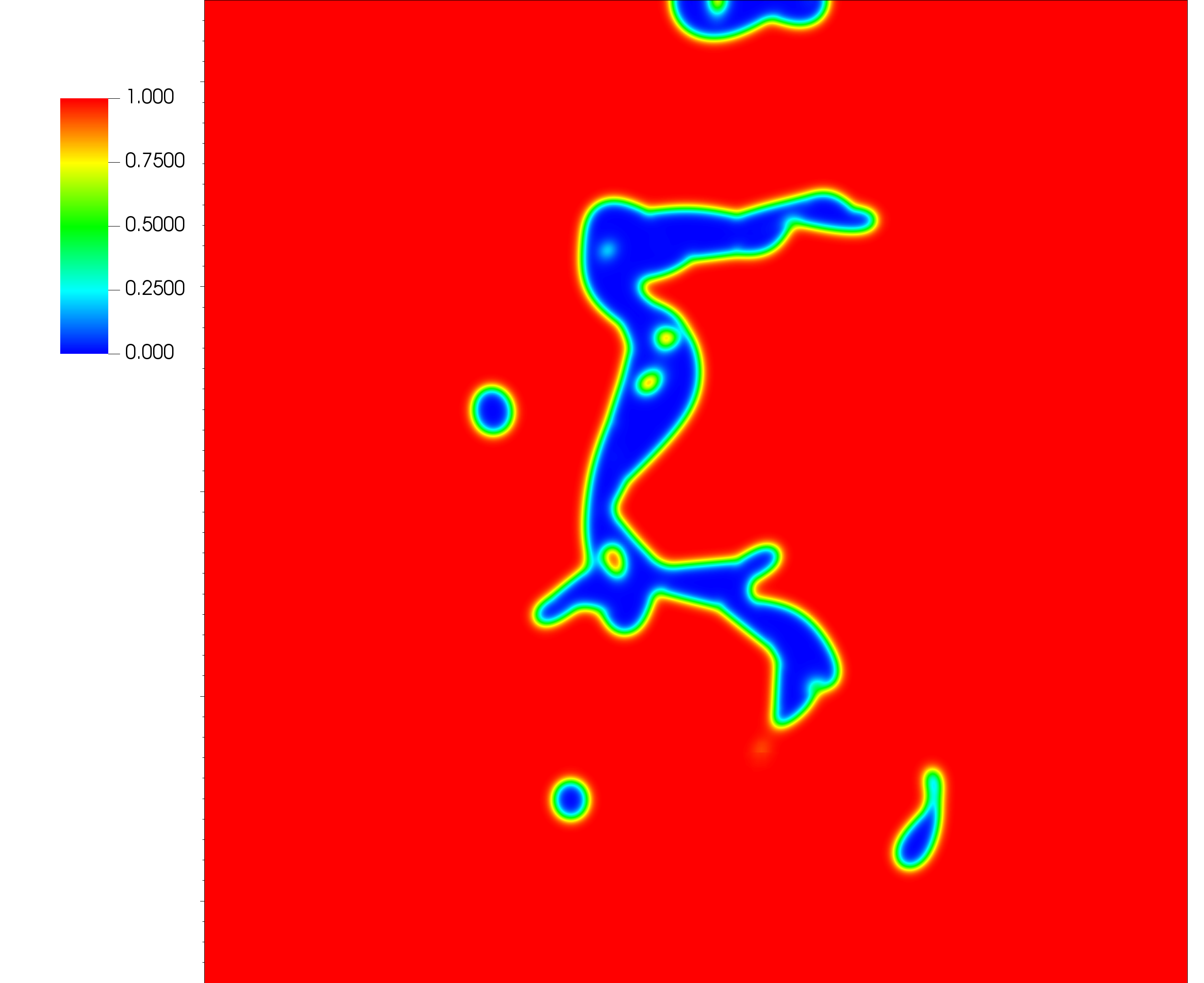}
    \caption{
    Evolution of the order parameter field; red corresponds to $\alpha=0,\eta=1$ (flow region), blue to $\alpha=1,\eta=0$ (solid region).
    $\alpha$ is seeded randomly and allowed to initially equilibrate bidirectionally, and then is evolved to zero.
    Initially, there is a full barrier between the inlet (left) and outlet; this is eventually breached as the solid region erodes away.
    }
  \label{fig:allencahn-eta}
\end{figure}

\begin{figure}
  \centering
  \begin{minipage}{0.45\linewidth}
    \includegraphics[height=3.4cm,clip,trim=15cm 40cm 0cm 0cm]{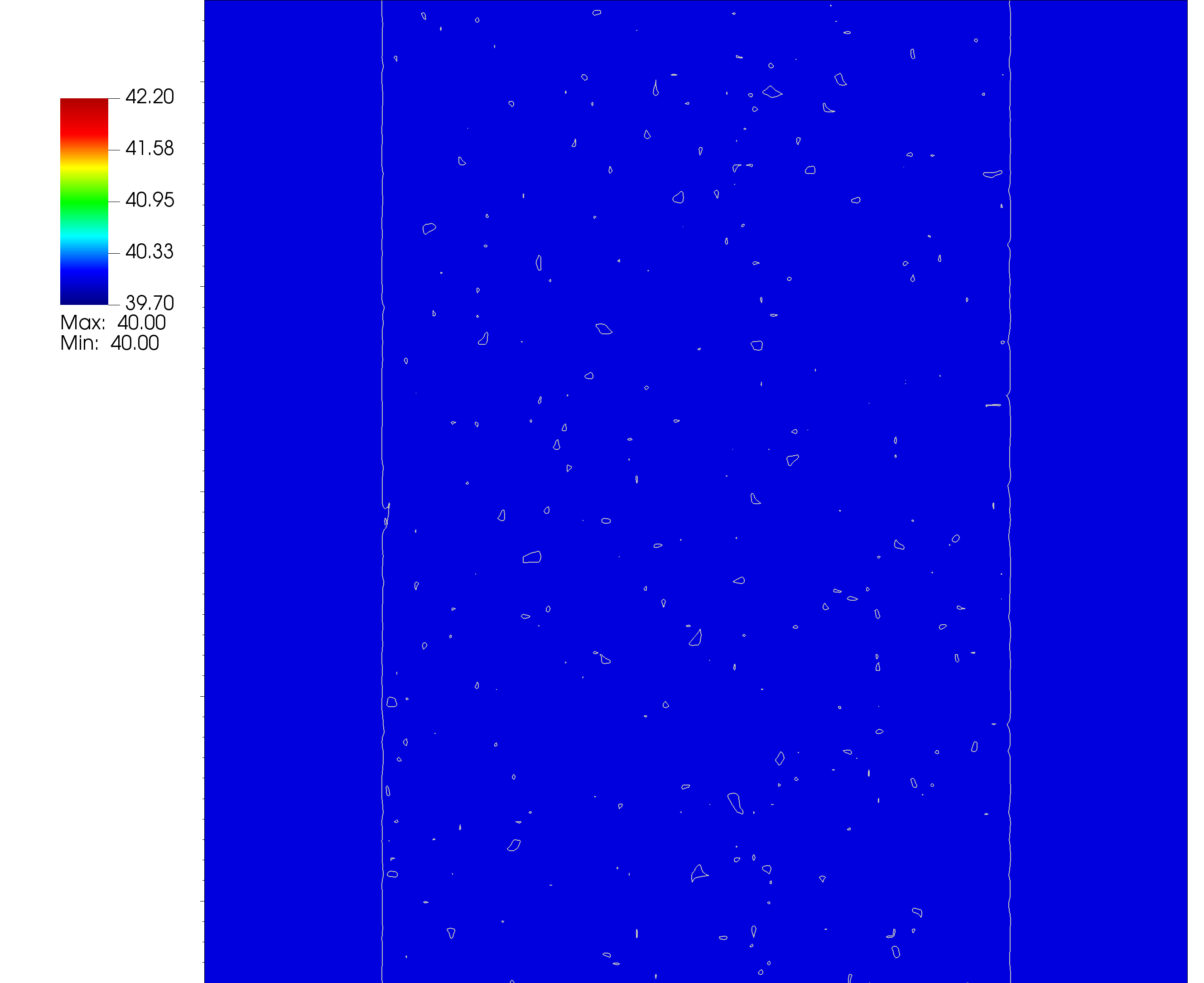}
    \includegraphics[height=3.4cm,clip,trim=15cm 40cm 0cm 0cm]{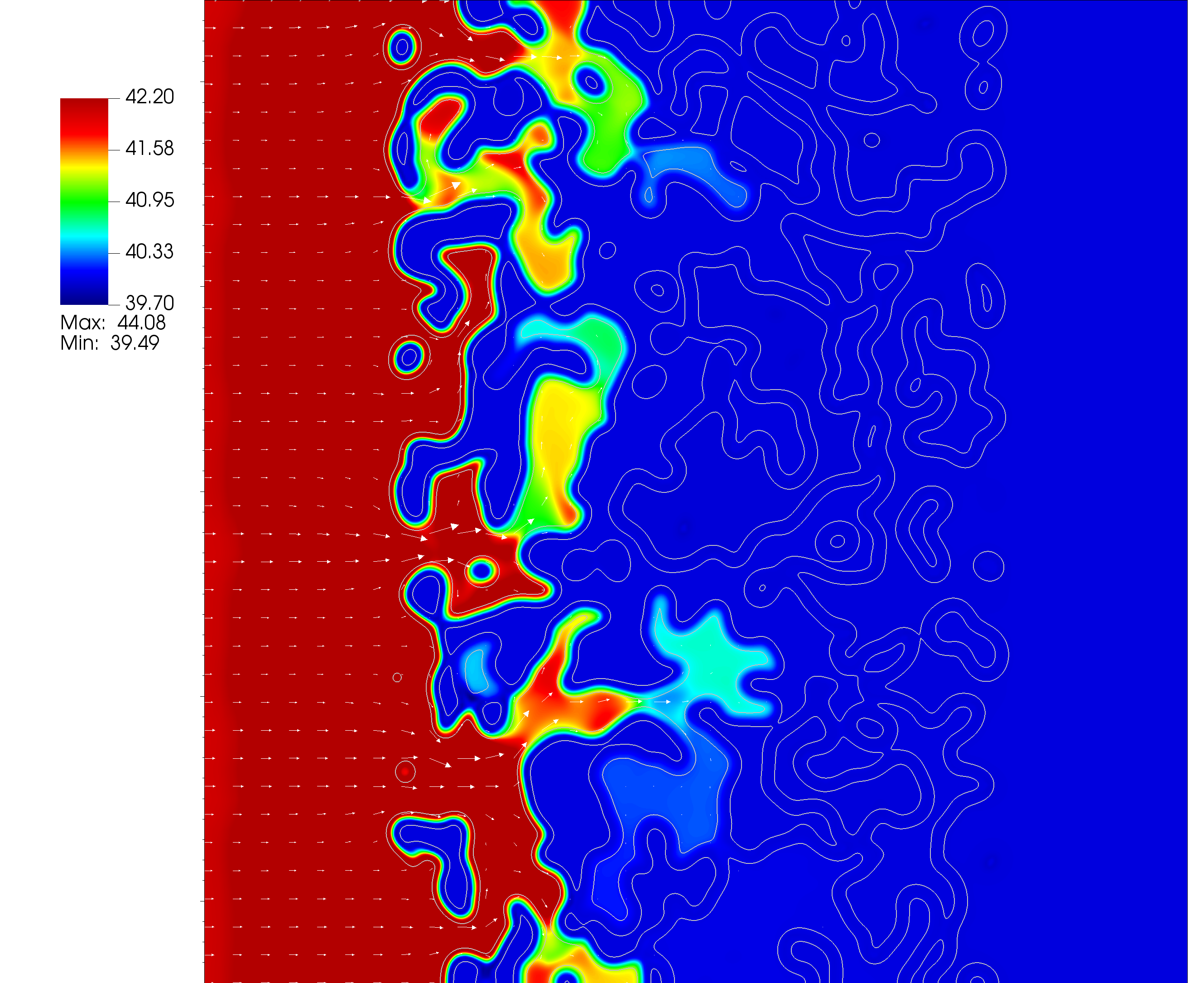}
    \includegraphics[height=3.4cm,clip,trim=15cm 40cm 0cm 0cm]{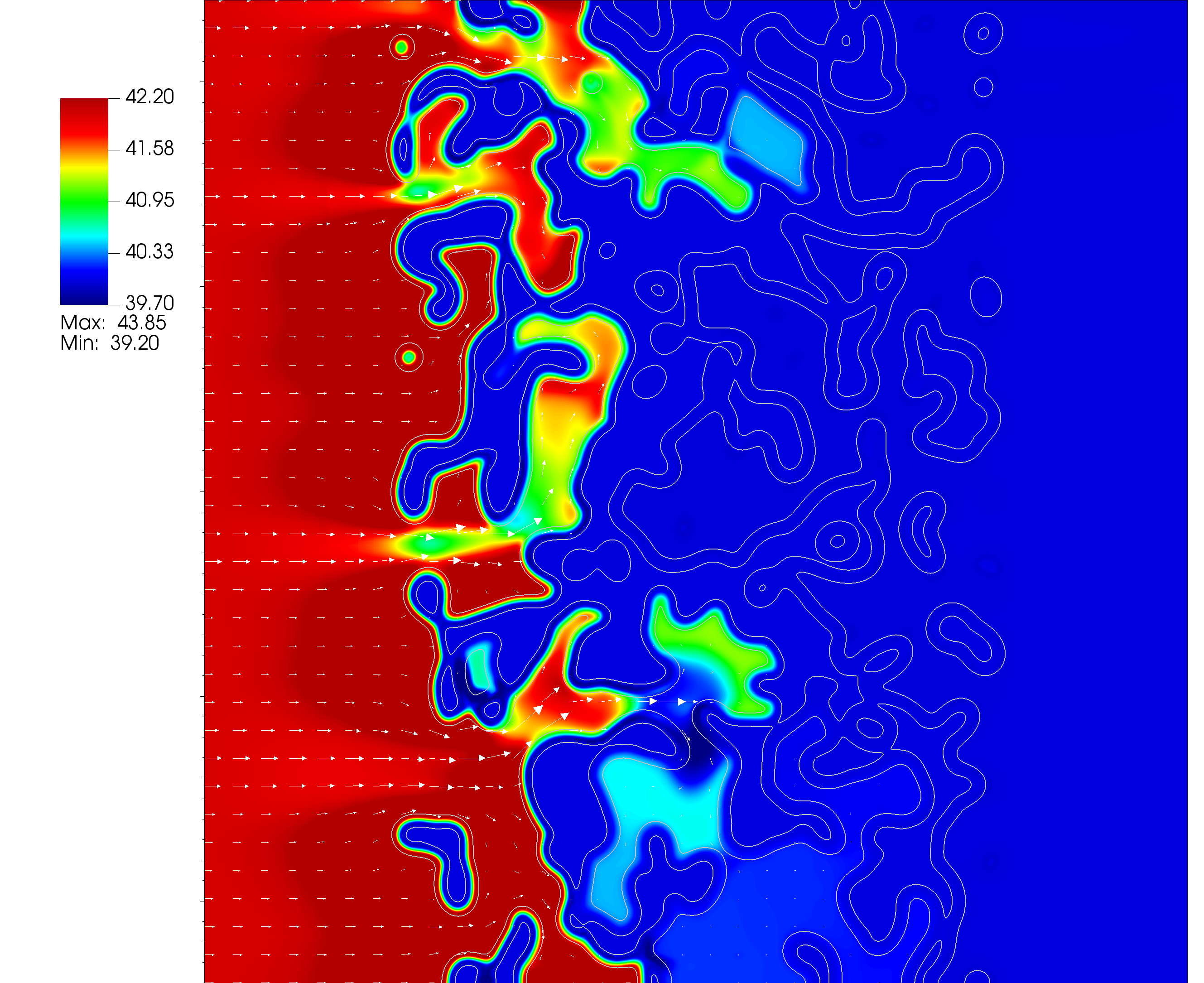}
    \includegraphics[height=3.4cm,clip,trim=15cm 40cm 0cm 0cm]{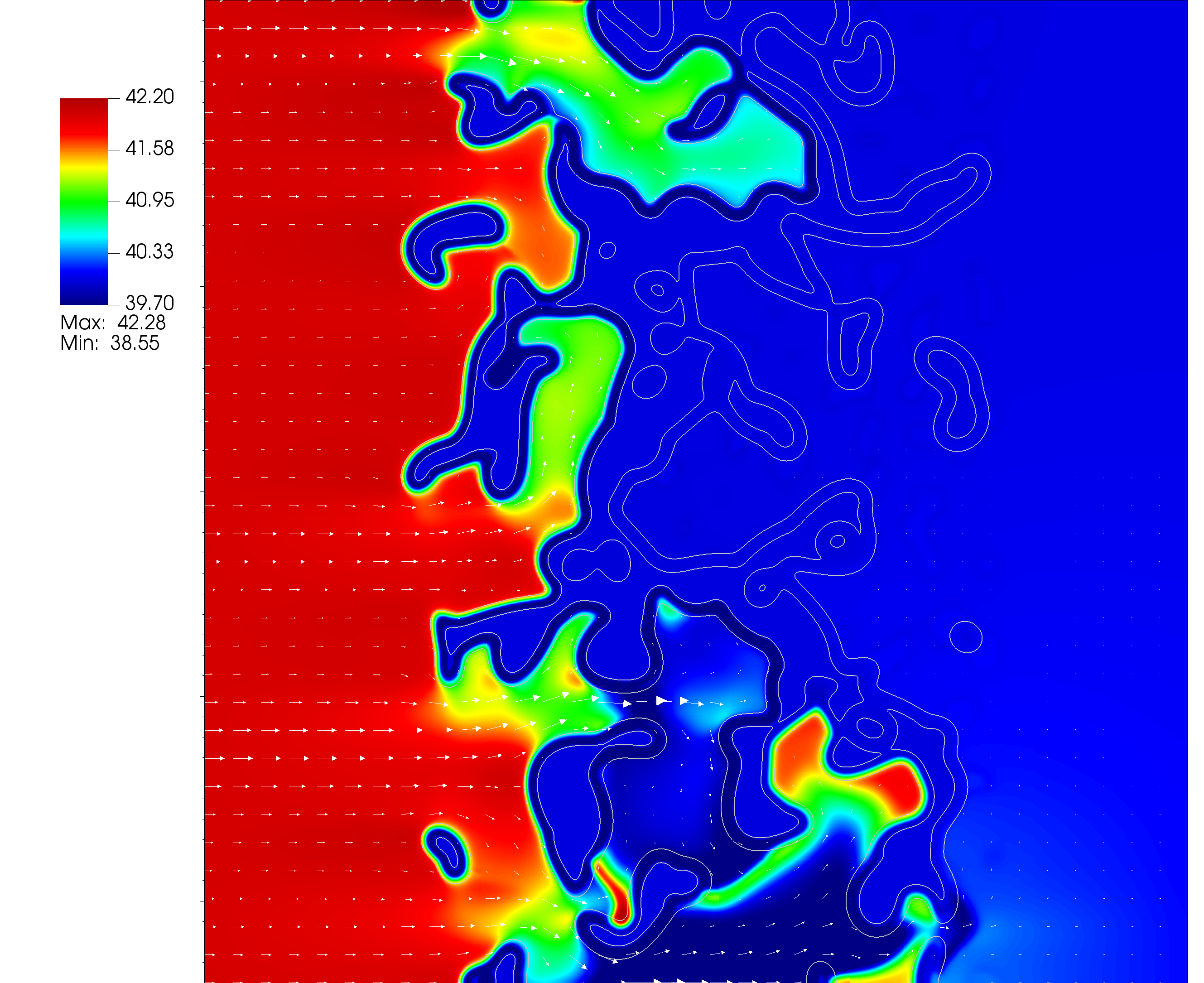}
    \includegraphics[height=3.4cm,clip,trim=15cm 40cm 0cm 0cm]{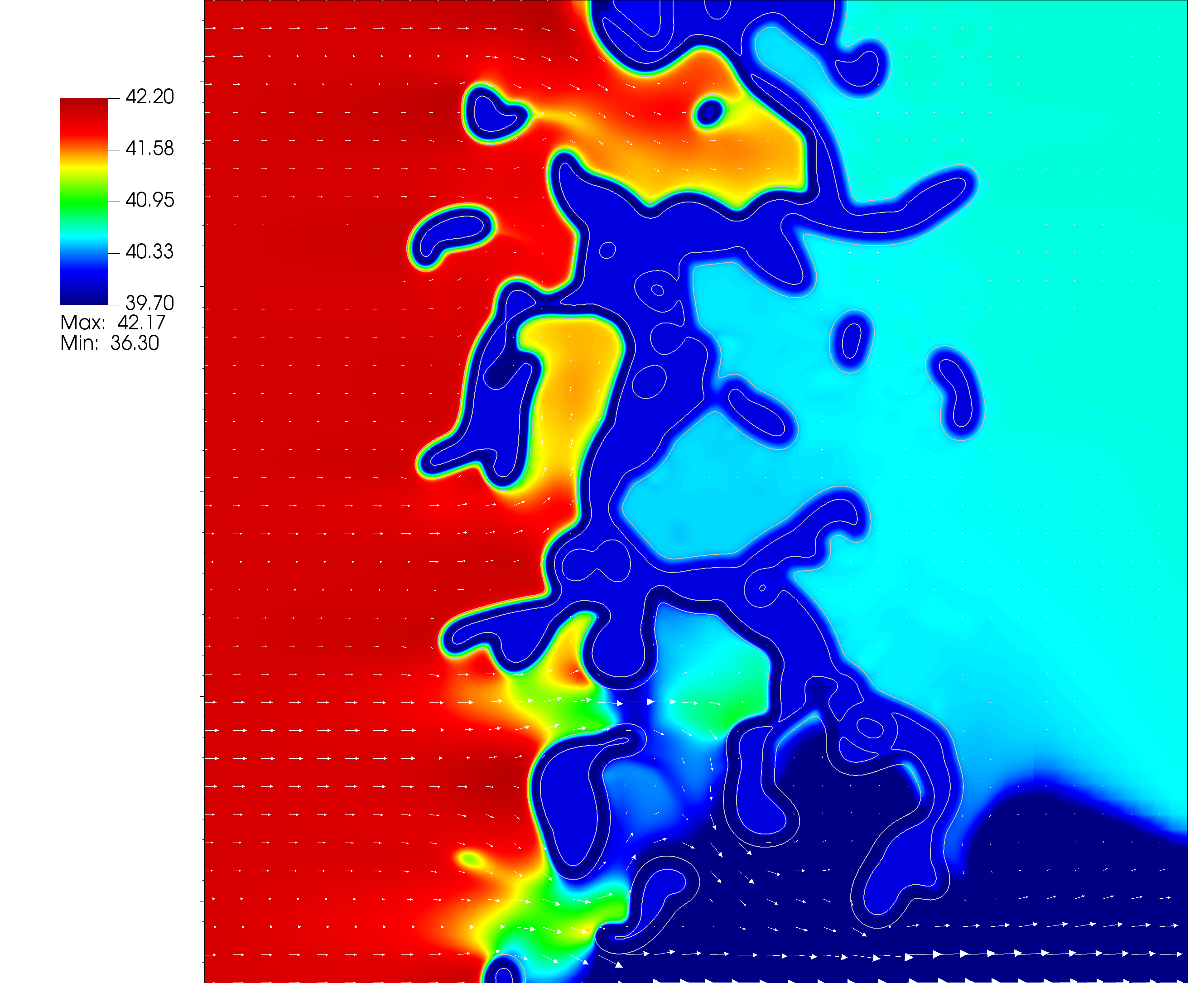}
  \end{minipage}%
  \begin{minipage}{0.55\linewidth}
    \includegraphics[height=3.4cm,clip,trim=15cm 40cm 0cm  0cm]{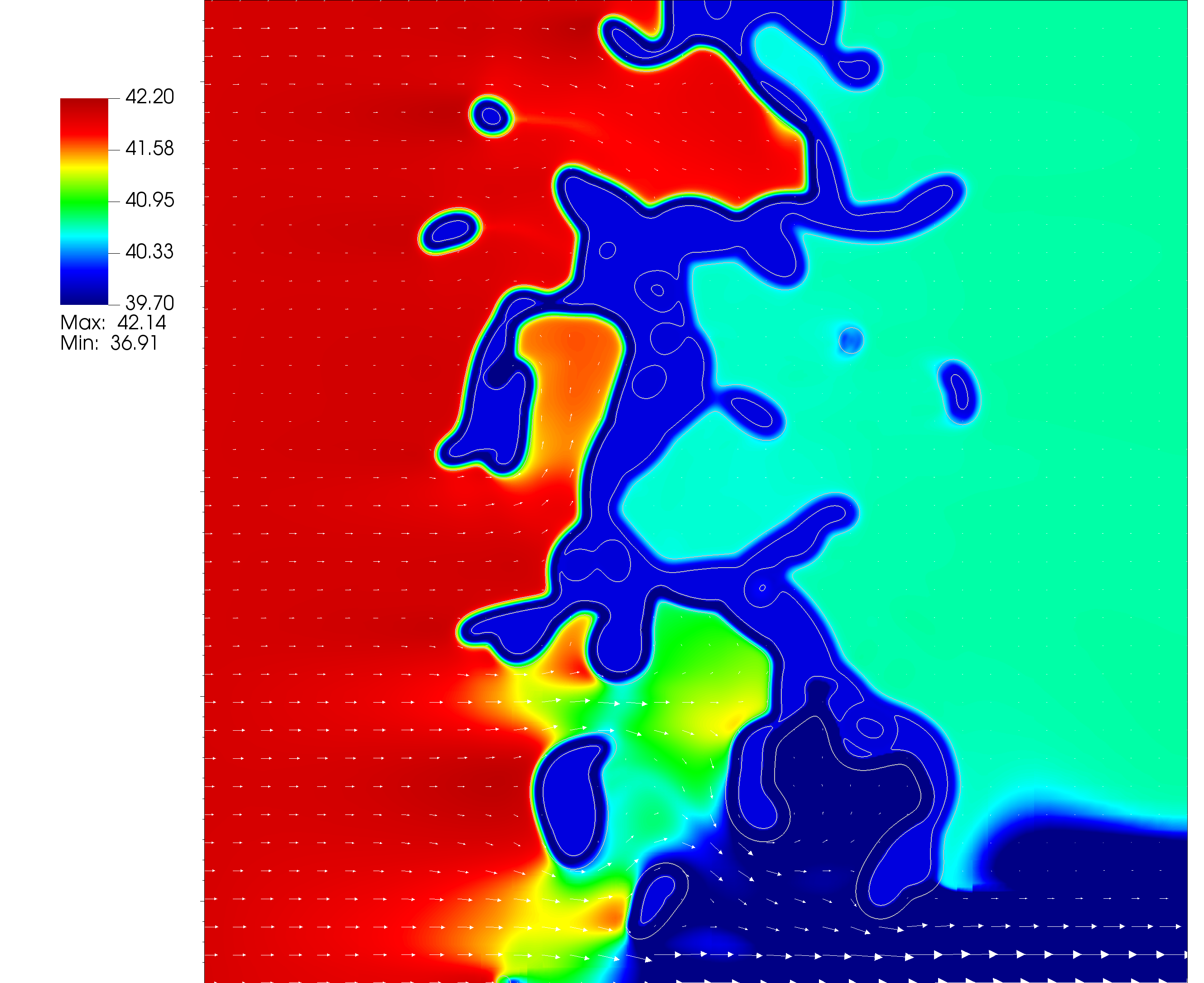}%
    \includegraphics[height=3.4cm,clip,trim=4cm  52.25cm 80cm 5cm]{results/AllenCahn/output.allencahn.27232/rho_vel0000.png}
    \includegraphics[height=3.4cm,clip,trim=15cm 40cm 0cm  0cm]{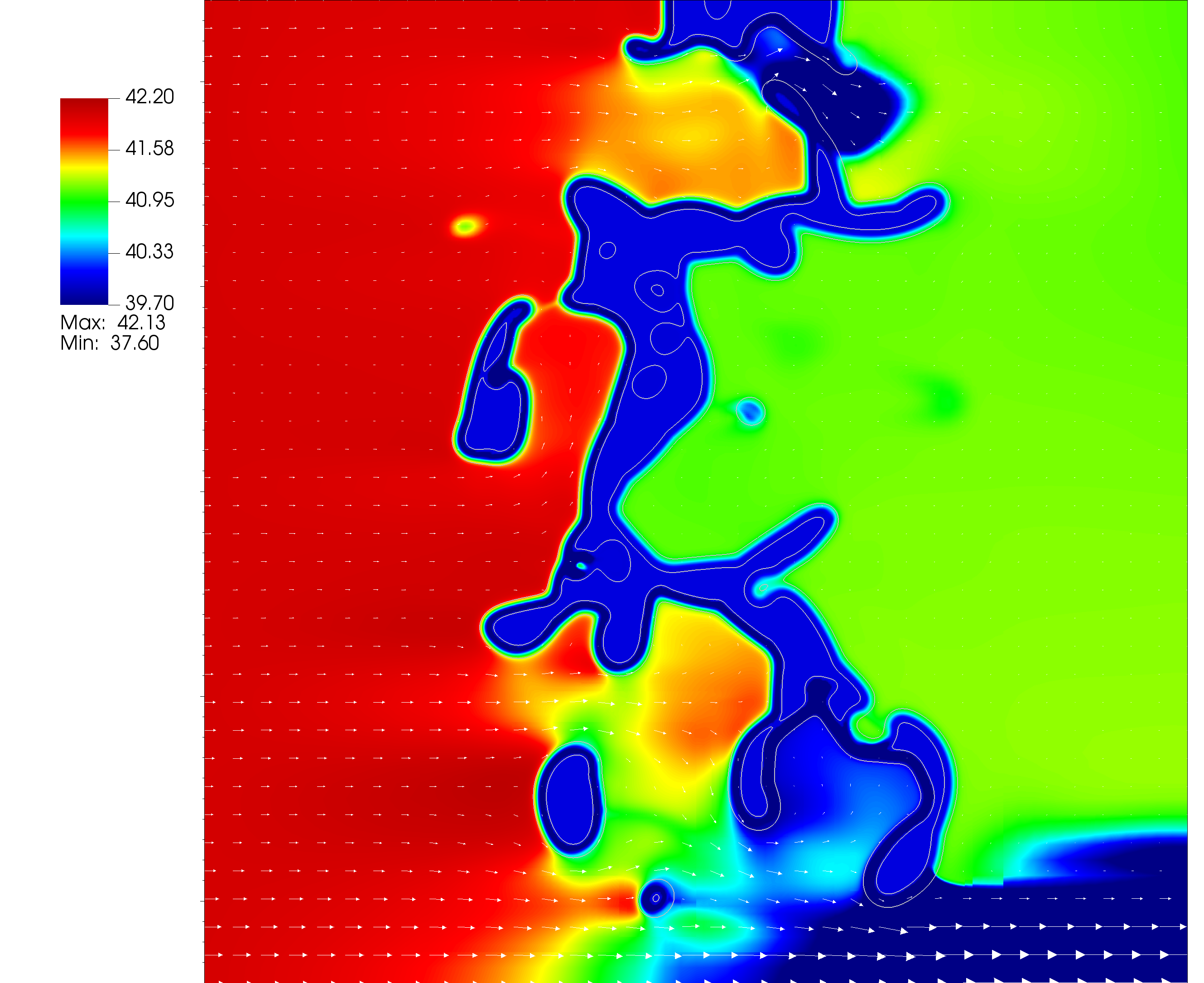}
    \includegraphics[height=3.4cm,clip,trim=15cm 40cm 0cm  0cm]{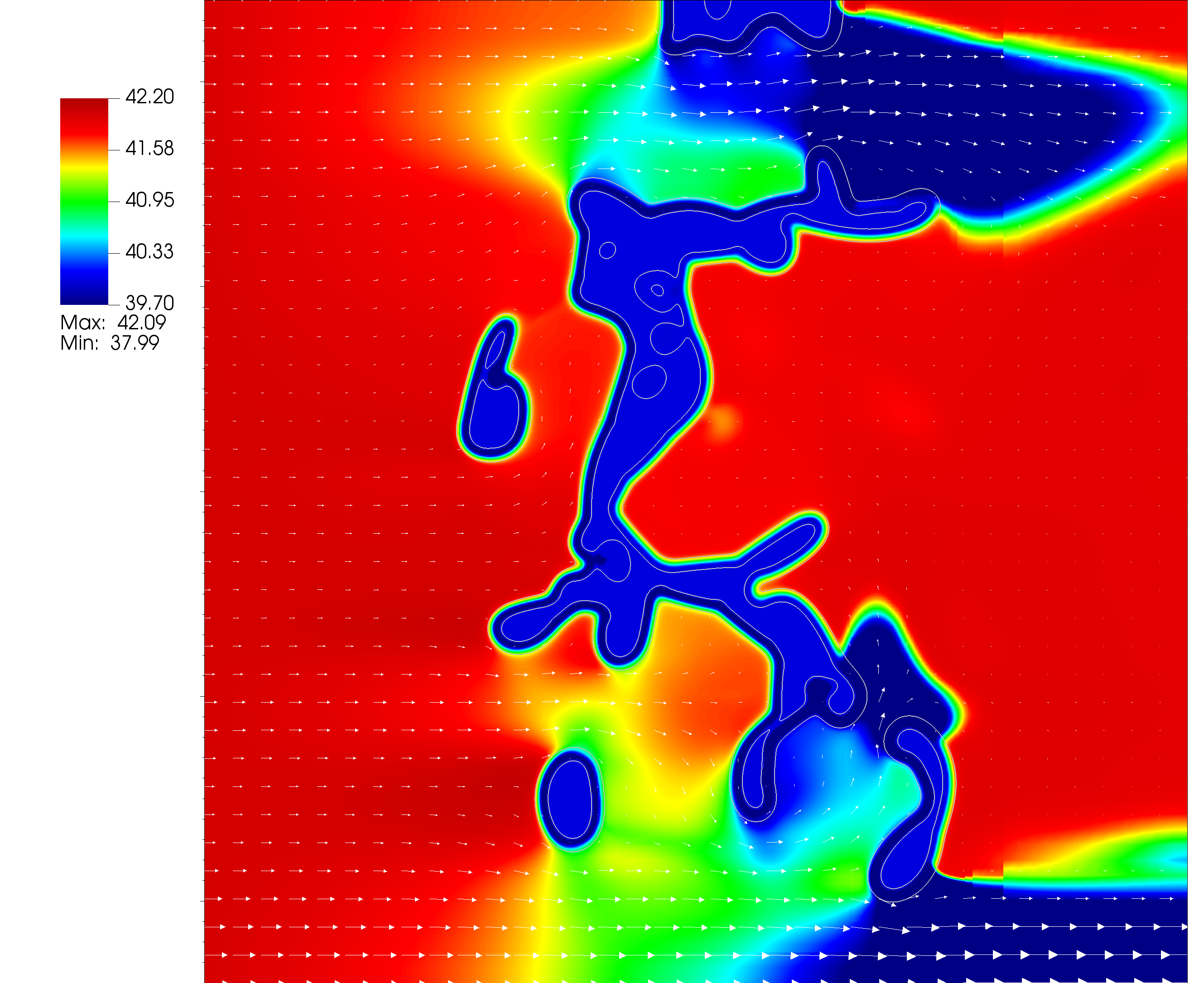}
    \includegraphics[height=3.4cm,clip,trim=15cm 40cm 0cm  0cm]{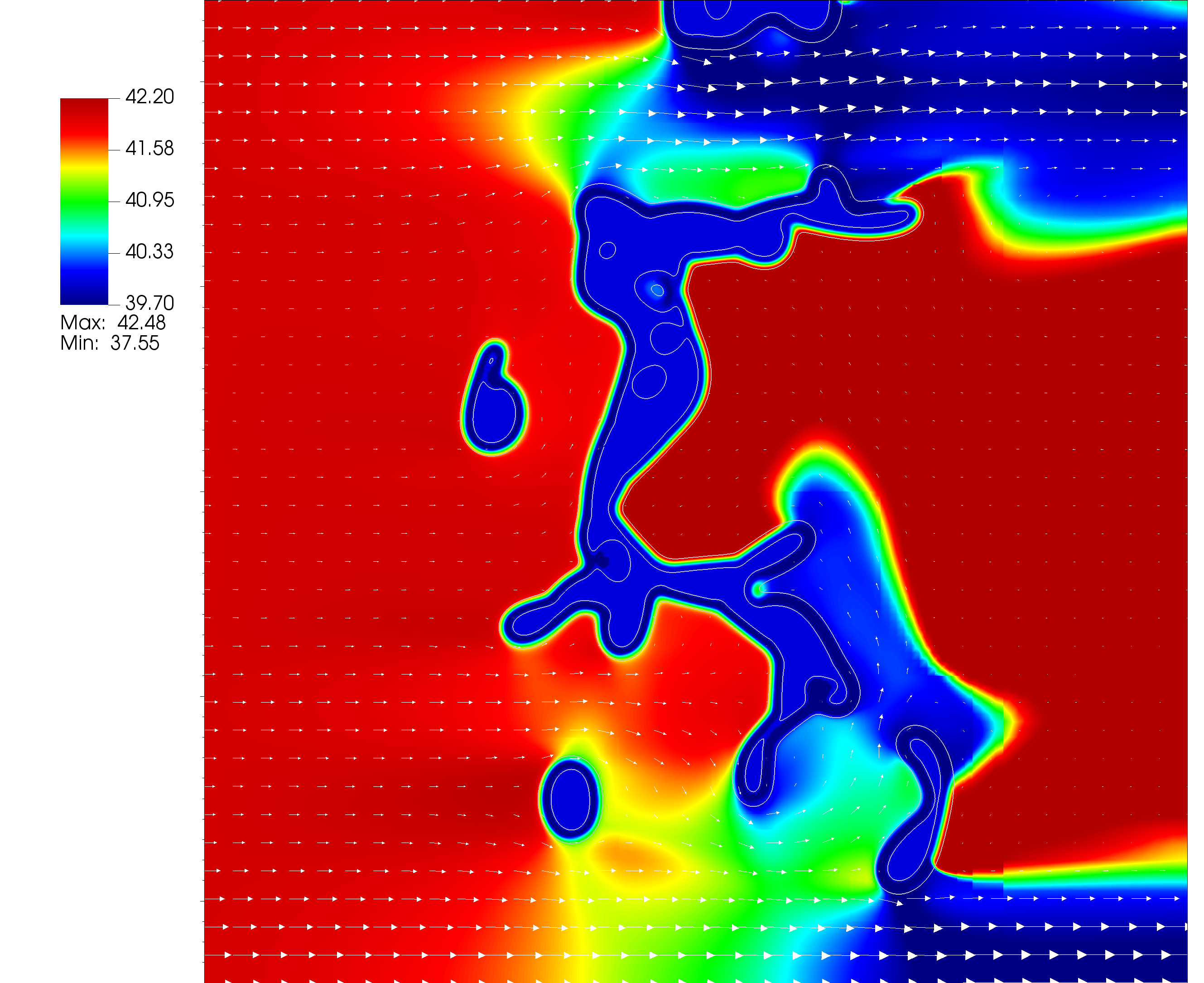}
    \includegraphics[height=3.4cm,clip,trim=15cm 40cm 0cm  0cm]{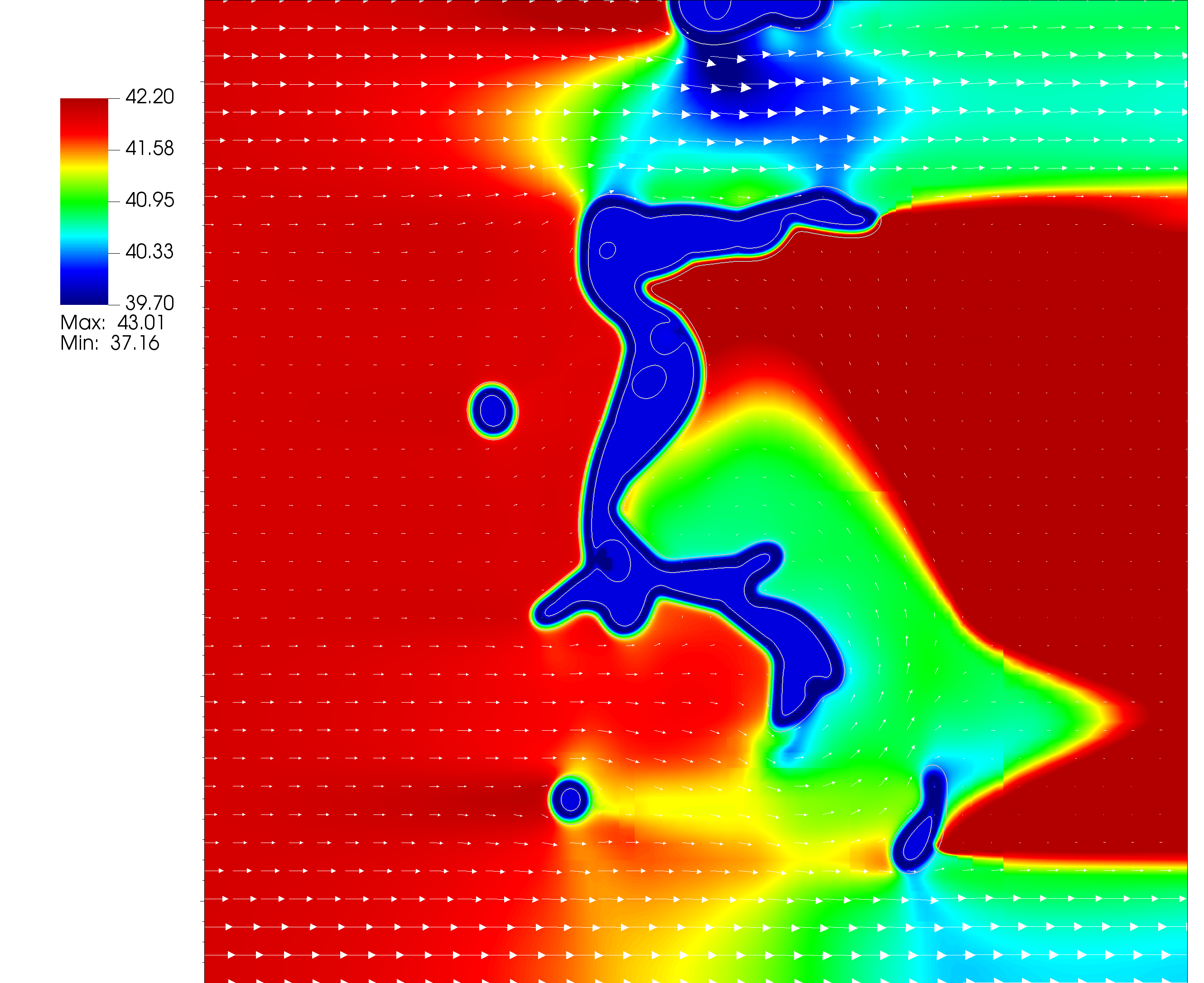}
  \end{minipage}
  \caption{Evolution of flow field: colors correspond to density; vector glyphs indicate direction of velocity, and contours indicate the location of the diffuse boundary. }
  \label{fig:allencahn-rho_vel}
\end{figure}

The final example of the proposed method considers the flow field, driven by a pressure gradient, through a semi-porous eroding solid material.
The applications of flow interaction through eroding\slash evolving media are numerous and range from fracking \cite{bavzant2014fracking} to coffee extraction \cite{matias2021flow}.
The behavior of the porous medium is dependent on the interaction of the flow with the highly complex, evolving structure of the medium.
Numerous such studies have been performed (for example, \cite{schneider2023fluid,chiu2020viscous}), where the evolution of the porous structure is modeled sometimes using a straightforward degradation model or using the phase field method \cite{ehlers2017phase}.
In all cases, the solution of the flow field requires fine detail of the porous media, and is vulnerable to instability in the event of topological transitions.

The scenario of flow driven through a porous media is an application of the present work.
This section considers a simplified Allen-Cahn phase field model for the medium, driven by the free energy functional
\begin{align}
  W = \int_\Omega \Big[\frac{\lambda}{\epsilon}\alpha^2(\alpha - 1)^2 +
  \frac{1}{2}\varepsilon g |\nabla\alpha|^2\Big]\,d\bm{x},
\end{align}
where $\alpha=1-\eta$ is the phase field order parameter, $\varepsilon$ is the diffuse interface width, $\lambda$ is a chemical potential, and $g$ is the value of the interfacial energy per unit area.
Here, $\varepsilon=0.1$, $g=0.01$, and $\lambda=0.01$.
The order parameter $\alpha$ is initialized to random values in $[0,1]$ over a band through the center of the domain, and initialized to $0$ near the inlet and exit.
Neumann boundary conditions are used for all domain boundaries.
The size of the domain is $48\times48$ in nondimensionalized units.
Before initializing the flow simulation, the phase field simulation is evolved according to the L2 gradient flow 
\begin{align}
  \frac{\partial\alpha}{\partial t} 
  =
  -L \frac{\delta W}{\delta \alpha}
  =
  - L \Big(\frac{\lambda}{\epsilon}(
  2\alpha  - 6\alpha^2 + 4\alpha^3
  ) +
  \epsilon\,\kappa\,\Delta \alpha \Big)
\end{align}
with $L=0.1$ and a timestep of $0.01$ until $t=10$ (enough to smooth out the features of the porous medium).
It is unphysical for the material to ``heal'', and so at $t=10$, the mobility parameter $L$ is set to
\begin{align}
  L = \begin{cases}0.01 & \delta W / \delta \alpha > 0 \\ 0 & \text{else}\end{cases},
\end{align}
which restricts all motion to erosion only.
(It is noted that one can easily increase the sophistication of the model using more advanced erosion models that depend on the local flow field.
Here, only the base case is considered as the focus is on the testing of the method, rather than the development of novel erosion models.)
This causes the order parameter to erode following curvature-driven flow (\cref{fig:allencahn-eta}).
The $\alpha$ field initially forms a continuous barrier to the flow, so that after initial equilibration, the pressure gradient is sustained across the boundary.
Eventually, holes are formed; some open up small sub-pockets of previously undisturbed fluid, others cause pathways to open up and allowing the eventual equilibration of the pressure.
The parameters have been selected to demonstrate behavior that is as complex as possible with a fairly substantial difference between timescales.

For the flow, the fluid parameters $\gamma=1.4$ and viscosity $\mu=0.01$ are used.
A small regularization value of $\zeta=1\times 10^{-8}$ was used and the reference pressure is $p_\mathrm{ref} = 10^5$.
The $\lambda$ coefficient for the additional non-penetration condition is set to $2\times 10^4$.
However, because the Allen Cahn equation does not necessarily result in compact support of the diffuse boundary, a fairly large cutoff value of $0.2$ was used to ensure that flow did not bleed through the narrow regions of the solid.
(Note that this could also be adjusted to reflect sub-grid porosity.)
Dirichlet boundary of $\rho=42.0$ is enforced for density and $E=1000$ for energy at the inlet (this drives the flow) and Neumann values are imposed at the outlet.
Momentum boundaries are Neumann at both the inlet and outlet.
At the top and bottom, density, energy, and x-momentum are Neumann; y-momentum is set to zero.
Density is initialized to 42.0, momentum is initialized to zero, and energy is initialized to $10^4$.
These values are also used for the solid density, momentum, and energy. 
Adaptive timestepping is used with a CFL set to 1.0, and the simulation is run until the entire solid region is consumed and the flow equilibrates.

The resulting flow is visualized by plotting the mixed density field $\bar{\rho}$ with isocontours at $\eta=0.1$, $0.5$, and $0.9$, and vector fields corresponding to $\eta\mathbf{v}$ (\cref{fig:allencahn-rho_vel}).
The flow initially interacts with the degrading solid region (frames 1--4) and semi-equilibrates.
Eventually, a hole appears in the solid region, causing pressure to increase at the outlet (frame 5), though the flow towards the top is still mostly equilibrated.
Finally, another break in the solid barrier appears towards the top (frame 6), causing an increased velocity and pressure decrease through the opening (frames 7--9) and near-equilibration of pressure (frame 10).

\begin{figure}
  \centering
    \includegraphics[height=3.6cm,clip,trim=15cm 0cm 0cm 0cm]{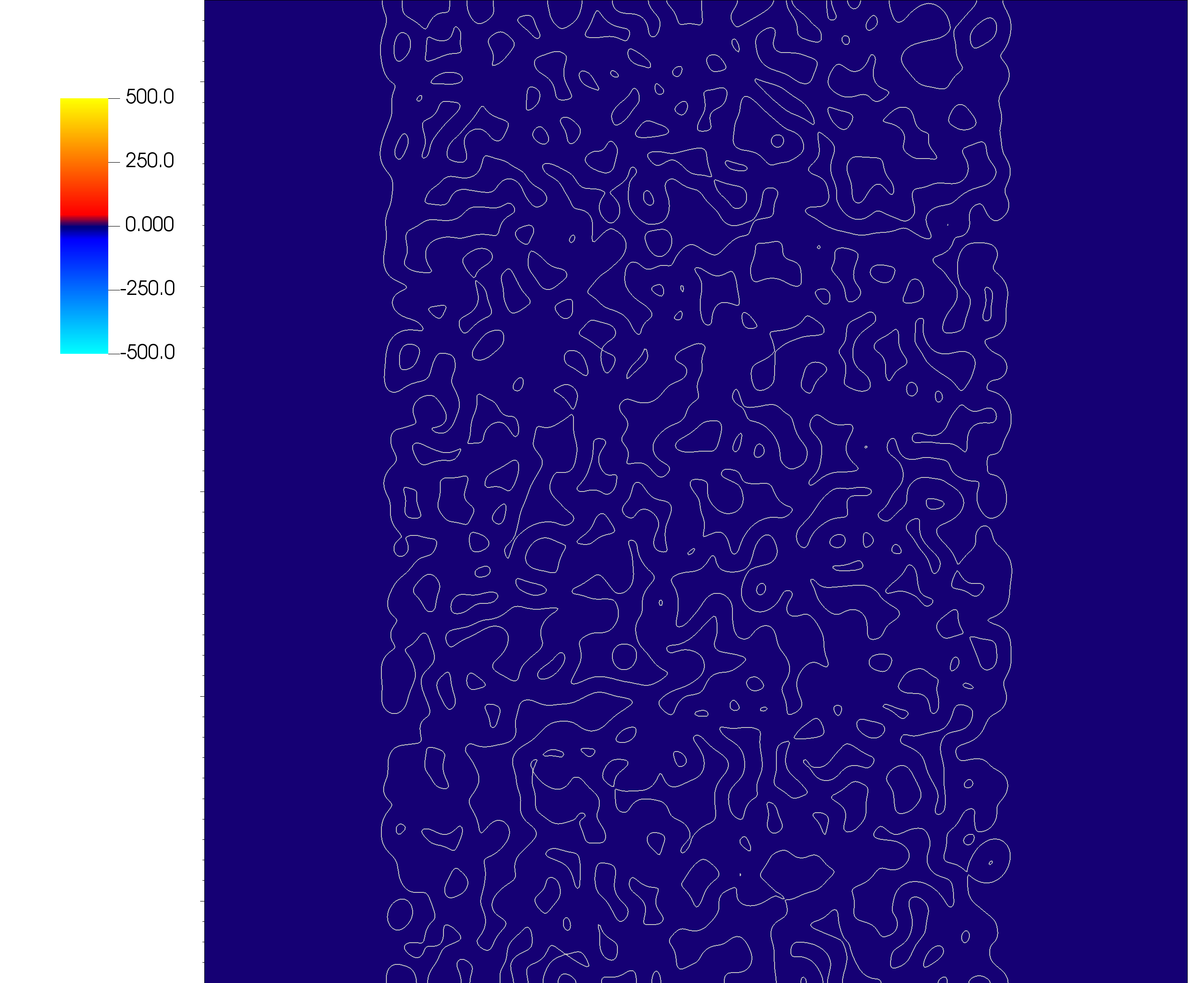}
    \includegraphics[height=3.6cm,clip,trim=15cm 0cm 0cm 0cm]{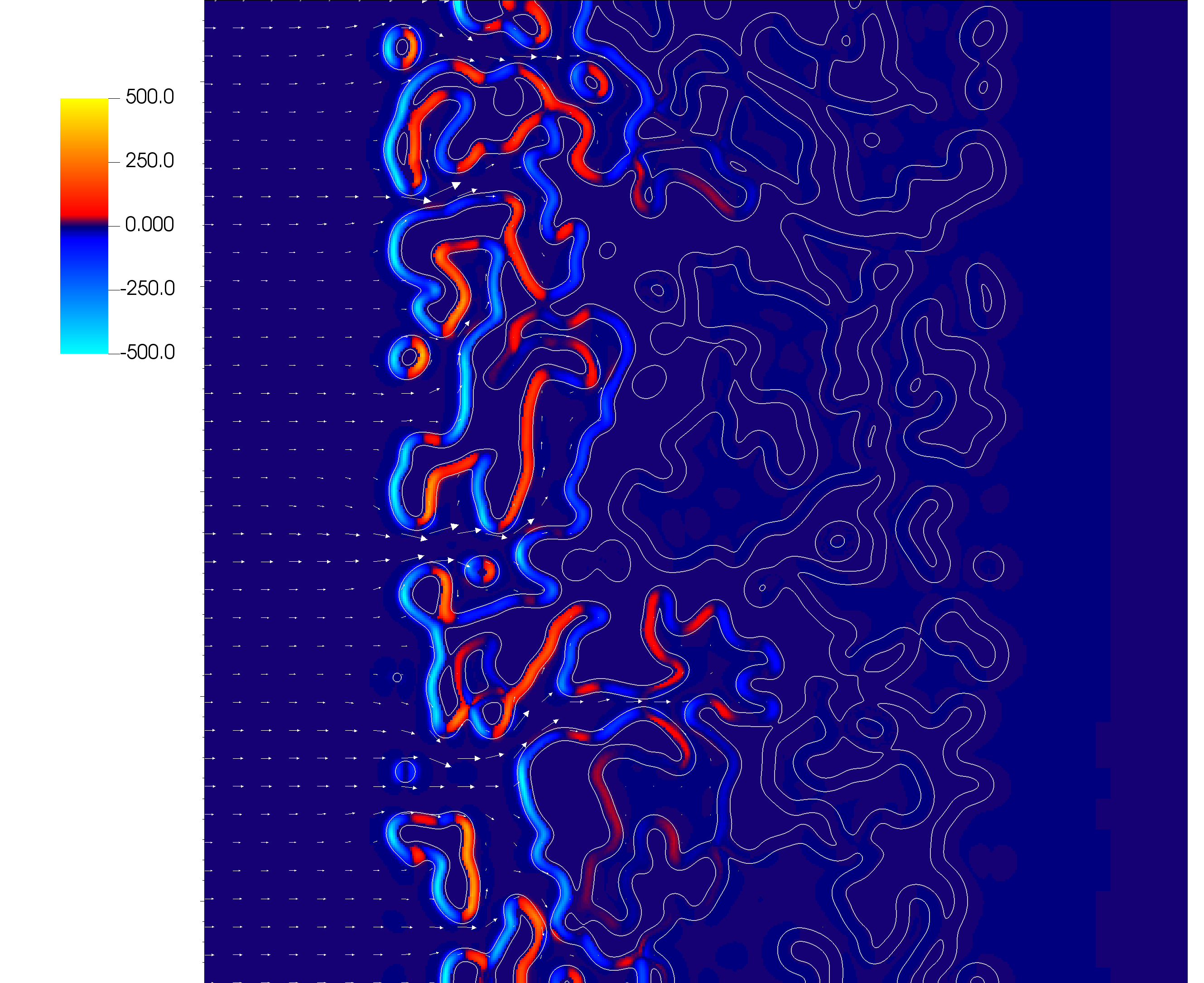}
    \includegraphics[height=3.6cm,clip,trim=15cm 0cm 0cm 0cm]{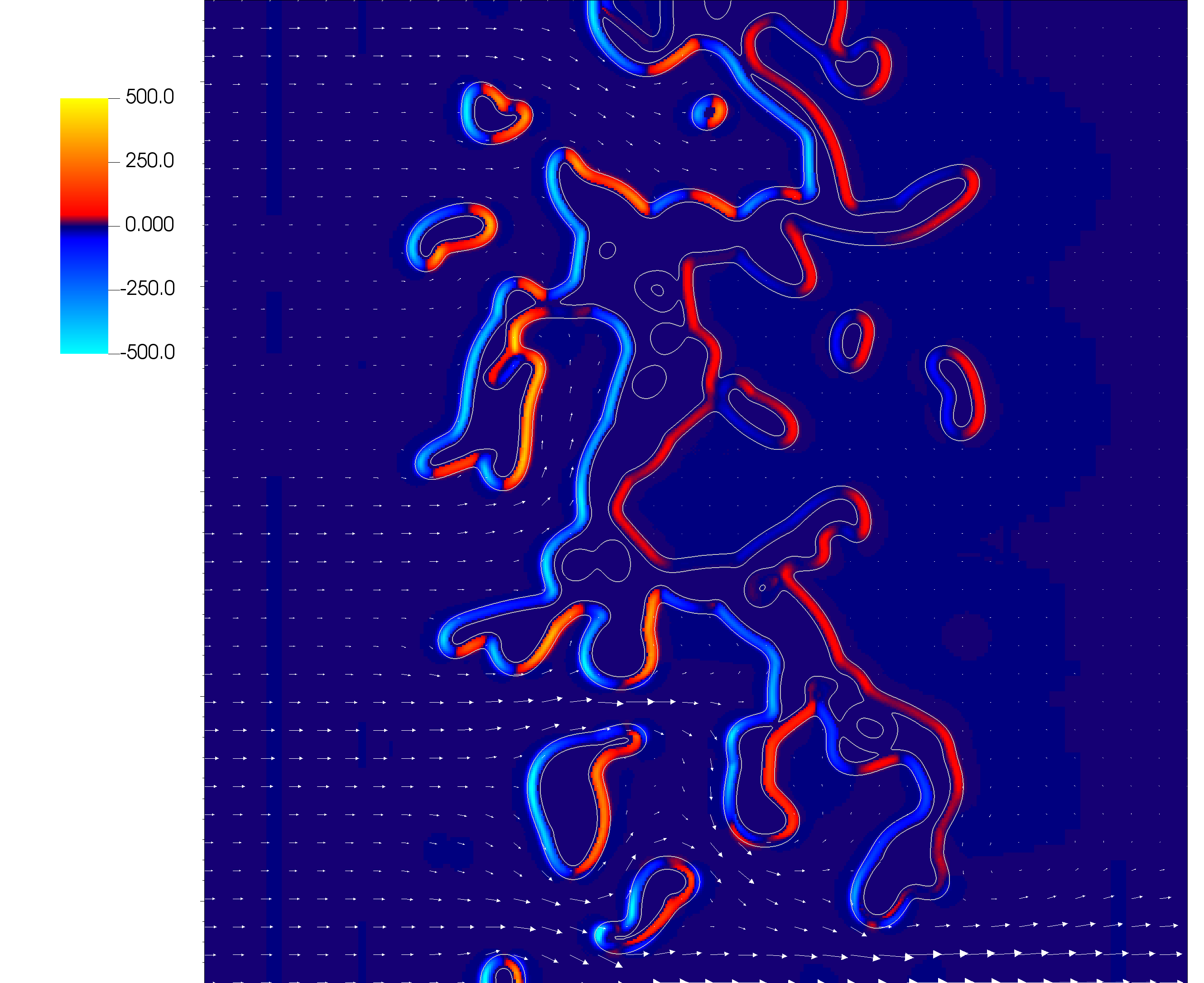}
    \includegraphics[height=3.6cm,clip,trim=15cm 0cm 0cm 0cm]{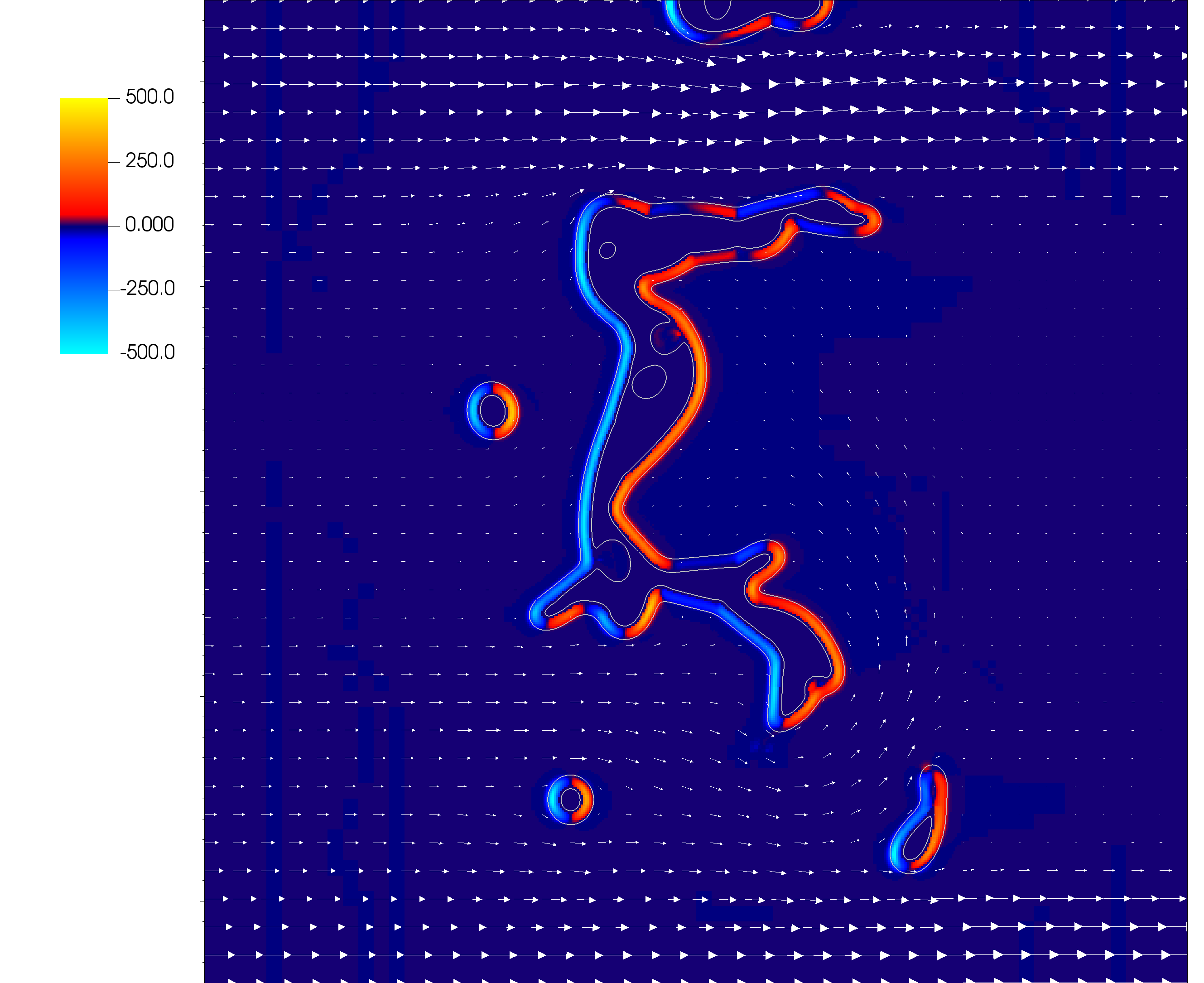}
  \caption{Horizontal component of the linear momentum source term $\dot{\mathbf{P}}$ (red is positive, blue is negative). }
  \label{fig:allencahn-pdot}
\end{figure}

The effect of the flow on the solid can be seen by plotting the linear momentum source term $\dot{\mathbf{P}}$, which, as was discussed previously, is exactly the normal force exerted by the flow on the solid region (\cref{fig:allencahn-pdot}).
Because the angular and linear momentum source terms are combined in this work, the viscous and mass-conserving effects are combined here.

\section{Conclusion}\label{sec:conclusion}

This work presents a systematic process for formulating arbitrary boundary conditions in implicitly defined domains with diffuse boundaries.
It is demonstrated that boundary conditions written in the form of mass\slash linear momentum\slash angular momentum\slash energy fluxes converge to the sharp interface equivalent in the limit as the diffuse boundary thickness goes to zero.
This is then specialized to a selection of boundary conditions of interest, which incidentally yields interesting insight into the role played by constitutive behavior in the prescription of essential boundary conditions (non-penetration is mediated by the equation of state; no-slip is mediated by viscosity).
The formulation is general and agnostic to the particular type of flow solver; however, there are many numerical considerations that must be accounted for.
In this work, the implementation is a compressible hydrodynamics solver with a basic Riemann Roe scheme; specialization to other types of implementations is left to future work.
The method and implementation is then applied to a selection of problems of interest, ranging from simple verification exercises to multiphysics problems driven by complex, evolving boundaries.
These demonstrate the accuracy and efficacy of the method.

Limitations of the work are acknowledged.
The flow solver that is used here is fairly primitive, and is not representative of modern hydrodynamic methods.
As discussed earlier (\cref{sec:numerical_considerations}), the method can be sensitive to particular types of discretization schemes, and it is essential that the diffuse source terms used are consistent.
This must be taken into account when extending this method to more advanced solvers.
It is also noted that, for simplicity, a very simple equation of state was used, which necessitated the integration of an additional ``wall'' EOS to sustain the non-penetration condition.
As described in the theory section of this paper, the EOS plays a critical role in mediating this boundary condition.
Finally, it is noted that all of the results are two-dimensional.
There are no inherent limitations of the method that restrict it to 2D, but the behavior of some source terms (particularly the angular momentum source) are considerably more complex in 3D.
Exploration of the diffuse boundary method extended to 3D is therefore left to future work.

\section*{Acknowledgements}
Authors EB, MM, and BR acknowledge support from the Office of Naval Research, USA, grant number N00014-21-1-2113.
BR also acknowledges support from ONR under grant number N00014-25-1-2029.
ES acknowledges support from the DOD SMART Scholarship Program and NAWCWD Fellowship.
This work used the INCLINE cluster at the University of Colorado Colorado Springs; INCLINE is supported by the National Science Foundation, grant \#2017917. 
The work also used the Nova cluster managed by the HPC group at Iowa State University. 
 
Finally, the authors wish to thank Dr.~Brian Bojko for his valuable insights throughout this work.

\printbibliography

\appendix
\setcounter{figure}{0}

\section{Detailed derivation of angular momentum conservation}\label{sec:angular_momentum_detailed_derivation}

This section provides some of the intermediate steps in the derivation of the angular momentum equations.
Writing in index notation,
\begin{align}
  \frac{d}{dt}\int_{B(t)} \epsilon_{ijk}x_j \rho u_k dV
  &=
    \int_{\partial B(t)} \epsilon_{ijk} x_j \sigma_{km} n_m dA.
    + \int_{\partial B(t)} \dot{\mathrm{L}}_{ij}n_{j}\,dV
\end{align}
Application of Reynolds transport theorem to the left hand side yields
\begin{align}
  \int_{B(t)}\Big( \frac{\partial }{\partial t} (\epsilon_{ijk}x_j \rho u_k)\Big)\, dV + \int_{\partial\Omega(t)} \epsilon_{ijk}x_j \rho u_k u_m n_m dA = \int_{\partial B(t)} \epsilon_{ijk} x_j \sigma_{km} n_m dA
  + \int_{\partial B(t)} \dot{\mathrm{L}}_{ij}n_j\,dV
\end{align}
Next, applying the divergence theorem to both sides allows for the combination, 
\begin{align}
  \int_{B(t)} \Big[\frac{\partial }{\partial t} (\epsilon_{ijk}x_j \rho u_k) + \frac{\partial}{\partial x_m}\left(\epsilon_{ijk}x_j \rho u_k u_m \right)\Big] \, dV
  &= \int_{B(t)}\Big[ \frac{\partial}{\partial x_m}\left(\epsilon_{ijk} x_j \sigma_{km} + \dot{\mathrm{L}}_{im}\right) \Big]\,dV.
\end{align}
Moving all terms to the left hand side gives
\begin{align}
    \int_{B(t)} \Big[\frac{\partial }{\partial t} (\epsilon_{ijk}x_j \rho u_k) + \frac{\partial}{\partial x_m}\left(\epsilon_{ijk}x_j \rho u_k u_m \right) - \frac{\partial}{\partial x_m}\left(\epsilon_{ijk} x_j \sigma_{km}\right)  - \dot{\mathrm{L}}_{ij,j}\Big] dV &= 0.
\end{align}
Applying the product rule to the time derivative and bringing the Levi-Civita tensor out of all derivatives,
\begin{align}
    \int_{B(t)}\Big[ \epsilon_{ijk} \frac{\partial x_j}{\partial t} ( \rho u_k) + \epsilon_{ijk} x_j \frac{\partial (\rho u_k)}{\partial t} + \epsilon_{ijk}\frac{\partial}{\partial x_m}\left(x_j \rho u_k u_m \right) - \epsilon_{ijk}\frac{\partial}{\partial x_m}\left( x_j \sigma_{km}\right) - \dot{\mathrm{L}}_{ij,j}\Big] dV &= 0
\end{align}
Then, recognizing that $\partial x_j/\partial t = u_j$ and noting that $\bm{u}\times\bm{u}=\bm{0}$,
\begin{align}
    \int_{B(t)} \Big[\epsilon_{ijk} x_j \frac{\partial (\rho u_k)}{\partial t} + \epsilon_{ijk}\frac{\partial}{\partial x_m}\left(x_j \rho u_k u_m \right) - \epsilon_{ijk}\frac{\partial}{\partial x_m}\left( x_j \sigma_{km}\right)  - \dot{\mathrm{L}}_{ij,j}\Big] dV &= 0.
\end{align}
Applying the product rule to the two divergence terms,
\begin{align}
  \int_{B(t)} \Big[
  &
  \epsilon_{ijk} x_j \frac{\partial (\rho u_k)}{\partial t}
  + \epsilon_{ijk}\cancelto{\delta_{jm}}{\frac{\partial x_j}{\partial x_m}}\rho u_k u_m
  + \epsilon_{ijk}x_j\frac{\partial}{\partial x_m}\left(\rho u_k u_m \right)
  - \epsilon_{ijk}\cancelto{\delta_{jm}}{\frac{\partial x_j}{\partial x_m}}\,\sigma_{km}
  - \epsilon_{ijk} x_j\frac{\partial \sigma_{km}}{\partial x_m}
  - \dot{\mathrm{L}}_{ij,j}\Big] dV = 0
\end{align}
simplifying further:
\begin{align}
    \int_{B} \Big[\epsilon_{ijk} x_j \frac{\partial (\rho u_k)}{\partial t} + \cancelto{0}{\epsilon_{ijk}\rho u_k u_j}  + \epsilon_{ijk}x_j\frac{\partial}{\partial x_m}\left(\rho u_k u_m \right) - \epsilon_{ijk}\sigma_{kj} - \epsilon_{ijk} x_j\frac{\partial \sigma_{km}}{\partial x_m}  - \dot{\mathrm{L}}_{ij,j}\Big] dV &= 0.
\end{align}
Temporary switching to invariant notation, re-write in a more suggestive form, assuming basis vectors $\{\bm{g}_i\}$,
\begin{align}
  \int_B\bm{x}\times\Big[\frac{\partial\rho\bm{u}}{\partial t} + \operatorname{div}\big(\rho\bm{u}\otimes\bm{u}-\bm{\sigma}\big)\Big]\,dV
  = \int_B\,\Big[\epsilon_{ijk}\sigma_{jk}\bm{g}_i + L_{ij,j}\Big]\,dV,
\end{align}

\end{document}